\documentclass[
aps,
twocolumn,
reprint,
amsmath,
amssymb,
superscriptaddress,
longbibliography]{revtex4-2}
\usepackage{amssymb}
\usepackage{amsmath}
\usepackage{amsthm}
\usepackage[english]{babel}
\usepackage{bm}
\usepackage{bbold}
\usepackage{hyperref}
\usepackage{graphicx}
\usepackage{subfigure}
\usepackage{xcolor}
\usepackage{tabularx}
\usepackage{booktabs}
\newcommand{\ket}[1]{|#1\rangle}
\newcommand{\bra}[1]{\langle #1|}
\newcommand{\inp}[2]{\langle #1|#2\rangle}

\newcommand{\Tr}{\mathrm{Tr}}
\def\<{\langle}  
\def\>{\rangle}  

\newtheorem{lemma}{Lemma}

\hypersetup{
	colorlinks  = true,
	urlcolor    = black,
	citecolor   = blue,
	linkcolor   = blue
}

\begin{document}
	
	\title{Simultaneous Measurement of Multiple Incompatible Observables and Tradeoff in Multiparameter Quantum Estimation}

\author{Hongzhen Chen}
\email{hzchen@szu.edu.cn}
\affiliation{Institute of Quantum Precision Measurement, State Key Laboratory of Radio Frequency Heterogeneous Integration, College of Physics and Optoelectronic Engineering, Shenzhen University, Shenzhen, China}
\affiliation{Quantum Science Center of Guangdong-Hong Kong-Macao Greater Bay Area (Guangdong), Shenzhen, China}
\affiliation{Department of Mechanical and Automation Engineering, The Chinese University of Hong Kong, Shatin, Hong Kong SAR, China}

\author{Lingna Wang}
\email{lnwang@mae.cuhk.edu.hk}
\affiliation{Department of Mechanical and Automation Engineering, The Chinese University of Hong Kong, Shatin, Hong Kong SAR, China}

\author{Haidong Yuan}
\email{hdyuan@mae.cuhk.edu.hk}
\affiliation{Department of Mechanical and Automation Engineering, The Chinese University of Hong Kong, Shatin, Hong Kong SAR, China}





\begin{abstract}
\noindent\textbf{Abstract}\\
\noindent
	How well can multiple incompatible observables be implemented by a single measurement? This is a fundamental problem in quantum mechanics with wide implications for the performance optimization of numerous tasks in quantum information science. While existing studies have been mostly focusing on the approximation of two observables with a single measurement, in practice multiple observables are often encountered, for which the errors of the approximations are little understood. Here we provide a framework to study the implementation of an arbitrary finite number of observables with a single measurement. Our methodology yields novel analytical bounds on the errors of these implementations, significantly advancing our understanding of this fundamental problem. Additionally, we introduce a more stringent bound utilizing semi-definite programming that, in the context of two observables, generates an analytical bound tighter than previously known bounds. The derived bounds have direct applications in assessing the trade-off between the precision of estimating multiple parameters in quantum metrology, an area with crucial theoretical and practical implications. To validate the validity of our findings, we conducted experimental verification using a superconducting quantum processor. This experimental validation not only confirms the theoretical results but also effectively bridges the gap between the derived bounds and empirical data obtained from real-world experiments. Our work paves the way for optimizing various tasks in quantum information science that involve multiple noncommutative observables.
 
\end{abstract}



\maketitle

\noindent\textbf{Introduction}\\
\noindent
One of the distinctive features of quantum mechanics is its noncommutativity, setting it apart from classical physics. 
This noncommutativity is prominently manifested in the properties of observables, leading to phenomena that defy classical expectations. 
When multiple observables commute with each other, their simultaneous measurement is feasible through projective measurements on their shared eigenspaces. However, for noncommuting observables, exact simultaneous measurement becomes unattainable, necessitating approximation. 
A central issue in understanding and harnessing the full potential of quantum systems is then determining the degree to which a single measurement can accurately capture multiple non-commuting observables. 
Intuitively, measuring multiple noncommuting observables comes with an inherent tradeoff.
As we strive to estimate one observable with higher accuracy, the imprecision in determining other incompatible observables tends to increase. 
This tradeoff is deeply rooted in the uncertainty principle \cite{Heisenberg27}, which stands as a fundamental principle in quantum mechanics.

The standard Robertson-Schrödinger uncertainty relation \cite{Robe29,Schrodinger:1930ty}, $(\Delta X_1)^2(\Delta X_2)^2 \geq \frac{1}{4}|\Tr(\rho[X_1,X_2])|^2$ with $(\Delta X_{1/2})^2=\langle X_{1/2}^2\rangle-\langle X_{1/2}\rangle^2$, describes the impossibility of preparing quantum states with sharp distributions for non-commuting observables simultaneously, which is also referred to as the preparation uncertainty relation. 
The uncertainty relations that describe the approximation of non-commuting observables via a single measurement are called measurement uncertainty relations \cite{arthurs1965,Arthurs1988,Ozawa2003,OZAWA2004367,OZAWA2004350,OZAWA200321,Ozawa_2014,Hall2004,Branciard2013,Branciard2014,Lu2014,2014Error,Francesco2014,RevModPhys.86.1261,BUSCH2007155}. 
In the field of measurement uncertainty relations, there are two main approaches: the state-independent approach \cite{BUSCH2007155,RevModPhys.86.1261,Qin2019,Busch2013,Ma2016,Mao2022}, which establish error bounds for measurement devices regardless of the input state; and state-dependent relations \cite{arthurs1965,Arthurs1988,Ozawa2003,OZAWA2004367,OZAWA2004350,OZAWA200321,Ozawa_2014,Hall2004,Branciard2013,Branciard2014,2014Error,Lu2014}, which assess the trade-off between errors in joint measurements on a predetermined quantum state. In this article, we focus on state-dependent measurement uncertainty relations as they are particularly relevant to multi-parameter quantum estimation where the state is often constrained.

Current research in state-dependent measurement uncertainty relations predominantly concentrate on the error-tradeoff relations for approximating pairs of observables \cite{Ozawa2003,OZAWA2004367,OZAWA2004350,OZAWA200321,Ozawa_2014,Hall2004,Branciard2013,Branciard2014,2014Error,Lu2014}, yet practical applications typically require dealing with multiple observables.
Fields such as vector magnetometry \cite{HouMinimal2020,meng2023sep}, reference frame alignment \cite{PhysRevA.64.050302} and quantum imaging \cite{taylor2016quantum,ALBARELLI2020126311} all require a nuanced understanding and manipulation of three or more observables simultaneously.
The inherent uncertainty therein cannot be fully understood or quantified through the lens of pairwise uncertainty relations. 
This gap highlights a critical need for expanding the framework to encompass multiple observables, which is of both fundamental and practical importance.
However, similar to numerous challenges in quantum information science, akin to the complex quantification of multipartite entanglement \cite{Horodecki2009}, extending results from two-party scenarios to multi-partite ones often necessitates novel methodologies.

In this article, we present dual methods capable of yielding error-tradeoff relations for approximating an arbitrary number of observables. 
The first approach delivers analytical tradeoff relations for any number of observables, while the second method offers more stringent bounds through semidefinite programming. 
By merging these strategies, we are able to derive analytical tradeoff relations that are even tighter than any existing trade-off relations for two observables. Subsequently, we apply these approaches to quantum metrology, deriving tighter tradeoff relations for estimating an arbitrary number of parameters---a topic central to contemporary quantum estimation research. Furthermore, we empirically validate our findings through experimentation conducted on a superconducting quantum processor.
Our results provide significant insights into the interplay among multiple observables involved in various quantum information tasks, notably in the calibration of performance for multiparameter quantum metrology.
\bigskip
\noindent\textbf{Results}\\
\noindent\textbf{Analytical error-tradeoff relation}

\noindent
We commence by deriving an analytical measurement uncertainty relation for a general set of $n$ observables. The objective is to use a single Positive Operator-Valued Measurement (POVM), denoted  $\mathcal{M} = \{M_m\}$, to approximate the given $n$ observables ${X_1, X_2, \ldots, X_n}$ when applied to a quantum state $\rho$ and to determine relations that set limits on the minimum cumulative weighted approximation error. 

According to Neumark's dilation theorem \cite{NielC00book}, the POVM, $\mathcal{M}=\{M_m=K_m^\dagger K_m\}$, is equivalent to a projective measurement on $\rho\otimes\sigma$ in an extended Hilbert space $\mathcal{H_S}\otimes\mathcal{H_A}$, here $\sigma=\ket{\xi_0}\bra{\xi_0}$ is an ancillary state such that $( I\otimes\bra{\xi_0})U^\dagger(I\otimes \ket{\xi_m}\bra{\xi_m})U(  I\otimes\ket{\xi_0})=M_m$, where $\{|\xi_m\rangle\}$ is an orthonormal basis for the ancillary system, $U$ is a unitary operator on the extended space such that for any $|\psi\rangle$, $U|\psi\rangle\otimes \ket{\xi_0}=\sum_m K_m |\psi\rangle |\xi_m\rangle$. Denote $V_m=U^\dagger(I\otimes \ket{\xi_m}\bra{\xi_m})U$, we then have  $\Tr[(\rho\otimes\ket{\xi_0}\bra{\xi_0})V_m]=\Tr(\rho M_m)$. From the measurement, we can construct a set of commuting observables, $\{F_j=\sum_m f_j(m)V_m| 1\leq j\leq n\}$, to approximate $\{X_j\otimes I\}$ in the extended Hilbert space (see Fig.\ref{fig:setting}). The mean squared error of the approximation on the state is given by \cite{Ozawa2003,OZAWA2004367,Branciard2013}
\begin{equation}
\epsilon_j^2 = \Tr\left[(F_j-X_j\otimes I)^2\left(\rho\otimes\sigma\right)\right].
\end{equation}

\begin{figure}
    \centering
    \includegraphics[width=0.48\textwidth]{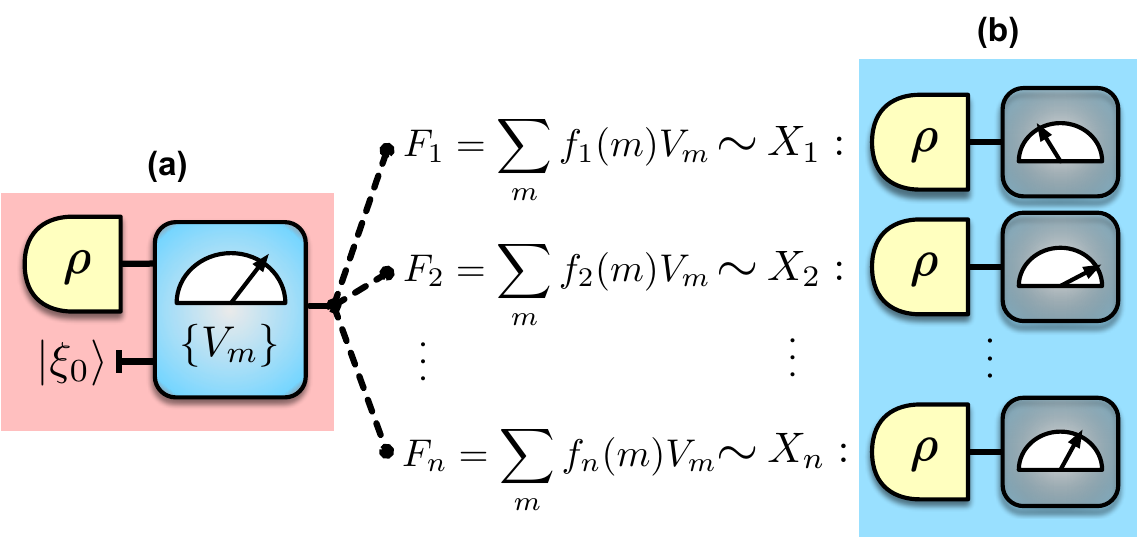}
    \caption{Simultaneous measurement of multiple observables $\{X_1,X_2,...,X_n\}$ via a single measurement. As illustrated in Fig.(a), a Positive Operator-Valued Measure (POVM) applied to $\rho$ effectively acts as a projective measurement $\{V_m\}$ on the extended state $\rho\otimes\ket{\xi_0}\bra{\xi_0}$. This setup enables the construction of a set of commuting observables $\{F_1, F_2, \ldots, F_n\}$, as depicted in Fig. (b), designed to approximate the individual measurements of each observable in the set $\{X_1, X_2, \ldots, X_n\}$.}
    \label{fig:setting}
\end{figure}

In the case of two observables, Ozawa obtained an error-tradeoff relation as \cite{Ozawa2003,OZAWA2004367}
\begin{equation}\label{eq:ozawa}
\epsilon_1\cdot\epsilon_2+\epsilon_1\cdot\Delta X_2 + \Delta X_1\cdot\epsilon_2\geq c_{12},
\end{equation}
here $\Delta X=\sqrt{\langle X^2\rangle-\langle X\rangle^2}$ is the standard deviation of an observable on the state with $\langle \,\cdot\,\rangle=\Tr(\rho \,\cdot\,)$, $c_{12}=\frac{1}{2}\Tr(\rho[X_1,X_2])$. Branciard strengthened this relation as \cite{Branciard2013,Branciard2014}
\begin{equation}\label{eq:branc}
\begin{aligned}
	&\epsilon_1^2\cdot(\Delta X_2)^2+(\Delta X_1)^2\cdot\epsilon_2^2\\
	&+2\sqrt{(\Delta X_1)^2\cdot(\Delta X_2)^2-c_{12}^2}\epsilon_1\cdot \epsilon_2 \geq c_{12}^2,
\end{aligned}
\end{equation}
which is tight for pure states. For mixed states, Ozawa further tightened the relation by replacing $c_{12}$ with $\frac{1}{2}\|\sqrt{\rho}[X_1,X_2]\sqrt{\rho}\|_1$ \cite{2014Error}. However, it is worth noting that even with this improvement, for mixed states the bound is not tight, and the geometrical method employed to derive these relations are not readily extendable to scenarios involving more than two observables. For general $n$ observables, the error-tradeoff relation is little understood.

We now present an approach that can lead to analytical tradeoff relations for an arbitrary number of observables. Let
\begin{equation}
\begin{aligned}
	\mathcal{A}_u=&\begin{pmatrix}
		\bra{u}\sqrt{\rho\otimes\sigma}E_1 \\ \vdots \\ \bra{u}\sqrt{\rho\otimes\sigma}E_n \\  \bra{u}\sqrt{\rho\otimes\sigma}(X_1\otimes I) \\
		\vdots \\ \bra{u}\sqrt{\rho\otimes\sigma}(X_n\otimes I)
	\end{pmatrix}\begin{pmatrix}
		\bra{u}\sqrt{\rho\otimes\sigma}E_1 \\ \vdots \\ \bra{u}\sqrt{\rho\otimes\sigma}E_n \\  \bra{u}\sqrt{\rho\otimes\sigma}(X_1\otimes I) \\
		\vdots \\ \bra{u}\sqrt{\rho\otimes\sigma}(X_n\otimes I)
	\end{pmatrix}^\dagger\\
	=&\begin{pmatrix}
		Q_u & R_u \\
		R_u^\dagger & S_u
	\end{pmatrix}
	\geq 0,
\end{aligned}
\end{equation}
here $E_j=F_j-X_j\otimes I$ is the error operator, $|u\rangle$ is any vector, $Q_u$, $R_u$, $S_u$ are $n\times n$ matrices with the entries given by (note $E_j$ and $X_j$ are all Hermitian operators)
\begin{equation}
\begin{aligned}
	(Q_u)_{jk}&=\bra{u}\sqrt{\rho\otimes\sigma}E_jE_k\sqrt{\rho\otimes\sigma}\ket{u},\\
	(R_u)_{jk}&=\bra{u}\sqrt{\rho\otimes\sigma}E_j(X_k\otimes I)\sqrt{\rho\otimes\sigma}\ket{u},\\
	(S_u)_{jk}&=\bra{u}\sqrt{\rho\otimes\sigma}(X_j\otimes I)(X_k\otimes I)\sqrt{\rho\otimes\sigma}\ket{u}.
\end{aligned}
\end{equation}
We can write these matrices in terms of the real and imaginary parts as $Q_u=Q_{u,\text{Re}}+iQ_{u,\text{Im}}$, $R_u=R_{u,\text{Re}}+iR_{u,\text{Im}}$, $S_u=S_{u,\text{Re}}+iS_{u,\text{Im}}$.

Given any set of states $\{\ket{u_q}\}$ such that $\sum_q\ket{u_q}\bra{u_q}= I$, we can derive a corresponding set of matrices $\{\mathcal{A}_{u_q}\}$. We then construct a matrix $\tilde{{\mathcal{A}}}$ as the sum of these matrices with each $\tilde{\mathcal{A}}_{u_q}$ being either $\mathcal{A}_{u_q}$ or its transpose, $\mathcal{A}_{u_q}^{T}$. Since both $\mathcal{A}_{u_q}$ and $\mathcal{A}_{u_q}^{T}$ are positive semi-definite, it follows that:
\begin{equation}\label{eq:Auq}
\tilde{\mathcal{A}}=\sum_q\tilde{\mathcal{A}}_{u_q}=
\begin{pmatrix}
	\tilde{Q} & \tilde{R}\\
	\tilde{R}^\dagger & \tilde{S}
\end{pmatrix}\geq 0,
\end{equation}
where the components are defined as $\tilde{Q} = \sum_q \tilde{Q}_{u_q}$, $\tilde{R} = \sum_q \tilde{R}_{u_q}$, and $\tilde{S} = \sum_q \tilde{S}_{u_q}$, with every $(\tilde{Q}_{u_q}, \tilde{R}_{u_q}, \tilde{S}_{u_q})$ being either $(Q_{u_q}, R_{u_q}, S_{u_q})$ or their complex conjugate  $(\bar{Q}_{u_q}, \bar{R}_{u_q}, \bar{S}_{u_q})$, here $\bar{M}=M_{\text{Re}}-iM_{\text{Im}}$ and for Hermitian matrix $\bar{M}=M^T$. It's important to highlight that the real parts of the matrix elements in $\tilde{\mathcal{A}}$ are unaffected by the choice between $\mathcal{A}_{u_q}$ and its transpose. Consequently, the real parts of $\tilde{Q}$ and $\tilde{S}$ remain constant and are determined solely by the original components without regard to whether they were chosen as $Q_{u_q}$ or $\bar{Q}_{u_q}$ (and correspondingly for $S_{u_q}$), with their specific values given by
\begin{equation}
\begin{aligned}
(\tilde{Q}_{\text{Re}})_{jk}&=\sum_q(\tilde{Q}_{u_q,\text{Re}})_{jk}
	=\frac{1}{2}\Tr\left[(\rho\otimes\sigma)\{E_j,E_k\}\right],\\
(\tilde{S}_{\text{Re}})_{jk}&=\sum_q(\tilde{S}_{u_q,\text{Re}})_{jk}
	=\frac{1}{2}\Tr(\rho\{X_j,X_k\}).
\end{aligned}
\end{equation}
Specifically, the diagonal elements of $\tilde{Q}$ and $\tilde{S}$ are given as $(\tilde{Q})_{jj} = \epsilon_j^2$ and $(\tilde{S})_{jj} = \Tr(\rho X_j^2)$, respectively.

From Eq.(\ref{eq:Auq}), we can derive an analytical error-tradeoff relation for approximating $n$ observables (see Sec. S1 of the Supplementary Material for a detailed derivation):
\begin{equation}\label{eq:multi_obs_boundmain}
\Tr(S_{\text{Re}}^{-1}Q_{\text{Re}}) \geq \left(\sqrt{\|S_{\text{Re}}^{-\frac{1}{2}}\tilde{S}_{\text{Im}}S_{\text{Re}}^{-\frac{1}{2}}\|_{F}+1}-1\right)^2,
\end{equation}
where $\|\cdot\|_F = \sqrt{\sum_{j,k} |(\cdot)_{jk}|^2}$ represents the Frobenius norm. In this inequality, the term $Q_{\text{Re}}$ is the sole quantity dependent on the measurement strategy and its diagonal entries correspond to the mean-square errors of the approximation. Both $S_{\text{Re}}$ and $\tilde{S}_{\text{Im}}$ are independent of the specific measurement process; instead, they are entirely determined by the inherent properties of the observables when applied to the given quantum state.

The inequality in Eq.(\ref{eq:multi_obs_boundmain}) establishes a fundamental limit on the minimum achievable errors for any POVM that approximates the given set of observables on a quantum state. It provides a bound that holds true for any choice of orthonormal basis ${|u_q\rangle}$, and the tightest bound can be obtained by optimizing over all possible ${|u_q\rangle}$. 

In the case of pure states, the selection of a specific ${|u_q\rangle}$ is not necessary. The derived analytical bound guarantees to be tighter than simply summing up Branciard's bounds for two observables pairwisely when the total number of observables exceeds four. The comparison is detailed in Section S7 of the Supplementary Material.

It is also important to recognize that Eq.(\ref{eq:Auq}) inherently implies $\tilde{S}\geq 0$, which constitutes a refined version of Robertson's preparation uncertainty relation that solely reflects the observables' properties on the state without considering any measurements. The conventional Robertson's preparation uncertainty can be viewed as a specific instance of this refinement by consistently setting $\tilde{S}_{u_q}=S_{u_q}$. This refined formulation thus also paves the way for tighter preparation uncertainty bounds with independent significance.

\bigskip
\noindent\textbf{Error-tradeoff relation via semidefinite programming}

\noindent
We proceed to introduce a secondary approach that yields even tighter tradeoff relations. This method bypasses the need for selecting specific $\{|u_q\rangle\}$ and can be formulated as semi-definite programming (SDP), enabling efficient computation through readily available algorithms such as CVX \cite{cvx,diamond2016cvxpy} and YALMIP \cite{1393890}.

Again for any POVM, $\{M_m\}\in H_S$, it can be realized as projective measurement, $\{V_m\}\in H_S\otimes H_A$, with $(I\otimes\bra{\xi_0})V_m(I\otimes\ket{\xi_0})=M_m$. We can then construct $\{F_j=\sum_mf_j(m)V_m\}$ to approximate $\{X_j\otimes I_A\}$. Let $Q$ be an $n\times n$ Hermitian matrix, with its $jk$th element given as
\begin{equation}
\begin{aligned}
	Q_{jk}=&\Tr\left[(\rho\otimes\sigma)(F_j-X_j\otimes I)(F_k-X_k\otimes I)\right]\\
	=&\Tr\left[\rho \sum_m f_j(m)M_m f_k(m)\right]-\Tr\left(\rho R_jX_k\right)\\
	&-\Tr\left(\rho X_j R_k\right)+\Tr\left(\rho X_j X_k\right),
\end{aligned}
\end{equation}
here 
$R_j=(I\otimes\bra{\xi_0})F_j(I\otimes\ket{\xi_0})=\sum_m f_j(m)M_m$ is a Hermitian matrix in $\mathcal{H}_S$. 
We let $\mathbb{S}$ be a $n\times n$ block operator whose $jk$-th block is  $\mathbb{S}_{jk}=\sum_{m}f_j(m)M_m f_k(m)$, which is itself a Hermitian matrix, and let $\mathbb{R}=\begin{pmatrix}R_1 & R_2 & \cdots & R_n\end{pmatrix}^{\dagger}$, $\mathbb{X}=\begin{pmatrix}X_1 & X_2 & \cdots & X_n\end{pmatrix}^{\dagger}$. $\mathbb{S}$ and $\mathbb{R}$ both depend on the measurement with $\mathbb{S}_{jk}=\mathbb{S}_{jk}^\dagger=\mathbb{S}_{kj}$, $R_j=R_j^\dagger$. We have $\mathbb{S}\geq \mathbb{R}\mathbb{R}^{\dagger}$ (see Sec. S2 of the Supplementary Material). $Q$ can then be rewritten as $Q=\Tr_S\left[(I_n\otimes\rho)(\mathbb{S}-\mathbb{R}\mathbb{X}^{\dagger}-\mathbb{X}\mathbb{R}^{\dagger}+\mathbb{X}\mathbb{X}^{\dagger})\right]$ and the weighted mean squared error can be written as
\begin{equation}
\begin{aligned}
	\mathcal{E}=&\Tr(WQ)=\Tr\left[(W\otimes\rho)(\mathbb{S}-\mathbb{R}\mathbb{X}^{\dagger}-\mathbb{X}\mathbb{R}^{\dagger}+\mathbb{X}\mathbb{X}^{\dagger})\right],
\end{aligned}
\end{equation}
where $W\geq 0$ is a weighted matrix, which typically takes a diagonal form as $W=\mathrm{diag}\{w_1,\cdots,w_n\}$, but can also take other forms.  

Now assume $\{R_j^{\star}\}_{j=1}^n$ and $\mathbb{S}^{\star}$ are the optimal operators 
that lead to the minimal error, 
we then have
\begin{equation}
	\begin{aligned}
		\mathcal{E}\geq 
		&\Tr\left[(W\otimes\rho)(\mathbb{S}^{\star}-\mathbb{R}^{\star}\mathbb{X}^{\dagger}-\mathbb{X}\mathbb{R}^{\star\dagger}+\mathbb{X}\mathbb{X}^{\dagger})\right]\\
		\geq &\min_{\mathbb{S},\{R_j\}_{j=1}^n} \left\{\Tr\left[(W\otimes\rho)(\mathbb{S}-\mathbb{R}\mathbb{X}^{\dagger}-\mathbb{X}\mathbb{R}^{\dagger}+\mathbb{X}\mathbb{X}^{\dagger})\right]\right.\\
		&\left.|\mathbb{S}\geq \mathbb{R}\mathbb{R}^{\dagger},\mathbb{S}_{jk}=\mathbb{S}_{kj}=\mathbb{S}_{jk}^{\dagger}, R_j=R_j^{\dagger}\right\}.
	\end{aligned}
\end{equation}
The minimization can be formulated as a semi-definite programming with
\begin{equation}\label{main:SDP0}
\begin{aligned}
	\mathcal{E}_0=\min_{\mathbb{S},\{R_j\}_{j=1}^n} &\Tr\left[(W\otimes\rho)(\mathbb{S}-\mathbb{R}\mathbb{X}^{\dagger}-\mathbb{X}\mathbb{R}^{\dagger}+\mathbb{X}\mathbb{X}^{\dagger})\right]\\
	\text{subject to}\quad &\mathbb{S}_{jk}=\mathbb{S}_{kj}=\mathbb{S}_{jk}^{\dagger},\ \forall j,k\\
	&R_j=R_j^{\dagger},\ \forall j\\
	&\begin{pmatrix}
		I & \mathbb{R}^{\dagger}\\
		\mathbb{R} & \mathbb{S}
	\end{pmatrix}\geq 0.
\end{aligned}
\end{equation}
The derived lower bound, $\mathcal{E} \geq \mathcal{E}_0$, offers a tighter constraint than the analytical bounds from the previous section for any selection of $\{|u_q\rangle\}$ (refer to Sec. S3 of the Supplementary Material). Furthermore, an explicit construction detailing the optimal approximation strategy that attains this bound for pure states is provided in Sec. S4 of the Supplementary Material, which demonstrates the tightness of the bound for any number of observables when applied to pure states.
For mixed states, however, the bound is in general not tight (see Sec. S4 of the Supplementary Material for an example).



\bigskip
\noindent\textbf{Tighter analytical relation for two observables}

\noindent
By leveraging the SDP bound provided in Eq.(\ref{main:SDP0}) and employing a judicious selection of ${|u_q\rangle}$ analogues to the analytical bound in Eq.(\ref{eq:multi_obs_boundmain}), we can derive analytical bounds on mixed states for two observables that are tighter than the Ozawa's relation, the tightest analytical bound previously known. 

We first show that (see Sec. S5 of the Supplementary Material for details) when $\rho=|\psi\rangle\langle \psi|$ is a pure state and $W=\mathrm{diag}\{w_1,w_2\}$,  Eq.(\ref{main:SDP0}) can be analytically solved as
\begin{equation}\label{eq:mainanalytical}
w_1\epsilon_1^2+w_2\epsilon_2^2\geq\frac{1}{2}\left(\alpha-\sqrt{\alpha^2-\beta^2}\right),
\end{equation}
where 
\begin{eqnarray}\label{eq:ab}
\begin{aligned}
	\alpha&=w_1(\Delta X_1)^2+w_2(\Delta X_2)^2,\\   \beta&=i\sqrt{w_1w_2}\bra{\psi}[X_1,X_2]\ket{\psi}.     
\end{aligned}
\end{eqnarray}
Since the SDP bound is tight for pure state, this analytical bound is also tight for pure states. We now use it to obtain tighter analytical bounds for two observables on mixed states.  
For a mixed state, $\rho$, we can choose any $\{|u_q\rangle\}$ with $\sum_q |u_q\rangle\langle u_q|=I$ and write 
\begin{equation}
\rho=\sum_q\sqrt{\rho}\ket{u_q}\bra{u_q}\sqrt{\rho}=\sum_q\lambda_q\ket{\phi_q}\bra{\phi_q},
\end{equation}
here
$\lambda_q=\bra{u_q}\rho\ket{u_q}$, $\ket{\phi_q}=\frac{\sqrt{\rho}\ket{u_q}}{\sqrt{\bra{u_q}\rho\ket{u_q}}}.$ For each $|\phi_q\rangle$ we can get a corresponding lower bound $\mathcal{E}_{|\phi_q\rangle}$ by substituting $|\phi_q\rangle\langle \phi_q|$ in Eq.(\ref{main:SDP0}), 
and solve it analytically to get 
$\mathcal{E}_{|\phi_q\rangle}=\frac{1}{2}\left(\alpha_q-\sqrt{\alpha_q^2-\beta_q^2}\right)$,
where $\alpha_q$ and $\beta_q$ are obtained from Eq.(\ref{eq:ab}) by substituting $\ket{\psi}$ with $\ket{\phi_q}$. 
Since for any function that satisfies $f(\sum_q\lambda_q x_q,y)=\sum_q \lambda_qf(x_q,y)$, we have $\min_yf(\sum_q\lambda_q x_q,y)\geq \sum_q \lambda_q \min_yf(x_q,y)$, by substitute $f$ with $\Tr\left[(W\otimes\rho)(\mathbb{S}-\mathbb{R}\mathbb{X}^{\dagger}-\mathbb{X}\mathbb{R}^{\dagger}+\mathbb{X}\mathbb{X}^{\dagger})\right]$, $x_q$ with $\ket{\phi_q}\bra{\phi_q}$ and $\sum_q\lambda_q x_q$ with $\rho$ in Eq.(\ref{main:SDP0}), we can get
$\mathcal{E}_0\geq\sum_q\lambda_q\mathcal{E}_{|\phi_q\rangle}$. This then leads to an analytical bound
\begin{equation}\label{eq:mixbound}
	\aligned
	w_1\epsilon_1^2+w_2\epsilon_2^2&\geq
	\mathcal{E}_0\geq\sum_q\lambda_q\mathcal{E}_{|\phi_q\rangle}\\
	&=
	\sum_q \frac{\lambda_q}{2}\left(\alpha_q-\sqrt{\alpha_q^2-\beta_q^2}\right)\equiv \mathcal{E}_A.    
	\endaligned
\end{equation}

In comparison, the minimal weighted error corresponding to the Ozawa's relation is given by(see Sec. S6 of the Supplementary Material)
\begin{equation}\label{eq:ozawamix}
	w_1\epsilon_1^2+w_2\epsilon_2^2\geq \frac{1}{2}\left(\alpha_{\rho}-\sqrt{\alpha_{\rho}^2-\beta_{\rho}^2}\right)\equiv\mathcal{E}_{\text{Ozawa}},
\end{equation} 
with $\alpha_{\rho}=w_1(\Delta X_1)^2+w_2(\Delta X_2)^2$ and $\beta_{\rho}=i\sqrt{w_1w_2}\|\sqrt{\rho}[X_1,X_2]\sqrt{\rho}\|_1$, here $(\Delta X_{1/2})^2=\Tr(\rho X_{1/2}^2)-\Tr(\rho X_{1/2})^2$.

By choosing $\{|u_q\rangle\}$ as the eigenstates of $\sqrt{\rho}[X_1,X_2]\sqrt{\rho}$, we can get 
(see Sec. S6 of the Supplementary Material for detail)
\begin{equation}
	\sum_q \frac{\lambda_q}{2}\left(\alpha_q-\sqrt{\alpha_q^2-\beta_q^2}\right)\geq \frac{1}{2}\left(\alpha_{\rho}-\sqrt{\alpha_{\rho}^2-\beta_{\rho}^2}\right).   
\end{equation}
The analytical bound in Eq.(\ref{eq:mixbound}) is thus tighter than the bound obtained from the Ozawa's relation. 
An illustrative example is shown in Fig. \ref{fig:example-main} with
\begin{equation}\label{eq:example-main}
    \rho=\begin{pmatrix}
			\frac{p}{2} & 0 & 0\\
			0 & 1-p & 0\\
			0 & 0 & \frac{p}{2}
		\end{pmatrix},
    X_1=\begin{pmatrix}
			0 & 1 & 0\\
			1 & 0 & 1\\
			0 & 1 & 0
		\end{pmatrix},
    X_2=\begin{pmatrix}
			0 & -i & 0\\
			i & 0 & -i\\
			0 & i & 0
		\end{pmatrix}.
\end{equation}
The presented framework thus not only extends to scenarios involving an arbitrary number of observables but also provides improved analytical bounds in the case of two observables.

\begin{figure}[htb]
	\centering
	\includegraphics[width=0.45\textwidth]{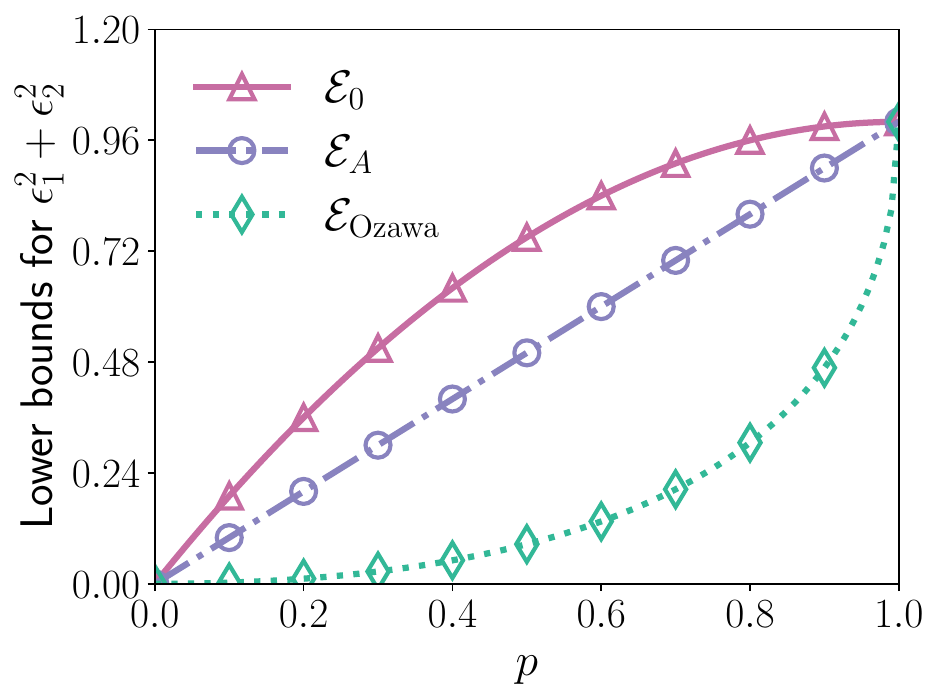}
	\caption{Lower bounds for the errors of the simultaneous measurement of two observables $X_1,X_2$ on state $\rho$ with different $p$, where $X_1,X_2$, and $\rho$ are given in Eq.(\ref{eq:example-main}). Here $\mathcal{E}_0$ is computed with SDP in Eq.(\ref{main:SDP0}), $\mathcal{E}_A$ is analytically obtained from Eq.(\ref{eq:mixbound}) and $\mathcal{E}_{\text{Ozawa}}$ is obtained from the Ozawa's relation (Eq.(\ref{eq:ozawamix})). It can be seen that for the simultaneous measurement of $X_1,X_2$, we have $\mathcal{E}_0\geq\mathcal{E}_A\geq\mathcal{E}_{Ozawa}$.}
	\label{fig:example-main}
\end{figure}

    \bigskip
\noindent\textbf{Tradeoff relations for multiparameter quantum estimation}

\noindent
We proceed to apply the derived error-tradeoff relations to the field of quantum metrology, thereby elucidating tradeoffs in the precision limits for estimating multiple parameters, which is currently a central focus in quantum metrology.
	
Given a quantum state $\rho_x$, where $x=(x_1,x_2,\ldots,x_n)$ are $n$ unknown parameters to be estimated, by performing a POVM, $\{M_{m}\}$, on the state, we can get the measurement result, $m$, with a probability $p_m(x)=\Tr(\rho_xM_{m})$. For any locally unbiased estimator, $\hat{x}=(\hat{x}_1,\cdots,\hat{x}_n)$, the Cram\'er-Rao bound \cite{Cram46, Fish22} provides an achievable bound
	$ Cov(\hat{x})\geq \frac{1}{\nu}F_C^{-1}$,
	here $Cov(\hat{x})$ is the covariance matrix for the estimators with the $jk$-th entry given by $Cov(\hat{x})_{jk}=E[(\hat{x}_j-x_j)(\hat{x}_k-x_k)]$, $E[\cdot]$ denotes the expectation, $\nu$ is the number of the measurement repeated independently, $F_C$ is the Fisher information matrix of a single measurement whose $jk$-th entry is given by $(F_C)_{jk}=\sum_{m} \frac{\partial_{x_j} p_m(x)\partial_{x_k}p_m(x)}{p_m(x)}$ \cite{Fish22}. 
	Regardless of the choice of the measurement, the covariance matrix is always lower bounded by the quantum Cram\'er-Rao bound as \cite{Hels76book,Hole82book}
	\begin{equation}
		Cov(\hat{x})\geq \frac{1}{\nu}F_C^{-1}\geq \frac{1}{\nu}F_Q^{-1},
	\end{equation}
	here $F_Q$
	is the quantum Fisher information matrix with the $jk$-th entry given by
	$   (F_Q)_{jk}=\Tr(\rho_x\frac{L_jL_k+L_kL_j}{2}),$
	where $L_q$ is the symmetric logarithmic derivative (SLD) corresponding to the parameter $x_q$, which satisfies $\partial_{x_q}\rho_x=\frac{1}{2}(\rho_xL_q+L_q\rho_x)$. When there is only one parameter, the QCRB can be saturated. In particular it can be saturated with the projective measurement on the eigen-spaces of the SLD, i.e., the SLD is the optimal observable for the estimation of the corresponding parameter.

	When there are multiple parameters, the QCRB is in general not saturable since the SLDs typically do not commute with each other. 
	A central task in multi-parameter quantum estimation is to understand the tradeoff induced by such incompatibility \cite{genoni2013jan,steinlechner2013aug,zhuang2017oct,huang2021jul,MatsumotoThesis,HongzhenPra,HongzhenPRL,GillM00,Zhu2018universally,Lu2021,Suzuki2016,Sidhu2021,Nagaoka1,Conlon2021,ALBARELLI2020126311,Federico2021,Carollo_2019,Ragy2016,Chen_2017,Liu_2019,ChenHZ2019,Rafal2020,vidrighin2014,crowley2014,Yue2014,Zhang2014,Liu2017,Roccia_2017,e22111197,Candeloro_2021,Yuxiang2019,HouMinimal2020,HouSuper2021,Houeabd2986,Ragy2016,Szczykulska2016,FrancescoPRL}.
	In a recent seminal work \cite{Lu2021}, by applying Ozawa's uncertainty relation, Lu and Wang obtained analytical tradeoff relations for the estimation of a pair of parameters. For multiple ($n>2$) parameters, however, if we simply add the tradeoff for each pair of parameters directly, the obtained tradeoff relation is typically loose, which restricts the scope of its applications. 
	
    \bigskip
    
\noindent\textit{Analytical tradeoff relation for multiparameter quantum estimation}

\noindent
By directly applying the tradeoff relation from Eq.(\ref{eq:multi_obs_boundmain}), we can readily obtain a tradeoff relation for estimating multiple parameters by simply substituting the $n$ observables with the $n$ SLDs, $\{L_1, \cdots, L_n\}$. In this context, $S_{\text{Re}}$ corresponds to the quantum Fisher information matrix, $F_Q$. For any POVM $\{M_m\}$, we construct $\{F_j = \sum_m f_j(m)V_m\}$ to approximate the SLDs. As previously defined, each $V_m$ represents a projective measurement in the extended space that projects onto $M_m$ when acting on the system. Given that the error-tradeoff relation in Eq.(\ref{eq:multi_obs_boundmain}) holds for any choice of the functions $\{f_j(m)\}$, we can specifically select 
 \begin{equation}
		f_j(m)=\frac{1}{2}\frac{\Tr[\{L_j\otimes I,\rho_x\otimes\sigma\}V_m]}{\Tr[(\rho_x\otimes\sigma) V_m]}=\frac{\partial_{x_j}p_m(x)}{p_m(x)},
	\end{equation}
	where $p_m(x)=\Tr[(\rho_x\otimes\sigma) V_m]=\Tr[\rho_x M_m]$. This choice minimizes $\epsilon_j^2 = \Tr\left[(F_j-L_j\otimes I)^2\left(\rho_x\otimes\sigma\right)\right]$ under the given measurement (see Methods). With this choice we have $Q_{\text{Re}}=F_Q-F_C$, 
	Eq.(\ref{eq:multi_obs_boundmain}) then becomes
	\begin{equation}
		\Tr[F_Q^{-1}(F_Q-F_C)]\geq \left(\sqrt{\|F_Q^{-\frac{1}{2}}\tilde{S}_{\text{Im}}F_Q^{-\frac{1}{2}}\|_{F}+1}-1\right)^2,
	\end{equation}
	where $\tilde{S}_{\text{Im}}=\sum_q \tilde{S}_{u_q,\text{Im}}$ with each
	$\tilde{S}_{u_q,\text{Im}}$ equals to either $S_{u_q,\text{Im}}$ or $S_{u_q,\text{Im}}^T$, here $(S_{u_q,\text{Im}})_{jk}=\frac{1}{2i}\bra{u_q}(\sqrt{\rho_x}\otimes \ket{\xi_0}\bra{\xi_0}) ([L_j,L_k]\otimes  I)(\sqrt{\rho_x}\otimes \ket{\xi_0}\bra{\xi_0})\ket{u_q}$ with $\{\ket{u_q}\}$ as any set of vectors that satisfies $\sum_q \ket{u_q}\bra{u_q}=I$. 
	This then provides an upper bound on the achievable classical Fisher information matrix as
	\begin{equation}\label{eq:precision}
		\Tr[F_Q^{-1}F_C]\leq n-\left(\sqrt{\|F_Q^{-\frac{1}{2}}\tilde{S}_{\text{Im}}F_Q^{-\frac{1}{2}}\|_F+1}-1\right)^2.
	\end{equation}

    In comparison with the previous metrological bounds in \cite{HongzhenPra,HongzhenPRL}, the current bound is less stringent. This is because the previous metrological bounds exploit the properties of locally unbiased estimators, which introduce additional constraints that are not present in the simultaneous measurement of multiple observables. The bound here is based only on the characteristics of the SLD observables, reflecting the inherent uncertainties of the observables. The difference in constraints is what creates the difference in strictness between the previous bounds and the current one. This also means that previous metrological bounds cannot be directly generalized to the uncertainty relations for incompatible observables.

    \bigskip
    \noindent\textit{SDP-based tradeoff relation for multiparameter quantum estimation}
    
    \noindent
	In the scenario of estimating multiple parameters encapsulated in a quantum state, $\rho_x$, with $x=(x_1,\cdots,x_n)$, we can substitute the matrix $\mathbb{X}$ in Eq.(\ref{main:SDP0}) with $\mathbb{L}=\begin{pmatrix}L_1 & L_2 & \cdots & L_n\end{pmatrix}^{\dagger}$, here $L_q$ is the SLD for the parameter $x_q$. Upon choosing the optimal functions $f_j(m)$, the real part of the error matrix becomes $Q_{\text{Re}}=F_Q-F_C$. The derived bound $\mathcal{E}\geq \mathcal{E}_0$ then gives a tradeoff relation in multiparameter quantum estimation 
	\begin{equation}		\mathcal{E}=\Tr(WQ)=\Tr(WQ_{\text{Re}})=\Tr[W(F_Q-F_C)]\geq \mathcal{E}_0,
	\end{equation}
	where we used the property that $\Tr(WQ_{\text{Im}})=\Tr(W^{\frac{1}{2}}Q_{\text{Im}}W^{\frac{1}{2}})=0$ since $Q_{\text{Im}}$ is anti-symmetric, and here
	\begin{equation}
		\begin{aligned}
			\mathcal{E}_0=\min_{\mathbb{S},\{R_j\}_{j=1}^n} &\Tr\left[(W\otimes\rho_x)(\mathbb{S}-\mathbb{R}\mathbb{L}^{\dagger}-\mathbb{L}\mathbb{R}^{\dagger}+\mathbb{L}\mathbb{L}^{\dagger})\right]\\
			\text{subject to}\quad &\mathbb{S}_{jk}=\mathbb{S}_{kj}=\mathbb{S}_{jk}^{\dagger},\ \forall j,k\\
			&R_j=R_j^{\dagger},\ \forall j\\
			&\begin{pmatrix}
				I & \mathbb{R}^{\dagger}\\
				\mathbb{R} & \mathbb{S}
			\end{pmatrix}\geq 0.
		\end{aligned}
	\end{equation}
	We can perform a reparameterization by setting 
 $\tilde{x}=F_Q^{-\frac{1}{2}}x$ which leads to $\tilde{F}_Q=I$. Under this transformation, we have that $\Tr[I\otimes \rho_x)\tilde{\mathbb{L}}\tilde{\mathbb{L}}^\dagger]=\Tr(\tilde{F}_Q)=n$, where $\tilde{\mathbb{L}}=\begin{pmatrix}\tilde{L}_1 & \tilde{L}_2 & \cdots & \tilde{L}_n\end{pmatrix}^{\dagger}$ with $\tilde{L}_j=\sum_k (F_Q^{-\frac{1}{2}})_{jk}L_k$. By setting $W=I$ we can derive an upper bound for $\Tr(\tilde{F}_C)=\Tr(F_Q^{-1}F_C)$ as
	\begin{equation}\label{eq:precisionSDP}
		\begin{aligned}
			\Tr(F_Q^{-1}F_C)\leq 
			\max_{\mathbb{S},\{R_j\}_{j=1}^n} &\Tr\left[(I\otimes\rho_x)(-\mathbb{S}+\mathbb{R}\tilde{\mathbb{L}}^{\dagger}+\tilde{\mathbb{L}}\mathbb{R}^{\dagger})\right]\\
			\text{subject to}\quad &\mathbb{S}_{jk}=\mathbb{S}_{kj}=\mathbb{S}_{jk}^{\dagger},\ \forall j,k\\
			&R_j=R_j^{\dagger},\ \forall j\\
			&\begin{pmatrix}
				I & \mathbb{R}^{\dagger}\\
				\mathbb{R} & \mathbb{S}
			\end{pmatrix}\geq 0.
		\end{aligned}
	\end{equation}
	This resembles the Nagaoka-Hayashi bound \cite{Conlon2021} but are different. The Nagaoka-Hayashi bound quantifies $\Tr[WCov(\hat{x})]$, where $\hat{x}$ is required to be locally unbiased estimators to satisfy the classical Cramer-Rao bound as $Cov(\hat{x})\geq \frac{1}{\nu}F_C^{-1}$. The Nagaoka-Hayashi bound thus has the locally unbiased condition included in the constraints. And it does not have a direct connection to the approximation of the observables, which is reflected in the fact that its objective function does not contain the SLDs. While the bound here quantifies directly the relation between $F_Q$ and $F_C$, without the intermediate step of estimators, it thus does not have the locally unbiased condition in the constraints. The presence of the SLD operators in the objective function underscores an intrinsic connection to observable approximation, which is absent in the Nagaoka-Hayashi bound.
	
    \bigskip
    \noindent\textit{Sharper tradeoff relation for two-parameter quantum estimation}
    
    \noindent
	In the case of estimating two parameters, a tighter analytical bound can be derived by substituting $(X_1,X_2)$ in Eq.(\ref{eq:mixbound}) with $(L_1,L_2)$, which results in

\begin{equation}\label{eq:precision2}
		\begin{aligned}
			&w_1[(F_Q)_{11}-(F_C)_{11}]+w_2[(F_Q)_{22}-(F_C)_{22}]\\
			\geq &\sum_q \frac{\lambda_q}{2}\left(\alpha_q-\sqrt{\alpha_q^2-\beta_q^2}\right),
		\end{aligned}
	\end{equation}
	where $\lambda_q=\bra{u_q}\rho_x\ket{u_q}$, $\alpha_q$ and $\beta_q$ are given as
	\begin{equation}
		\begin{aligned}
			\alpha_q&=w_1(\Delta_{\ket{\phi_q}} L_1)^2+w_2(\Delta_{\ket{\phi_q}} L_2)^2,\\   \beta_q&=i\sqrt{w_1w_2}\bra{\phi_q}[L_1,L_2]\ket{\phi_q},
		\end{aligned}
	\end{equation}
	here $\ket{\phi_q}=\frac{\sqrt{\rho_x}\ket{u_q}}{\sqrt{\bra{u_q}\rho_x\ket{u_q}}}$ and $\Delta_{\ket{\phi_q}} L_{1/2}$ denotes the standard deviation of the SLD $L_{1/2}$ on the state $\ket{\phi_q}$. 
	By selecting ${|u_q\rangle}$ as the eigenvectors of $\sqrt{\rho_x}[L_1,L_2]\sqrt{\rho_x}$, Eq. (\ref{eq:precision2}) provides a tighter bound than the Lu-Wang bound for mixed states \cite{Lu2021}, which is based on Ozawa's relation(see Section S6 of the supplementary material for detail). 
 
	
	

\bigskip
\noindent\textbf{Experiment validation of the error-tradeoff relations in a superconducting quantum processor}

\noindent
We conducted an experimental verification of the error-tradeoff relations on a superconducting quantum processor, utilizing the Quafu cloud quantum computing platform \cite{Quafu3}. The selected processor, ScQ-P136, consists of 136 qubits with single-qubit gate fidelities surpassing $99\%$ \cite{Quafu1,Quafu2,Quafu3}, and for our analysis, we focused exclusively on the first qubit.
Further details about the processor's architecture and parameters are provided in the Methods section. 

To experimentally quantify the error $\epsilon_j$ for each observable $X_j$ when measured through a specific measurement set $\{M_m\}$, we adopt the ``3-state method''\cite{Jacqueline2012,Ringbauer2014}. The essence of this approach is illustrated in Fig.\ref{fig:experiment}(a), which involves preparing and measuring three distinct quantum states:
\begin{equation}
    \rho_1 = \rho,\ \rho_2 \simeq X_j\rho X_j,\ \rho_3 \simeq (I + X_j)\rho(I + X_j).
\end{equation}
By analyzing the measurement statistics of $\{M_m\}$ on these three states—$\rho_1$, $\rho_2$, and $\rho_3$—we can obtain $\epsilon_j$, the error of the approximation. Detailed information on how to determine the error from these measurements can be found in the Methods section.


\begin{figure}
    \centering
    \includegraphics[width=0.48\textwidth]{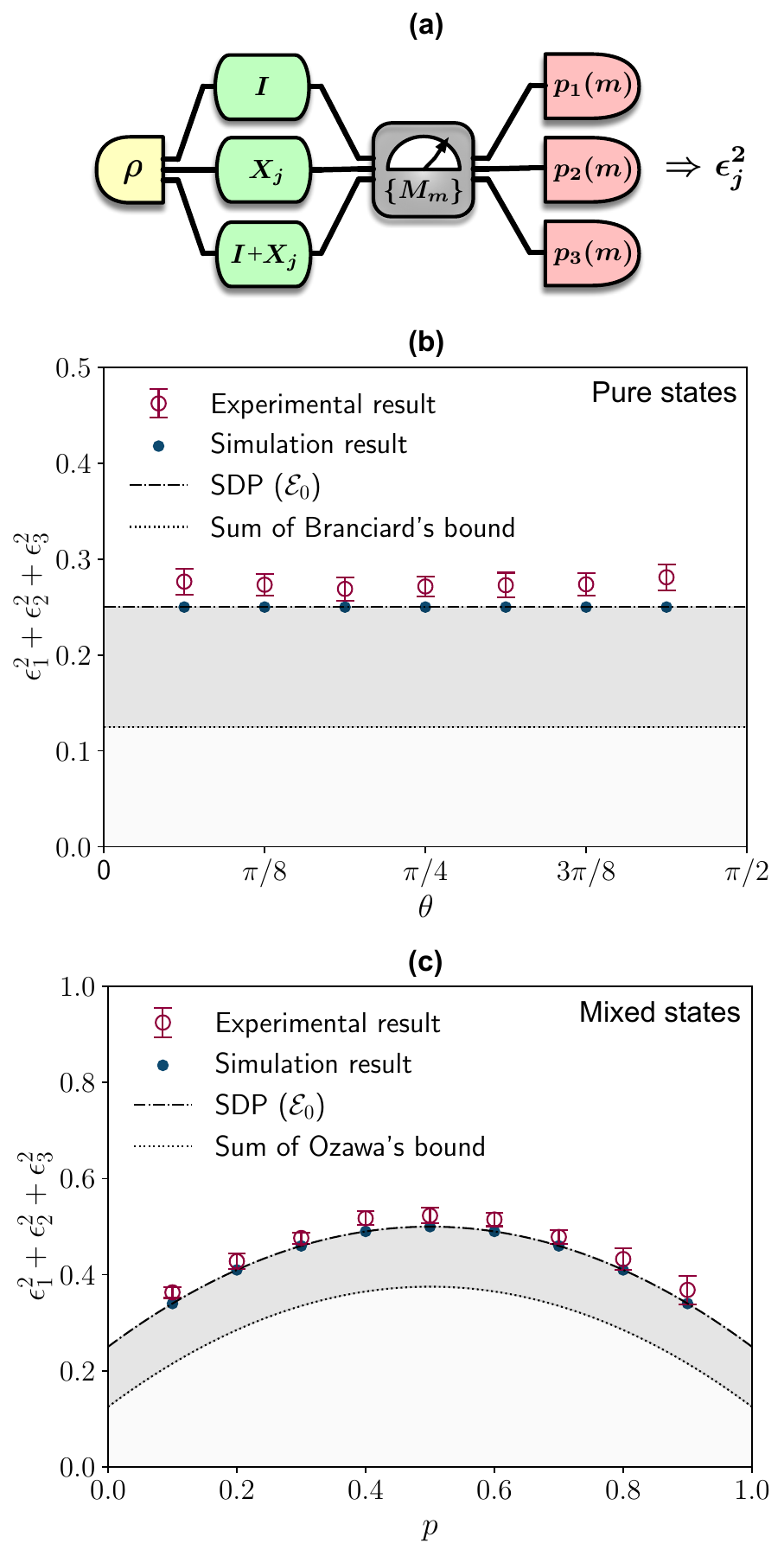}
    \caption{Experimental results for testing the error-tradeoff relations in superconducting quantum processor ``ScQ-P136'' accessed through the Quafu cloud quantum computation platform.
    (a) Scheme diagram to evaluate the mean squared error of each $X_j$ using the ``3-state method''.  
    (b) Error-tradeoff relations for the simultaneous measurement of three spin operators, $\{\frac{1}{2}\sigma_x,\frac{1}{2}\sigma_y,\frac{1}{2}\sigma_z\}$, on a pure state $\ket{\psi}=R_z(\frac{\pi}{2})R_y(\theta)\ket{0}$.
    (c) Error-tradeoff relations for the simultaneous measurement of three spin operators, $\{\frac{1}{2}\sigma_x,\frac{1}{2}\sigma_y,\frac{1}{2}\sigma_z\}$, on a mixed state $\rho=p\ket{0}\bra{0}+(1-p)\ket{1}\bra{1}$.
    For both cases, the statistics of each measurement was obtained through 2000 shots and the error bars represent the standard deviations obtained by repeating the experiment 20 times.
    }
    \label{fig:experiment}
\end{figure}

\bigskip

\noindent
We commence with the simultaneous measurement of the three Pauli spin operators, $\{\frac{1}{2}\sigma_x, \frac{1}{2}\sigma_y, \frac{1}{2}\sigma_z\}$, on a single qubit system. The state of the qubit is first prepared as a pure state, specifically a rotation of the ground state by an angle $\theta$ around the $y$-axis followed by a rotation of $\frac{\pi}{2}$ around the $z$-axis, which can be expressed as $\ket{\psi} = R_z(\frac{\pi}{2})R_y(\theta)\ket{0}$ with $0<\theta\leq\frac{\pi}{2}$.
Before conducting the experiment, we first substitute the given qubit state and spin operators into Eq.(\ref{main:SDP0}) to derive a tight theorectical lower bound for the total mean-squared-error $\epsilon_1^2+\epsilon_2^2+\epsilon_3^2$. The optimal primal variables of $\mathcal{E}_0$ are recorded as $\{R_1^{\star}, R_2^{\star}, R_3^{\star}\}$.
Leveraging these values, we construct the corresponding optimal measurement scheme $\{M_m^{\star}\}$ (see Sec. S4 of the Supplementary Material for details).

In the experimental phase, we apply the optimal measurement scheme $\{M_m^{\star}\}$ to the prepared pure state $\rho = \ket{\psi}\bra{\psi}$ to approximate the joint measurement of the three Pauli operators $\{\frac{1}{2}\sigma_x,\frac{1}{2}\sigma_y,\frac{1}{2}\sigma_z\}$. The errors for each operator are then estimated using the ``3-state method'' by preparing auxiliary states $\rho_2 \simeq X_j\rho X_j$ and $\rho_3 \simeq (I + X_j)\rho(I + X_j)$. We carry out this procedure for various values of the rotation angle $\theta$. The corresponding total mean-squared-errors are plotted in Fig.\ref{fig:experiment}(b), where simulated results are also presented for comparison.

\noindent
Next, we proceed to verify our tradeoff relations for mixed states by simultaneously measuring the three spin operators $\{\frac{1}{2}\sigma_x,\frac{1}{2}\sigma_y,\frac{1}{2}\sigma_z\}$ on a mixed state of the form $\rho = p\ket{0}\bra{0} + (1-p)\ket{1}\bra{1}$, where $0 \leq p \leq 1$. Similarly, prior to the experimental phase, we input the mixed state and observables into Eq.(\ref{main:SDP0}) to derive a theoretical lower bound for the total mean-squared-error $\epsilon_1^2+\epsilon_2^2+\epsilon_3^2$. This is achieved by solving the SDP problem and computing its optimal solution. In the context of mixed states, we resort to numerical methods to determine the corresponding optimal measurements $\{M_m^{\star}\}$ instead of analytical solutions. These numerically derived optimal measurements are then applied in the experimental setup to approximate the joint measurement of the spin operators $\{\frac{1}{2}\sigma_x,\frac{1}{2}\sigma_y,\frac{1}{2}\sigma_z\}$ on the given mixed state. Once more, we utilize the ``3-state method'' to calculate the total errors by preparing auxiliary states: $\rho_2 \simeq X_j\rho X_j$ and $\rho_3 \simeq (I + X_j)\rho(I + X_j)$. This process is repeated for a range of different mixing parameter values $p$, and the corresponding total mean-squared-errors are presented graphically in Fig.\ref{fig:experiment}(c), along with accompanying simulation data for comparison.
	
    \bigskip
    \noindent\textbf{Discussion}
    
    \noindent
    We have developed methodologies that establish tradeoff relations for approximating any number of observables using a single measurement, and we have also refined the existing analytical bounds in scenarios involving two observables. Each derived bound has its unique benefits and drawbacks:

1. The inequality $$\Tr(S_{\text{Re}}^{-1}Q_{\text{Re}}) \geq \left(\sqrt{\|S_{\text{Re}}^{-\frac{1}{2}}\tilde{S}_{\text{Im}}S_{\text{Re}}^{-\frac{1}{2}}\|_{F}+1}-1\right)^2$$ offers an analytical constraint applicable to any number of observables. For pure states, the selection of $\{|u_q\rangle\}$ is not necessary and the bound outperforms the sum of Ozawa's relations when estimating more than four observables.
However, for mixed states, the tightness depends on the choice of the set $\{|u_q\rangle\}$ where $|u_q\rangle\langle u_q|=I$ (see Section S9.A of the Supplementary Material for an explicit example). Although every selection of $\{|u_q\rangle\}$ produces a valid bound, there is currently no systematic method available for identifying the optimal choice. This analytical expression provides a universal benchmark but might not always yield the most stringent limitation, especially in complex mixed-state scenarios.

2. The error tradeoff relation $\mathcal{E} = \Tr(WQ) \geq \mathcal{E}_0$ provides the most stringent bound, where $\mathcal{E}_0$ can be efficiently computed using semidefinite programming (SDP),
\begin{equation*}
    \begin{aligned}
        \mathcal{E}_0 = \min_{\mathbb{S}, \{R_j\}_{j=1}^n} &\Tr\left[(W \otimes \rho)(\mathbb{S} - \mathbb{R}\mathbb{X}^\dagger - \mathbb{X}\mathbb{R}^\dagger + \mathbb{X}\mathbb{X}^\dagger)\right]\\
        \text{subject to:}\quad &\mathbb{S}_{jk} = \mathbb{S}_{kj} = \mathbb{S}_{jk}^{\dagger}, \ \forall j,k \\
        &R_j = R_j^{\dagger}, \ \forall j \\
        &\begin{pmatrix}
            I & \mathbb{R}^{\dagger}\\
            \mathbb{R} & \mathbb{S}
        \end{pmatrix} \geq 0.
    \end{aligned}
\end{equation*}
This SDP-based approach is universally applicable for any number of observables and offers a tighter bound than the first analytical relation. For pure states, this bound is exact, and it outperforms the first bound for mixed states regardless of the choice of the set $\{|u_q\rangle\}$ satisfying $\sum_q |u_q\rangle\langle u_q| = I$. Furthermore, it consistently provides a stricter constraint compared to the sum of Ozawa's relations in all scenarios. This SDP-based bound does not have a general analytical expression, necessitating numerical methods for its computation.

3. For a pair of observables $X_j$ and $X_k$ acting on the state $\rho$, we have an analytical bound given by:
\begin{equation*}
    \aligned
    w_j\epsilon_j^2+w_k\epsilon_k^2 &\geq \sum_q \frac{\lambda_q}{2}\left(\alpha_q-\sqrt{\alpha_q^2-\beta_q^2}\right),   
    \endaligned
\end{equation*}
where $\lambda_q=\bra{u_q}\rho\ket{u_q}$, $\ket{\phi_q}=\frac{\sqrt{\rho}\ket{u_q}}{\sqrt{\bra{u_q}\rho\ket{u_q}}},$ $\alpha_q=w_j[\bra{\phi_q} X_j^2\ket{\phi_q}-\bra{\phi_q} X_j\ket{\phi_q}^2]+w_k[\bra{\phi_q} X_k^2\ket{\phi_q}-\bra{\phi_q} X_k\ket{\phi_q}^2]$,   $\beta_q=i\sqrt{w_jw_k}\bra{\phi_q}[X_j,X_k]\ket{\phi_q}$. 
This inequality holds for any choice of $\{|u_q\rangle\}$ satisfying $\sum_q |u_q\rangle\langle u_q| = I$. Notably, if one selects $\{|u_q\rangle\}$ to be the eigenvectors of $\sqrt{\rho}[X_j,X_k]\sqrt{\rho}$, this bound is tighter than Ozawa's relation for mixed states. However, it only applies to pairs of observables.
Each of the error-tradeoff relations can be directly applied to assess the precision tradeoffs in multi-parameter quantum metrology. These relations are valid for local measurements, which involve independently measuring each copy of the state $\rho_x$. When considering $p$-local measurements, where the measurement process may involve collective actions on up to $p$ copies of $\rho_x$, analogous tradeoff relations can be derived by substituting $\rho_x$ and its set of SLDs $\{L_j\}$ with $\rho_x^{\otimes p}$ and the corresponding SLDs for $\rho_x^{\otimes p}$. In Sec. S8 of the Supplementary Materials, we present an illustrative example that showcases the tradeoff relation under collective measurements. This paves the way for further exploration into state-dependent measurement uncertainty relations for multiple observables. Moreover, it reinforces the connection between measurement uncertainty and the incompatibility inherent to multi-parameter quantum estimation, thereby promoting deeper investigations across both domains. 

\bigskip
\noindent\textbf{Methods}
\noindent\textbf{Tradeoff in multi-parameter quantum estimation}

\noindent
By taking the set of SLDs, $\{L_q\}$, as the observables, we can obtain the trade-off relations in multiparameter quantum estimation. 

For any POVM $\mathcal{M}=\{M_m\}$, it can be equivalently written as a projective measurement $\{V_m=U^\dagger(I\otimes \ket{\xi_m}\bra{\xi_m})U\}$ on an extended state $\rho_x\otimes \sigma$ with $\sigma=\ket{\xi_0}\bra{\xi_0}$ such that $ (I\otimes\bra{\xi_0})V_m  (I\otimes\ket{\xi_0})=M_m$. Commuting observables, $\{F_1, F_2, \ldots, F_n\}$ with $F_j=\sum_m f_j(m)V_m$, are then constructed from this measurement on the extended Hilbert space to approximate $\{L_1\otimes I,\cdots, L_n\otimes I\}$. Note that the tradeoff relations for the approximate measurement holds for any choice of $\{f_j(m)\}$, here we make a particular choice of $\{f_j(m)\}$ to minimize the root-mean-squared error $\{\epsilon_j^2=\Tr\left[(F_j-L_j\otimes I)^2\left(\rho_x\otimes\sigma\right)\right]\}$. As
\begin{equation}
    \begin{aligned}
        & \Tr\left[(F_j-L_j\otimes I)^2\left(\rho_x\otimes\sigma\right)\right] \\
        = &\Tr\left[\left(\sum_m f_j(m)V_m-L_j\otimes I\right)^2\left(\rho_x\otimes\sigma\right)\right]\\
        =& \Tr(\rho_x L_j^2) - \sum_{m} p_m(x) \left(\frac{1}{2}\frac{\Tr[\{L_j\otimes I,\rho_{x}\otimes\sigma\}V_m]}{\Tr[(\rho_{x}\otimes\sigma) V_m]}\right)^2\\
        &+\sum_{m} p_m(x) \left(f_j(m)-\frac{1}{2}\frac{\Tr[\{L_j\otimes I,\rho_{x}\otimes\sigma\}V_m]}{\Tr[(\rho_{x}\otimes\sigma) V_m]}\right)^2,
    \end{aligned}
\end{equation}
here $p_m(x)=\Tr[(\rho_x\otimes\sigma) V_m]=\Tr(M_m\rho_x)$ is the probability for the measurement result $m$. The optimal $f_j(m)$ that minimizes the root-mean-squared error is then given by
\begin{equation}\label{eq:opt_fj_mix}
    f_j(m)=\frac{1}{2}\frac{\Tr[\{L_j\otimes I,\rho_{x}\otimes\sigma\}V_m]}{\Tr[(\rho_{x}\otimes\sigma) V_m]}.
\end{equation}
We note that when $p_m(x)=0$, $f_j(m)$ can take the form as 0/0, which should be computed as a multivariate limit with $x'\rightarrow x$.

So far we have not used any properties of the SLDs, the formula for the optimal choice of $f_j(m)$ works for any observables \cite{Branciard2013}. For the SLDs in particular, we have 
\begin{equation}
    \begin{aligned}
        f_j(m)&=\frac{1}{2}\frac{\Tr[\{L_j\otimes I,\rho_{x}\otimes\sigma\}V_m]}{\Tr[(\rho_{x}\otimes\sigma) V_m]}\\
        &=\frac{\Tr[ (\frac{1}{2}(L_j\rho_{x}+\rho_{x} L_j)\otimes \sigma) V_m]}{p_m(x)}\\
        &=\frac{\partial_{x_j} \Tr[(\rho_{x}\otimes \sigma) V_m]}{p_m(x)}
        =\frac{\partial_{x_j} p_m(x)}{p_m(x)}.
    \end{aligned}
\end{equation}
With this optimal choice, we have $F_j=\sum_{m} \frac{\partial_{x_j} p_m(x)}{p_m(x)}V_m$. The entries of $Q_{\text{Re}}$ can then be obtained as 
\begin{equation}
    \begin{aligned}
        &(Q_{\text{Re}})_{jk}=\text{Re}\Tr[(F_j-L_j\otimes  I)(F_k-L_k\otimes  I)(\rho_x\otimes \sigma)]\\
        &=\text{Re}\{\Tr[(F_jF_k)(\rho_x\otimes \sigma)]-\Tr[F_j (L_k\otimes  I)(\rho_x\otimes \sigma)]\\
        &\quad-\Tr[(L_j\otimes  I)F_k(\rho_x\otimes \sigma)]+\Tr[(L_jL_k\otimes  I)(\rho_x\otimes \sigma)]\}.
    \end{aligned}
\end{equation}
Here the first term,  
\begin{equation}
    \begin{aligned}
        &\text{Re}\{\Tr[F_jF_k(\rho_x\otimes \sigma)]\}\\
        &=\text{Re}\{\Tr[\sum_{m}  \frac{\partial_{x_j} p_m(x)}{p_m(x)}\frac{\partial_{x_k} p_m(x)}{p_m(x)}V_m(\rho_{x}\otimes \sigma)]\}\\
        &=\sum_{m} \frac{\partial_{x_j} p_m(x)\partial_{x_k} p_m(x)}{p_m(x)},
    \end{aligned}
\end{equation}
equals to the $jk$-th entry of the classical Fisher information matrix, $(F_C)_{jk}$. 
While the second term, 
\begin{eqnarray}
    \aligned
    &\text{Re}\{\Tr[F_j (L_k\otimes  I)(\rho_x\otimes \sigma)]\}\\
    &=\frac{1}{2}\Tr[(F_j (L_k\otimes  I)+(L_k\otimes I)F_j) (\rho_x\otimes \sigma)]\\
    &=\sum_{m}\frac{\partial_{x_j}p_m(x)}{p_m(x)} \Tr[(\frac{1}{2}(L_k\rho_{x}+\rho_{x}L_k)\otimes \sigma) V_m]\\
    &=\sum_{m}\frac{\partial_{x_j}p_m(x)}{p_m(x)} \partial_{x_k} \Tr[ (\rho_{x}\otimes \sigma) V_m]\\
    &=\sum_{m}\frac{\partial_{x_j}p_m(x)\partial_{x_k}p_m(x)}{p_m(x)},
    \endaligned
\end{eqnarray}
also equals to $(F_C)_{jk}$. It can be similarly shown that the third term equals to $(F_C)_{jk}$ as well. For the last term we have 
\begin{equation}
    \begin{aligned}
        &\text{Re}\{\Tr[(L_jL_k\otimes  I)(\rho_x\otimes \sigma)]\}=\text{Re}\{\Tr[L_jL_k \rho_x]\}\\
        &=\frac{1}{2}\Tr[(L_jL_k+L_kL_j)\rho_x],
    \end{aligned}
\end{equation}
which is just $(F_Q)_{jk}$. Put the four terms together, we can get $(Q_{\text{Re}})_{jk}=(F_Q)_{jk}-(F_C)_{jk}$. It is also straightforward to see that with the SLDs as the observables, we have $S_{\text{Re}}=F_Q$. The error tradeoff for the approximate measurement of the SLDs then leads to a tradeoff relation 
\begin{equation}
    \Tr[F_Q^{-1}(F_Q-F_C)]\geq \left(\sqrt{\|F_Q^{-\frac{1}{2}}\tilde{S}_{\text{Im}}F_Q^{-\frac{1}{2}}\|_F+1}-1\right)^2,
\end{equation}
where $\tilde{S}_{\text{Im}}=\sum_q \tilde{S}_{u_q,\text{Im}}$ with each 
$\tilde{S}_{u_q,\text{Im}}$ equals to either $S_{u_q,\text{Im}}$ or $S_{u_q,\text{Im}}^T$, here  
$(S_{u_q,\text{Im}})_{jk}=\frac{1}{2i}\bra{u_q}(\sqrt{\rho_x}\otimes \ket{\xi_0}\bra{\xi_0}) ([L_j,L_k]\otimes  I)(\sqrt{\rho_x}\otimes \ket{\xi_0}\bra{\xi_0})\ket{u_q}$ with $\{\ket{u_q}\}$ as any set of vectors that satisfies $\sum_q \ket{u_q}\bra{u_q}=I$. This can be rewritten as
\begin{equation}
    \Tr[F_Q^{-1}F_C]\leq n-\left(\sqrt{\|F_Q^{-\frac{1}{2}}\tilde{S}_{\text{Im}}F_Q^{-\frac{1}{2}}\|_F+1}-1\right)^2.
\end{equation}

\bigskip
\noindent\textbf{Evaluate the mean-squared-error using ``3-state method''}

\noindent
To experimentally evaluate the errors $\epsilon_j$, note that for each $j$,
\begin{equation}
    \begin{aligned}
        \epsilon_j^2=&\Tr[(\rho\otimes\sigma)(F_j-X_j\otimes I)^2]\\
        =&\Tr(\rho X_j^2)+\sum_m f_j(m)^2\Tr(\rho M_m)\\
        &-2\sum_mf_j(m)\operatorname{Re}\Tr(\rho M_m X_j).
    \end{aligned}
\end{equation}
Here, $\Tr(\rho M_m)$ represents the probability of obtaining outcome $m$ in the measurement, which can be directly obtained from experimental data.
To evaluate the quantities $\operatorname{Re}\Tr(\rho M_m X_j)$, the "3-state method" can be employed \cite{Jacqueline2012, Ringbauer2014}. We can express the terms as
\begin{equation}
    \begin{aligned}
        \operatorname{Re}\Tr(\rho M_m X_j)=&\frac{1}{2}\left(\Tr\left[M_m(I+X_j)\rho(I+X_j)\right]\right.\\
        &\left.-\Tr\left(M_m X_j\rho X_j\right)-\Tr\left(M_m\rho\right)\right).
    \end{aligned}
\end{equation}
By performing the measurement ${M_m}$ on three different states $\rho_1=\rho$, $\rho_2=\frac{X_j\rho X_j}{\Tr(X_j\rho X_j)}$, and $\rho_3=\frac{(I+X_j)\rho(I+X_j)}{\Tr\left[(I+X_j)\rho(I+X_j)\right]}$, we can estimate the three terms in the expression above. 
Assuming that all three terms are non-zero without loss of generality, we can calculate the optimal value for $f_j(m)$ as $f_j(m)=\frac{\operatorname{Re}\Tr(\rho M_m X_j)}{\Tr(\rho M_m)}$. This value can be directly computed using the obtained values of $\Tr(\rho M_m)$ and $\operatorname{Re}\Tr(\rho M_m X_j)$.


We generalize the method in the previous work \cite{Ringbauer2014} to bound the value of $\operatorname{Re}\Tr(\rho M_m X_j)$ from noisy experimetal data.
Specifically, by taking the probabilities ${p_l(m) = \Tr(\rho_l M_m)|l=1,2,3}$ as constraints on the POVM, ${M_m}$, we can bound the value of $\operatorname{Re}\Tr(\rho M_m X_j)$ within a small interval by computing its minimum and maximum, which can be efficiently computed via semi-definite programmings as
\begin{equation}
    \begin{aligned}
        \min /\max\ &\operatorname{Re}\Tr(\rho M_m X_j)\\
        \text{subject to}\ &\sum_m M_m=I,\ M_m\geq 0,\ \forall m\\
        &\mathcal{S}(\{\Tr(\rho_l M_m)=p_l(m),\ \forall m|l=1,2,3\})
    \end{aligned}
\end{equation}
here $\mathcal{S}(\mathcal{A})$ denotes the largest subset of $\mathcal{A}$ such that all the constraints therein are independent.

Using the method described, we can evaluate the errors for each observable using the following equations:
\begin{equation}\label{apdx:err_experiment}
    \begin{aligned}
        \alpha_{jm}^{\min/\max}=&\min/\max\ \operatorname{Re}\Tr(\rho M_m X_j),\\
        f_j^{\min/\max}(m)=&\alpha_{jm}^{\min/\max}/p_1(m),\\
        (\epsilon_j^{A})^2=&\Tr(\rho X_j^2)+\sum_m f_j^{\min}(m)^2p_1(m)\\&-2\sum_mf_j^{\min}(m)\alpha_{jm}^{\min},\\
        =&\Tr(\rho X_j^2)-\sum_m (\alpha_{jm}^{\min})^2/p_1(m)\\
        (\epsilon_j^{B})^2=&\Tr(\rho X_j^2)+\sum_m f_j^{\max}(m)^2p_1(m)\\&-2\sum_mf_j^{\max}(m)\alpha_{jm}^{\max},\\
        =&\Tr(\rho X_j^2)-\sum_m (\alpha_{jm}^{\max})^2/p_1(m).\\
    \end{aligned}
\end{equation}
Denoting $\epsilon_j^{\min/\max}=\min/\max\ \{\epsilon_j^{A},\epsilon_j^{B}\}$, the error $\epsilon_j$ is then bounded in a small interval $[\epsilon_j^{\min},\epsilon_j^{\max}]$ based on the experimental observations.
In our experiment on a qubit system, the interval is typically too small to be visible and the errors can be approximately determined as $\epsilon_j\approx\epsilon_j^{\min}\approx\epsilon_j^{\max}$.


\bigskip
\noindent\textbf{Characterization of the superconducting platform ``ScQ-P136''.}

\noindent
The experiments are performed on the ``ScQ-P136" backend of the Quafu cloud quantum computing platform \cite{Quafu3}.
The parameters of the used qubit are shown in Table.\ref{tab:ScQ-P136qubit}.
The parameters and architecture of the processor can be found in Table.\ref{tab:ScQ-P136} and Fig.\ref{fig:ScQ-P136}.

\begin{table}[t]
    \centering
    \setlength{\tabcolsep}{0.05mm}
    \begin{tabular}{c@{\hskip 2em}c}
        \toprule
        Qubit index & 1 \\[.06cm]
        Qubit frequency (GHz) & 4.524 \\[.06cm]
        Readout frequency (GHz) & 6.832 \\[.06cm]
        Anharmonicity (MHz) & 0.294 \\[.06cm]
        Relaxation time, T1 ($\mu s$) & 22.88 \\[.06cm]
        Coherence time, T2 ($\mu s$) & 19.9 \\[.06cm]
        \bottomrule
    \end{tabular}
    \caption{Parameters of the used qubit in ``ScQ-P136''.} 
    \label{tab:ScQ-P136qubit}
    \vspace{1em}
    \setlength{\tabcolsep}{0.05mm}
    \begin{tabular}{c@{\hskip 1em}c@{\hskip 1em}c@{\hskip 1em}c@{\hskip 1em}c}
        \toprule
        T1 ($\mu s$) & T2 ($\mu s$) & Fidelity$_{\mathrm{CNOT}}$ & Fidelity$_{\mathrm{Qubit}}$ \\
        \midrule
        Avg: 34.208 & Avg: 17.485 & Avg: 0.946 & ~ \\
        min: 15.17 & min: 0.95 & min: 0.83 & $>$0.99 \\
        max: 59.1 & max: 53.41 & max: 0.996 & ~ \\
        \bottomrule
    \end{tabular}
    \caption{Parameters of the quantum processor ``ScQ-P136''.} 
    \label{tab:ScQ-P136}
\end{table}

\begin{figure}
    \centering
    \includegraphics[width=0.47\textwidth]{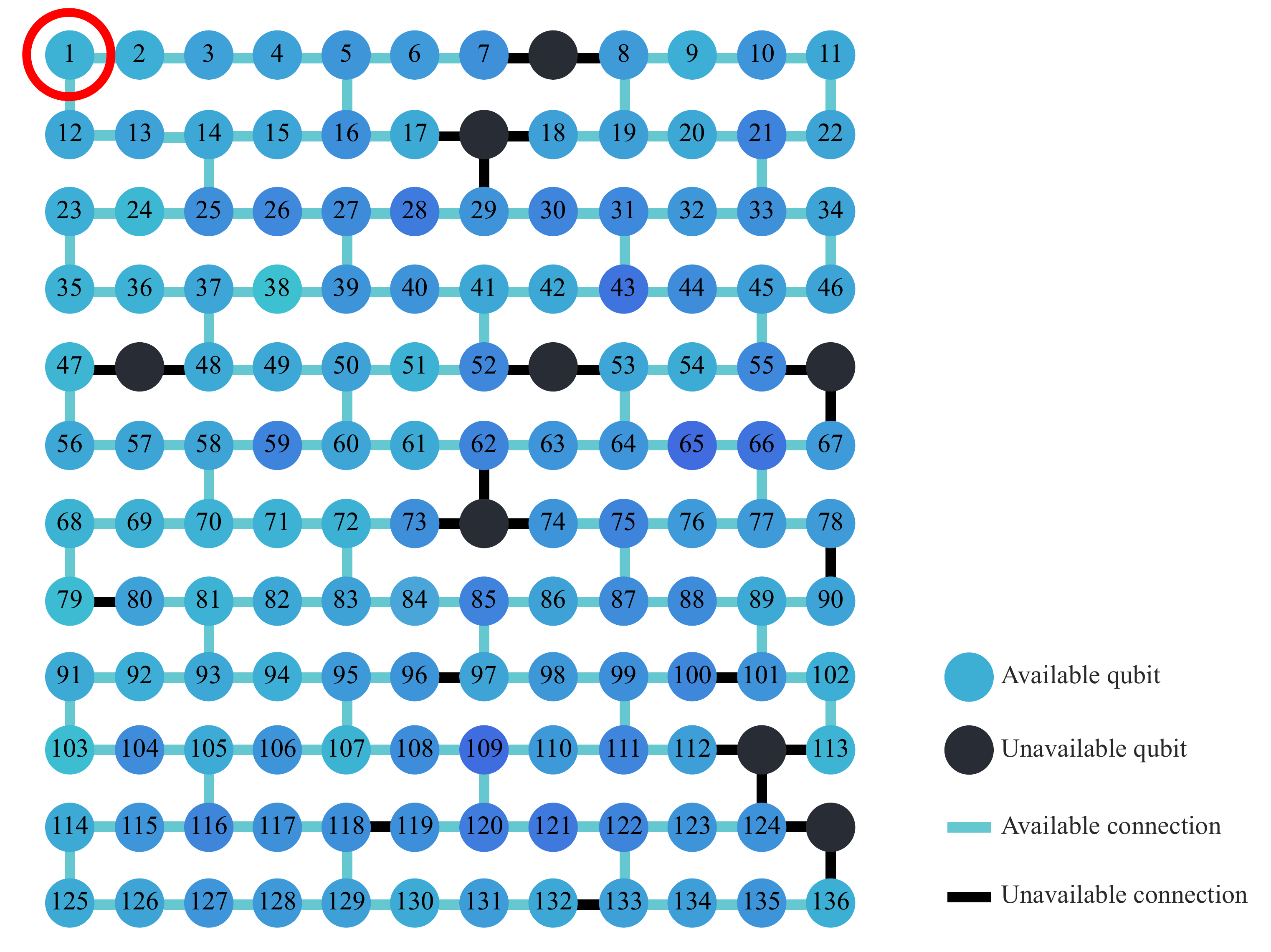}
    \caption{Architecture of the quantum processor ``ScQ-P136''. The used qubit is circled.}
    \label{fig:ScQ-P136}
\end{figure}

\bigskip

\noindent \textbf{Data availability}

\noindent
All relevant source codes are available from the authors upon request.

\noindent \textbf{Code availability}

\noindent
All relevant source codes are available from the authors upon request.

\noindent \textbf{Acknowledgements}\\
We acknowledge the use of Quafu cloud quantum computation platform for this work.
H.Y. acknowledges partial support from the Research Grants Council of Hong Kong with Grant No. 14307420, 14308019, 14309022, Ministry of Science and Technology, China (MOST2030 with Grant No 2023200300600, 2023200300603), the Guangdong Provincial Quantum Science Strategic Initiative (Grant No.GDZX2303007).
H.C. acknowledges the support from Shenzhen University with Grant No. 000001032510.

\noindent \textbf{Author contributions}

\noindent
H.Y. conceived the project. H.C. and H.Y. derived the theoretical results. H.C. and L.W. ran the Quafu experiment. H.C. and H.Y. wrote the manuscript. H.Y. supervised the project.

\noindent \textbf{Competing interests}

\noindent
The authors declare no competing interests.



\newpage
\setcounter{equation}{0}
\renewcommand{\thesection}{S\arabic{section}}
\renewcommand{\theequation}{S\arabic{equation}}
\renewcommand{\thefigure}{S\arabic{figure}}
\renewcommand{\thetable}{S\arabic{table}}
\begin{widetext}

\begin{center}
    \Large Supplementary Information for \\ Simultaneous Measurement of Multiple Incompatible Observables and Tradeoff in Multiparameter Quantum Estimation
\end{center}





\noindent
Section \ref{apdx:analyticaln}. Analytical error-tradeoff relation for multiple observables\\
Section \ref{apdx:errorasSDP}. Error-tradeoff relation with semidefinite programming\\
Section \ref{apdx:comparSDPanaly}. Comparison of the SDP and analytical bound\\
Section \ref{apdx:tightpure}. Tight error-tradeoff relation and optimal measurement for pure states\\
Section \ref{apdx:tightpuretwo}. Analytical error-tradeoff relation and optimal measurement for two observables on pure states\\
Section \ref{apdx:tightermixed}. Tighter analytical bound for two observables on mixed states\\
Section \ref{apdx:comparAnalyOzawa}. Comparison of the analytical bound for more than two observables with the Ozawa’s bound on pure states\\
Section \ref{apdx:comparSDPOzawa}. Comparison of the SDP bound for more than two observables with the Ozawa’s bound on mixed states\\
Section \ref{apdx:example}. Examples\\

\section{Analytical error-tradeoff relation for multiple observables}\label{apdx:analyticaln}
In this section we derive the tradeoff relation
\begin{equation}
    \Tr(S_{\text{Re}}^{-1}Q_{\text{Re}})\geq \left(\sqrt{\|S_{\text{Re}}^{-\frac{1}{2}}\tilde{S}_{\text{Im}}S_{\text{Re}}^{-\frac{1}{2}}\|_{F}+1}-1\right)^2.
\end{equation}

Recall given the observables $\{X_j|1\leq j\leq n\}$, we can construct $\{F_j=\sum_m f_j(m)V_m| 1\leq j\leq n\}$ to approximate the observables from a set of projective measurement $\{V_m\}$ on the extended space with an acilla. The mean squared error of the approximation is given by
\begin{equation}
    \epsilon_j^2 = \Tr\left[(F_j-X_j\otimes I)^2\left(\rho\otimes\sigma\right)\right],
\end{equation}
where $\sigma=|\xi_0\rangle\langle \xi_0|$ is a state of the ancilla.
We then let 
\begin{equation}
    \begin{aligned}
        \mathcal{A}_u=&\begin{pmatrix}
            \bra{u}\sqrt{\rho\otimes\sigma}E_1 \\ \vdots \\ \bra{u}\sqrt{\rho\otimes\sigma}E_n \\  \bra{u}\sqrt{\rho\otimes\sigma}(X_1\otimes I) \\ 
            \vdots \\ \bra{u}\sqrt{\rho\otimes\sigma}(X_n\otimes I)
        \end{pmatrix}\begin{pmatrix}
            \bra{u}\sqrt{\rho\otimes\sigma}E_1 \\ \vdots \\ \bra{u}\sqrt{\rho\otimes\sigma}E_n \\  \bra{u}\sqrt{\rho\otimes\sigma}(X_1\otimes I) \\ 
            \vdots \\ \bra{u}\sqrt{\rho\otimes\sigma}(X_n\otimes I)
        \end{pmatrix}^\dagger\\
        =&\begin{pmatrix}
            Q_u & R_u \\
            R_u^\dagger & S_u
        \end{pmatrix}
        \geq 0,
    \end{aligned}
\end{equation}
here $E_j=F_j-X_j\otimes I$, $|u\rangle$ is any vector, $Q_u$, $R_u$, $S_u$ are $n\times n$ matrices with the entries given by (note $E_j$ and $X_j$ are all Hermitian operators)
\begin{equation}
    \begin{aligned}
        (Q_u)_{jk}&=\bra{u}\sqrt{\rho\otimes\sigma}E_jE_k\sqrt{\rho\otimes\sigma}\ket{u},\\
        (R_u)_{jk}&=\bra{u}\sqrt{\rho\otimes\sigma}E_j(X_k\otimes I)\sqrt{\rho\otimes\sigma}\ket{u},\\
        (S_u)_{jk}&=\bra{u}\sqrt{\rho\otimes\sigma}(X_j\otimes I)(X_k\otimes I)\sqrt{\rho\otimes\sigma}\ket{u}.
    \end{aligned}
\end{equation}
We can write these matrices in terms of the real and imaginary parts as $Q_u=Q_{u,\text{Re}}+iQ_{u,\text{Im}}$, $R_u=R_{u,\text{Re}}+iR_{u,\text{Im}}$, $S_u=S_{u,\text{Re}}+iS_{u,\text{Im}}$, where
\begin{equation}
    \begin{aligned}
        (Q_{u,\text{Re}})_{jk}&=\frac{1}{2}\bra{u}\sqrt{\rho\otimes\sigma}\{E_j,E_k\}\sqrt{\rho\otimes\sigma}\ket{u},\\
        (R_{u,\text{Re}})_{jk}&=\frac{1}{2}\bra{u}\sqrt{\rho\otimes\sigma}\{E_j,X_k\otimes I\}\sqrt{\rho\otimes\sigma}\ket{u},\\
        (S_{u,\text{Re}})_{jk}&=\frac{1}{2}\bra{u}\sqrt{\rho\otimes\sigma}\{X_j\otimes I,X_k\otimes I\}\sqrt{\rho\otimes\sigma}\ket{u},
    \end{aligned}
\end{equation}
and
\begin{equation}
    \begin{aligned}
        (Q_{u,\text{Im}})_{jk}&=\frac{1}{2i}\bra{u}\sqrt{\rho\otimes\sigma}[E_j,E_k]\sqrt{\rho\otimes\sigma}\ket{u},\\
        (R_{u,\text{Im}})_{jk}&=\frac{1}{2i}\bra{u}\sqrt{\rho\otimes\sigma}[E_j,X_k\otimes I]\sqrt{\rho\otimes\sigma}\ket{u},\\
        (S_{u,\text{Im}})_{jk}&=\frac{1}{2i}\bra{u}\sqrt{\rho\otimes\sigma}[X_j\otimes I,X_k\otimes I]\sqrt{\rho\otimes\sigma}\ket{u}.
    \end{aligned}
\end{equation}
For any set of $\{\ket{u_q}\}$ with $\sum_q\ket{u_q}\bra{u_q}= I$, we can obtain a corresponding set of $\{\mathcal{A}_{u_q}\}$.
We then let $\tilde{{\mathcal{A}}}=\sum_q\tilde{\mathcal{A}}_{u_q}$ with each $\tilde{\mathcal{A}}_{u_q}$ equals to either $\mathcal{A}_{u_q}$ or $\mathcal{A}_{u_q}^{T}$.
Since $\mathcal{A}_{u_q}\geq 0$, and $\mathcal{A}_{u_q}^{T}\geq 0$, we have
\begin{equation}\label{apdx:Auq}
    \tilde{\mathcal{A}}=\sum_q\tilde{\mathcal{A}}_{u_q}=
    \begin{pmatrix}
        \tilde{Q} & \tilde{R}\\
        \tilde{R}^\dagger & \tilde{S}
    \end{pmatrix}\geq 0,
\end{equation}
where $\tilde{Q}=\sum_q \tilde{Q}_{u_q}$, $\tilde{R}=\sum_q \tilde{R}_{u_q}$, $\tilde{S}=\sum_q \tilde{S}_{u_q}$ with each $(\tilde{Q}_{u_q},\tilde{R}_{u_q}, \tilde{S}_{u_q})$ equals to either $(Q_{u_q},R_{u_q}, S_{u_q})$ or $(\bar{Q}_{u_q},\bar{R}_{u_q}, \bar{S}_{u_q})$, here $\bar{M}=M_{\text{Re}}-iM_{\text{Im}}$. 
We assume $\tilde{S}$ is nonsingular, which means $\{X_j\}$ are linearly independent.

Since $\tilde{\mathcal{A}}\geq 0$, with the Schur complement, we have
\begin{equation}
    \tilde{Q}-\tilde{R}\tilde{S}^{-1}\tilde{R}^{\dagger}\geq 0.
\end{equation}
This can be rewritten as $\tilde{Q}-(i\tilde{R})\tilde{S}^{-1}(i\tilde{R})^{\dagger}\geq 0$, which is the Schur complement of 
\begin{equation}
    \tilde{\mathcal{B}}=
    \begin{pmatrix}
        \tilde{Q} & i\tilde{R}\\
        (i\tilde{R})^{\dagger} & \tilde{S}
    \end{pmatrix}.
\end{equation}
Since $\tilde{Q}-(i\tilde{R})\tilde{S}^{-1}(i\tilde{R})^{\dagger}\geq 0$ and $\tilde{S}>0$, we then have $\tilde{\mathcal{B}}\geq 0$. This implies $\tilde{\mathcal{B}}_{\text{Re}}=\frac{1}{2}(\tilde{\mathcal{B}}+\tilde{\mathcal{B}}^{T})\geq 0$, which can be written as
\begin{equation}
    \tilde{\mathcal{B}}_{\text{Re}}=
    \begin{pmatrix}
        \tilde{Q}_{\text{Re}} & -\tilde{R}_{\text{Im}}\\
        -\tilde{R}_{\text{Im}}^{T} & \tilde{S}_{\text{Re}}
    \end{pmatrix}
    \geq 0.
\end{equation}
With the Schur complement of $\tilde{\mathcal{B}}_{\text{Re}}$ we have
\begin{equation}
    \tilde{Q}_{\text{Re}}-\tilde{R}_{\text{Im}}\tilde{S}_{\text{Re}}^{-1}\tilde{R}_{\text{Im}}^{T}\geq 0.
\end{equation}
This implies
\begin{equation}
    \tilde{S}_{\text{Re}}^{-\frac{1}{2}} \tilde{Q}_{\text{Re}}\tilde{S}_{\text{Re}}^{-\frac{1}{2}}-\tilde{S}_{\text{Re}}^{-\frac{1}{2}}\tilde{R}_{\text{Im}}\tilde{S}_{\text{Re}}^{-1}\tilde{R}_{\text{Im}}^{T}\tilde{S}_{\text{Re}}^{-\frac{1}{2}}\geq 0,
\end{equation}
thus
\begin{equation}\label{eq:QR}
    \aligned
    \Tr(\tilde{S}_{\text{Re}}^{-\frac{1}{2}} \tilde{Q}_{\text{Re}}\tilde{S}_{\text{Re}}^{-\frac{1}{2}})&\geq\Tr(\tilde{S}_{\text{Re}}^{-\frac{1}{2}}\tilde{R}_{\text{Im}}\tilde{S}_{\text{Re}}^{-1}\tilde{R}_{\text{Im}}^{T}\tilde{S}_{\text{Re}}^{-\frac{1}{2}})\\
    &=\|\tilde{S}_{\text{Re}}^{-\frac{1}{2}}\tilde{R}_{\text{Im}}\tilde{S}_{\text{Re}}^{-\frac{1}{2}}\|_{F}^2,
    \endaligned
\end{equation}
where $\|\cdot\|_F$ is the Frobenius norm.
Here the right side of the inequality depends on $\tilde{R}_{\text{Im}}$, which depends on the choice of the measurement. To get rid of $\tilde{R}_{\text{Im}}$, we use the fact that $[F_j,F_k]=0$ since they are constructed from the same projective measurement. $[F_j,F_k]=0$ can be rewritten as $[E_j+X_j\otimes  I, E_k+X_k\otimes I]=0$, which can be expanded as
\begin{equation}
    [E_j,E_k]+[E_j,X_k\otimes  I]+[X_j\otimes  I,E_k]+[X_j\otimes  I,X_k\otimes  I]=0.
\end{equation}
From which we have $Q_{u_q,\text{Im}}+R_{u_q,\text{Im}}-R_{u_q,\text{Im}}^{T}+S_{u_q,\text{Im}}=0$. 
Since $(\tilde{Q}_{u_q},\tilde{R}_{u_q}, \tilde{S}_{u_q})$ equal to either $(Q_{u_q},R_{u_q}, S_{u_q})$ or $(\bar{Q}_{u_q},\bar{R}_{u_q}, \bar{S}_{u_q})$, $(\tilde{Q}_{u_q,\text{Im}},\tilde{R}_{u_q,\text{Im}},\tilde{S}_{u_q,\text{Im}})$ then equal to either $(Q_{u_q,\text{Im}},R_{u_q,\text{Im}},S_{u_q,\text{Im}})$ or $(-Q_{u_q,\text{Im}},-R_{u_q,\text{Im}},-S_{u_q,\text{Im}})$. In either case
$\tilde{Q}_{u_q,\text{Im}}+\tilde{R}_{u_q,\text{Im}}-\tilde{R}_{u_q,\text{Im}}^{T}+\tilde{S}_{u_q,\text{Im}}=0$ holds, thus $\sum_q (\tilde{Q}_{u_q,\text{Im}}+\tilde{R}_{u_q,\text{Im}}-\tilde{R}_{u_q,\text{Im}}^{T}+\tilde{S}_{u_q,\text{Im}})=0$. This can be written as
\begin{equation}
    \tilde{Q}_{\text{Im}}+\tilde{R}_{\text{Im}}-\tilde{R}_{\text{Im}}^{T}+\tilde{S}_{\text{Im}}=0.
\end{equation}
By multiplying $\tilde{S}_{\text{Re}}^{-\frac{1}{2}}$ from both left and right, we have $\hat{Q}_{\text{Im}}+\hat{R}_{\text{Im}}-\hat{R}_{\text{Im}}^{T}+\hat{S}_{\text{Im}}=0$, which can be equivalently written as
\begin{equation}\label{eq:ineqImsupp}
    \hat{S}_{\text{Im}}=-\hat{Q}_{\text{Im}}-\hat{R}_{\text{Im}}+\hat{R}_{\text{Im}}^{T},
\end{equation}
where $\hat{\Box}=\tilde{S}_{\text{Re}}^{-\frac{1}{2}}\tilde{\Box}\tilde{S}_{\text{Re}}^{-\frac{1}{2}}$. Applying the triangle inequality we get
\begin{equation}
    \|\hat{S}_{\text{Im}}\|_F\leq \|\hat{Q}_{\text{Im}}\|_F+\|\hat{R}_{\text{Im}}\|_F+\|\hat{R}_{\text{Im}}^{T}\|_F.
\end{equation}
Note that Eq.(\ref{eq:QR}) can be rewritten as $\|\hat{R}_{\text{Im}}^{T}\|_F^2=\|\hat{R}_{\text{Im}}\|_F^2\leq \Tr(\hat{Q}_{\text{Re}})$,
and from $\hat{Q}\geq 0$ we have
\begin{equation}
    \Tr(\hat{Q}_{\text{Re}})\geq \|\hat{Q}_{\text{Im}}\|_1\geq \|\hat{Q}_{\text{Im}}\|_F.
\end{equation}
Combining these inequalities we then have
\begin{equation}
    \begin{aligned}
        \|\hat{S}_{\text{Im}}\|_F&\leq \|\hat{Q}_{\text{Im}}\|_F+\|\hat{R}_{\text{Im}}\|_F+\|\hat{R}_{\text{Im}}^{T}\|_F\\
        &\leq \Tr(\hat{Q}_{\text{Re}})+2\sqrt{\Tr(\hat{Q}_{\text{Re}})}.
    \end{aligned}
\end{equation}
This leads to the bound on $\Tr(\hat{Q}_{\text{Re}})$ as
\begin{equation}
    \Tr(\hat{Q}_{\text{Re}})\geq \left(\sqrt{	\|\hat{S}_{\text{Im}}\|_F+1}-1\right)^2,
\end{equation}
which can be equivalently written as
\begin{equation}\label{eq:multi_obs_bound}
    \Tr(S_{\text{Re}}^{-1}Q_{\text{Re}})\geq \left(\sqrt{\|S_{\text{Re}}^{-\frac{1}{2}}\tilde{S}_{\text{Im}}S_{\text{Re}}^{-\frac{1}{2}}\|_{F}+1}-1\right)^2.
\end{equation}

We note that this tradeoff relation is invariant under linear transformations of the observables, i.e., if we let $Y_j= \sum_k b_{jk}X_k$ with $b_{jk}$ as the $jk$-th entry of a non-singular real matrix $B$, then $Q_{\text{Re}}(Y)=B^TQ_{\text{Re}}(X)B$, $\tilde{S}(Y)=B^T\tilde{S}(x)B$, we then have $\Tr[S_{\text{Re}}^{-1}(Y)Q_{\text{Re}}(Y)]=\Tr[S_{\text{Re}}^{-1}(X)Q_{\text{Re}}(X)]$ and $\|S_{\text{Re}}^{-\frac{1}{2}}(Y)\tilde{S}_{\text{Im}}(Y)S_{\text{Re}}^{-\frac{1}{2}}(Y)\|_{F}=\|S_{\text{Re}}^{-\frac{1}{2}}(X)\tilde{S}_{\text{Im}}(X)S_{\text{Re}}^{-\frac{1}{2}}(X)\|_{F}$. The tradeoff relation thus remain as the same under linear transformation and it is often convenient to choose a linear transformation under which $S_{\text{Re}}=I$.

\section{Error-tradeoff relation with semidefinite programming}\label{apdx:errorasSDP}

Recall that given a set of observables $\{X_j\}_{j=1}^n$ in Hilbert space $\mathcal{H}_S$, we can construct $\{F_j=\sum_m f_j(m)V_m\}_{j=1}^n$ to approximate the observables from a set of projective measurement $\{V_m\}$ on the extended space $\mathcal{H}_S\otimes\mathcal{H}_A$.
We denote the weighted mean squared error of the approximations as
\begin{equation}
    \mathcal{E}=\sum_{j=1}^n w_j \epsilon_j^2=\sum_{j=1}^n w_j\Tr[(\rho\otimes\sigma)(F_j-X_j\otimes I)^2],
\end{equation}
where $\rho\in\mathcal{H}_S$ is the state of the system and $\sigma=\ket{\xi_0}\bra{\xi_0}\in\mathcal{H}_A$ is the state of the ancilla.
To formularize the minimization of $\mathcal{E}$ over all choices of $\{F_j\}_{j=1}^n$, we first define an $n\times n$ Hermitian matrix $Q$, whose $jk$-th element is given by
\begin{equation}
    \begin{aligned}
        Q_{jk}=&\Tr\left[(\rho\otimes\sigma)(F_j-X_j\otimes I)(F_k-X_k\otimes I)\right]\\
        =&\Tr\left[(\rho\otimes\sigma)F_jF_k\right]-\Tr\left[(\rho\otimes\sigma)F_j(X_k\otimes I)\right]
        -\Tr\left[(\rho\otimes\sigma)(X_j\otimes I)F_k\right]
        +\Tr\left[(\rho\otimes\sigma)(X_j\otimes I)(X_k\otimes I)\right]\\
        =&\Tr\left[\rho \sum_m f_j(m)M_m f_k(m)\right]-\Tr\left(\rho R_jX_k\right)
        -\Tr\left(\rho X_j R_k\right)+\Tr\left(\rho X_j X_k\right),
    \end{aligned}
\end{equation}
where $R_j=(I\otimes\bra{\xi_0})F_j(I\otimes\ket{\xi_0})=\sum_m f_j(m)M_m$ for $1\le j\le n$ are Hermitian matrices in $\mathcal{H}_S$,  $\{M_m=(I\otimes\bra{\xi_0})V_m(I\otimes\ket{\xi_0})\}$ forms a POVM.
We denote $\mathbb{S}_{jk}=\sum_{m}f_j(m)M_m f_k(m)$, $\mathbb{R}=\begin{pmatrix}R_1 & R_2 & \cdots & R_n\end{pmatrix}^{\dagger}$ and $\mathbb{X}=\begin{pmatrix}X_1 & X_2 & \cdots & X_n\end{pmatrix}^{\dagger}$,
the Hermitian matrix $Q$ can then be written as $Q=\Tr_S\left[(I_n\otimes\rho)(\mathbb{S}-\mathbb{R}\mathbb{X}^{\dagger}-\mathbb{X}\mathbb{R}^{\dagger}+\mathbb{X}\mathbb{X}^{\dagger})\right]$, and
\begin{equation}
    \begin{aligned}
        \mathcal{E}=&\Tr(WQ)=\Tr\left[(W\otimes\rho)(\mathbb{S}-\mathbb{R}\mathbb{X}^{\dagger}-\mathbb{X}\mathbb{R}^{\dagger}+\mathbb{X}\mathbb{X}^{\dagger})\right],
    \end{aligned}
\end{equation}
where $\mathbb{S}$ is an $n\times n$ block matrix with the $jk$-th block as $\mathbb{S}_{jk}$, which is itself a Hermitian operator, $W\geq 0$ is a weight matrix, which is typically taken as $W=\mathrm{diag}\{w_1,\cdots, w_n\}$ , but it can also be taken as a general positive semidefinite matrix.

We now show that for the defined $\mathbb{S}$ and $\mathbb{R}$, we have $\mathbb{S}\geq \mathbb{R}\mathbb{R}^{\dagger}$. Since $\mathbb{S}$ and $\mathbb{R}\mathbb{R}^{\dagger}$ are $nd_s\times nd_s$ matrices, where $d_S$ is the dimension of $\mathcal{H}_S$, $\mathbb{S}\geq \mathbb{R}\mathbb{R}^{\dagger}$ is equivalent to $b^\dagger(\mathbb{S}- \mathbb{R}\mathbb{R}^{\dagger})b\geq 0$ for any $nd_S\times 1$ vector $b$. 

For any $nd_S\times 1$ vector, $b$, we can write $b^\dagger$ as $b^\dagger=\begin{pmatrix}\bra{b_1} & \bra{b_2} & \cdots & \bra{b_n}\end{pmatrix}$ with $\{\bra{b_j}\}_{j=1}^n$ being $d_S$-dimensional row vectors. We then have
    \begin{equation}
        \begin{aligned}
            &b^{\dagger}\mathbb{S}b-b^{\dagger}\mathbb{R}\mathbb{R}^{\dagger}b=\sum_{j,k} \bra{b_j} (\mathbb{S}_{jk}-R_jR_k) \ket{b_k}\\
            =&\sum_{j,k} \bra{b_j} \left[\sum_{m}f_j(m)M_m f_k(m)-\left(\sum_m f_j(m)M_m\right)\left(\sum_l f_k(l)M_l\right)\right] \ket{b_k}\\
            =&\sum_m \left(\sum_j\bra{b_j}f_j(m)\right) M_m \left(\sum_k f_k(m)\ket{b_k}\right)-  \left[\sum_m \left(\sum_j \bra{b_j} f_j(m)\right)M_m\right]\left[\sum_l M_l  \left(\sum_k f_k(l) \ket{b_k}\right)\right]\\
            =&\sum_m \left[\sum_j\bra{b_j}f_j(m)-\sum_l \left(\sum_j \bra{b_j} f_j(l)\right)M_l\right]M_m\left[\sum_k f_k(m)\ket{b_k}-\sum_l M_l  \left(\sum_k f_k(l) \ket{b_k}\right) \right]
            \\
            =&\sum_m V_m^{\dagger}M_mV_m\geq 0,
        \end{aligned}
    \end{equation}
    where $V_m=\sum_k f_k(m)\ket{b_k}-\sum_l M_l  \left(\sum_k f_k(l) \ket{b_k}\right)=\sum_j f_j(m)\ket{b_j}-\sum_l M_l  \left(\sum_j f_j(l) \ket{b_j}\right)$. We thus have $\mathbb{S}\geq \mathbb{R}\mathbb{R}^{\dagger}$.

    Denote $\{R_j^{\star}\}_{j=1}^n$ and $\mathbb{S}^{\star}$ as the optimal operators that gives the minimal error, $\mathcal{E}^{\star}$, 
    we then have
    \begin{equation}
        \begin{aligned}
            \mathcal{E}\geq \mathcal{E}^{\star}&=\Tr\left[(W\otimes\rho)(\mathbb{S}^{\star}-\mathbb{R}^{\star}\mathbb{X}^{\dagger}-\mathbb{X}\mathbb{R}^{\star\dagger}+\mathbb{X}\mathbb{X}^{\dagger})\right]\\
            &\geq \min_{\mathbb{S}} \left\{\Tr\left[(W\otimes\rho)(\mathbb{S}-\mathbb{R}^{\star}\mathbb{X}^{\dagger}-\mathbb{X}\mathbb{R}^{\star\dagger}+\mathbb{X}\mathbb{X}^{\dagger})\right]|\mathbb{S}_{jk}=\mathbb{S}_{jk}^\dagger=\mathbb{S}_{kj},\mathbb{S}\geq \mathbb{R}^{\star}\mathbb{R}^{\star\dagger}\right\}\\
            &\geq \min_{\mathbb{S},\{R_j\}_{j=1}^n} \left\{\Tr\left[(W\otimes\rho)(\mathbb{S}-\mathbb{R}\mathbb{X}^{\dagger}-\mathbb{X}\mathbb{R}^{\dagger}+\mathbb{X}\mathbb{X}^{\dagger})\right]|\mathbb{S}_{jk}=\mathbb{S}_{jk}^\dagger=\mathbb{S}_{kj},\mathbb{S}\geq \mathbb{R}\mathbb{R}^{\dagger}, R_j=R_j^{\dagger}\right\}.
        \end{aligned}
    \end{equation}
    And the minimization can be formulated as a semidefinite programming as
    \begin{equation}\label{apdx:SDP0}
        \begin{aligned}
            \mathcal{E}_0=\min_{\mathbb{S},\{R_j\}_{j=1}^n} &\Tr\left[(W\otimes\rho)(\mathbb{S}-\mathbb{R}\mathbb{X}^{\dagger}-\mathbb{X}\mathbb{R}^{\dagger}+\mathbb{X}\mathbb{X}^{\dagger})\right]\\
            \text{subject to}\quad &\mathbb{S}_{jk}=\mathbb{S}_{kj}=\mathbb{S}_{jk}^{\dagger},\ \forall j,k\\
            &R_j=R_j^{\dagger},\ \forall j\\
            &\begin{pmatrix}
                I & \mathbb{R}^{\dagger}\\
                \mathbb{R} & \mathbb{S}
            \end{pmatrix}\geq 0.
        \end{aligned}
    \end{equation}

\section{Comparison of the SDP and analytical bound}\label{apdx:comparSDPanaly}

In this section we show that the bound $\mathcal{E}_0$ is tighter than the analytical bounds with any choice of $\{|u_q\rangle\}$.

Recall that for any vector $\ket{u}$ in $\mathcal{H}_S\otimes\mathcal{H}_A$, we have
\begin{equation}\label{apdx:schur}
    \begin{aligned}
        \mathcal{A}_u
        =&\begin{pmatrix}
            \bra{u}\sqrt{\rho\otimes\sigma}(F_1-X_1\otimes I) \\ \vdots \\ \bra{u}\sqrt{\rho\otimes\sigma}(F_n-X_n\otimes I) \\  \bra{u}\sqrt{\rho\otimes\sigma}(X_1\otimes I) \\ 
            \vdots \\ \bra{u}\sqrt{\rho\otimes\sigma}(X_n\otimes I)
        \end{pmatrix}
        \begin{pmatrix}
            \bra{u}\sqrt{\rho\otimes\sigma}(F_1-X_1\otimes I) \\ \vdots \\ \bra{u}\sqrt{\rho\otimes\sigma}(F_n-X_n\otimes I) \\  \bra{u}\sqrt{\rho\otimes\sigma}(X_1\otimes I) \\ 
            \vdots \\ \bra{u}\sqrt{\rho\otimes\sigma}(X_n\otimes I)
        \end{pmatrix}^\dagger\\
        =&\begin{pmatrix}
            Q_u & D_u \\
            D_u^\dagger & J_u
        \end{pmatrix}
        \geq 0,
    \end{aligned}
\end{equation}
where $Q_u$, $D_u$, $J_u$ are $n\times n$ matrices with the $jk$-th elements given by
\begin{equation}
    \begin{aligned}
        (Q_u)_{jk}&=\bra{u}\sqrt{\rho\otimes\sigma}(F_j-X_j\otimes I)(F_k-X_k\otimes I)\sqrt{\rho\otimes\sigma}\ket{u},\\
        (D_u)_{jk}&=\bra{u}\sqrt{\rho\otimes\sigma}(F_j-X_j\otimes I)(X_k\otimes I)\sqrt{\rho\otimes\sigma}\ket{u},\\
        (J_u)_{jk}&=\bra{u}\sqrt{\rho\otimes\sigma}(X_j\otimes I)(X_k\otimes I)\sqrt{\rho\otimes\sigma}\ket{u}.
    \end{aligned}
\end{equation}
Note that $[F_j,F_k]=0$, which gives
\begin{equation}\label{apdx:commtcond}
(Q_u+D_u+D_u^{\dagger}+J_u)_{jk}=(Q_u+D_u+D_u^{\dagger}+J_u)_{kj},
\end{equation}
i.e., the matrix $Q_u+D_u+D_u^{\dagger}+J_u$ is real symmetric.
For any set of vectors $\{\ket{u_q}\}$ that satistifies $\sum_q\ket{u_q}\bra{u_q}=I$, we have $Q=\sum_q Q_{u_q}$.
Together with Eq.(\ref{apdx:schur}) and Eq.(\ref{apdx:commtcond}), we can write the bound as a minimization,
\begin{equation}
\begin{aligned}
    \mathcal{E}_u=\min_{\{F_j\}}\ &\Tr[WQ]=\sum_q\Tr[WQ_{u_q}]\\
    \text{subject to}\quad &\begin{pmatrix}
        Q_{u_q} & D_{u_q} \\
        D_{u_q}^\dagger & J_{u_q}
    \end{pmatrix}
    \geq 0,\forall q\\
    &Q_{u_q}+D_{u_q}+D_{u_q}^{\dagger}+J_{u_q}\ \text{real symmetric},\forall q.
\end{aligned}
\end{equation}
$\mathcal{E}_u$ is tighter than any analytical bounds based on it, unless any other potential constraints could be introduced. In particular it is tighter than the analytical bound obtained in Appendix \ref{apdx:analyticaln}.
We will then show $\mathcal{E}_0$ is tighter than $\mathcal{E}_u$ for any $\{\ket{u_q}\}$.

We first formulate $\mathcal{E}_u$ as a semidefinite programming for a fixed $\{\ket{u_q}\}$, then show that $\mathcal{E}_0\geq\mathcal{E}_u$ for any $\{\ket{u_q}\}$.
First note that since $\sigma=\ket{\xi_0}\bra{\xi_0}$, we have 
\begin{eqnarray}
\aligned
&\sqrt{\rho\otimes\sigma}\ket{u_q}\bra{u_q}\sqrt{\rho\otimes\sigma}\\
&=(\sqrt{\rho}\otimes |\xi_0\rangle) (I\otimes \langle \xi_0|)\ket{u_q}\bra{u_q}(I\otimes |\xi_0\rangle) (\sqrt{\rho}\otimes \langle \xi_0|)\\
&=(\sqrt{\rho}\otimes |\xi_0\rangle) \tilde{\rho}(\sqrt{\rho}\otimes \langle \xi_0|)\\
&=\sqrt{\rho}\tilde{\rho}\sqrt{\rho}\otimes |\xi_0\rangle\langle \xi_0|\\
&=\rho_{u_q}\otimes\sigma,    
\endaligned
\end{eqnarray}
here $\tilde{\rho}=(I\otimes \langle \xi_0|)\ket{u_q}\bra{u_q}(I\otimes |\xi_0\rangle)$ and $\rho_{u_q}=\sqrt{\rho}\tilde{\rho}\sqrt{\rho}$ . 
We note that $\sum_q\rho_{u_q}=\rho$, and each $\rho_{u_q}$ is unnormalized with rank equal to 1. 
With $\mathbb{S}_{jk}=\sum_{m}f_j(m)M_m f_k(m)$, $\mathbb{R}=\begin{pmatrix}R_1 & R_2 & \cdots & R_n\end{pmatrix}^{\dagger}$ and $\mathbb{X}=\begin{pmatrix}X_1 & X_2 & \cdots & X_n\end{pmatrix}^{\dagger}$, we have $Q_{u_q}=\Tr_S\left[(I_n\otimes\rho_{u_q})(\mathbb{S}-\mathbb{R}\mathbb{X}^{\dagger}-\mathbb{X}\mathbb{R}^{\dagger}+\mathbb{X}\mathbb{X}^{\dagger})\right]$, $D_{u_q}=\Tr_S\left[(I_n\otimes\rho_{u_q})(\mathbb{R}\mathbb{X}^{\dagger}-\mathbb{X}\mathbb{X}^{\dagger})\right]$ and $J_{u_q}=\Tr_S\left[(I_n\otimes\rho_{u_q})\mathbb{X}\mathbb{X}^{\dagger}\right]$.
With Schur complement, $\begin{pmatrix}
Q_{u_q} & D_{u_q} \\
D_{u_q}^\dagger & J_{u_q}
\end{pmatrix}\geq 0$ is equivalent to $Q_{u_q}\geq D_{u_q}J_{u_q}^{-1}D_{u_q}^\dagger$, which further gives
\begin{equation}
    \begin{aligned}
        &Q_{u_q}+(D_{u_q}+J_{u_q})+(D_{u_q}+J_{u_q})^{\dagger}-J_{u_q} \geq (D_{u_q}+J_{u_q})J_{u_q}^{-1}(D_{u_q}+J_{u_q})^\dagger\\
        \Leftrightarrow\quad &\Tr_S\left[(I_n\otimes\rho_{u_q})\mathbb{S}\right] \geq \Tr_S\left[(I_n\otimes\rho_{u_q})\mathbb{R}\mathbb{X}^{\dagger}\right] \Tr_S\left[(I_n\otimes\rho_{u_q})\mathbb{X}\mathbb{X}^{\dagger}\right]^{-1} \Tr_S\left[(I_n\otimes\rho_{u_q})\mathbb{R}\mathbb{X}^{\dagger}\right]^{\dagger}\\
        \Leftrightarrow\quad &\begin{pmatrix}
            \Tr_S\left[(I_n\otimes\rho_{u_q})\mathbb{S}\right] & \Tr_S\left[(I_n\otimes\rho_{u_q})\mathbb{R}\mathbb{X}^{\dagger}\right] \\
            \Tr_S\left[(I_n\otimes\rho_{u_q})\mathbb{R}\mathbb{X}^{\dagger}\right]^\dagger & \Tr_S\left[(I_n\otimes\rho_{u_q})\mathbb{X}\mathbb{X}^{\dagger}\right]
        \end{pmatrix}\geq 0\\
        \Leftrightarrow\quad &\Tr_{S}\left[(I_2\otimes I_n\otimes\rho_{u_q})\begin{pmatrix}
            \mathbb{S} & \mathbb{R}\mathbb{X}^{\dagger} \\
            \mathbb{X}\mathbb{R}^\dagger & \mathbb{X}\mathbb{X}^{\dagger}
        \end{pmatrix}\right]\geq 0.
    \end{aligned}
\end{equation}
Also note that the constraint of $Q_{u_q}+D_{u_q}+D_{u_q}^{\dagger}+J_{u_q}$ being real symmetric is equivalent to $\Tr_S\left[(I_n\otimes\rho_{u_q})\mathbb{S}\right]$ being real symmetric.
$\mathcal{E}_u$ can then be formulated as a semidefinite programming for a fixed $\{\ket{u_q}\}$ as
\begin{equation}
\begin{aligned}
    \mathcal{E}_u=\min_{\mathbb{S},\{R_j\}_{j=1}^n} &\Tr\left[(W\otimes\rho)(\mathbb{S}-\mathbb{R}\mathbb{X}^{\dagger}-\mathbb{X}\mathbb{R}^{\dagger}+\mathbb{X}\mathbb{X}^{\dagger})\right]\\
    \text{subject to}\quad &\Tr_{S}\left[(I_2\otimes I_n\otimes\rho_{u_q})\begin{pmatrix}
        \mathbb{S} & \mathbb{R}\mathbb{X}^{\dagger} \\
        \mathbb{X}\mathbb{R}^\dagger & \mathbb{X}\mathbb{X}^{\dagger}
    \end{pmatrix}\right]\geq 0,\forall q\\
    &\Tr_S\left[(I_n\otimes\rho_{u_q})\mathbb{S}\right]\ \text{real symmetric}, \forall q\\
    &R_j=R_j^{\dagger}.\\
\end{aligned}
\end{equation}

Next we show that $\mathcal{E}_0$ is tighter than $\mathcal{E}_u$ for any $\{\ket{u_q}\}$.
Since $\begin{pmatrix}
\mathbb{R}\mathbb{R}^{\dagger} & \mathbb{R}\mathbb{X}^{\dagger} \\
\mathbb{X}\mathbb{R}^\dagger & \mathbb{X}\mathbb{X}^{\dagger}
\end{pmatrix}=\begin{pmatrix}
\mathbb{R}\\
\mathbb{X}
\end{pmatrix}\begin{pmatrix}
\mathbb{R}\\
\mathbb{X}
\end{pmatrix}^{\dagger}\geq 0$, combining with the constraint $\mathbb{S}\geq \mathbb{R}\mathbb{R}^{\dagger}$ in $\mathcal{E}_0$, we can obtain $\begin{pmatrix}
\mathbb{S} & \mathbb{R}\mathbb{X}^{\dagger} \\
\mathbb{X}\mathbb{R}^\dagger & \mathbb{X}\mathbb{X}^{\dagger}
\end{pmatrix}\geq 0$, which is stronger than the first constraint in $\mathcal{E}_u$.
From the constraint $\mathbb{S}_{jk}=\mathbb{S}_{kj}=\mathbb{S}_{jk}^{\dagger}$, $\forall j,k$ in $\mathcal{E}_0$, we can also obtain that $\Tr_S\left[(I_n\otimes\rho_{u_q})\mathbb{S}\right]$ is real symmetric $\forall q$.
The semidefinite programming $\mathcal{E}_0$ thus has the same objective function but stronger constraints than those in $\mathcal{E}_u$, which indicates that $\mathcal{E}_0\geq \mathcal{E}_u$.

\section{Tight error-tradeoff relation and optimal measurement for pure states}\label{apdx:tightpure}

In this section, we show that when $\rho$ is a pure state, $\mathcal{E}\geq \mathcal{E}_0$ is tight. Additionally we provide an explicit construction of the optimal measurement that saturates the bound.

Given any pure state $\rho=\ket{\psi}\bra{\psi}$, we have
\begin{equation}
    \begin{aligned}
        &\Tr\left[(W\otimes\rho)(\mathbb{S}-\mathbb{R}\mathbb{X}^{\dagger}-\mathbb{X}\mathbb{R}^{\dagger}+\mathbb{X}\mathbb{X}^{\dagger})\right]
        =\Tr\left[W\left(S-R^{\dagger}X-X^{\dagger}R+X^{\dagger}X\right)\right],
    \end{aligned}
\end{equation}
where $S$ is an $n\times n$ matrix with its $jk$th element given as $S_{jk}=\bra{\psi}\mathbb{S}_{jk}\ket{\psi}$, $R=\begin{pmatrix}\ket{r_1} & \ket{r_2} & \cdots & \ket{r_n}\end{pmatrix}$ and $X=\begin{pmatrix}\ket{x_1} & \ket{x_2} & \cdots & \ket{x_n}\end{pmatrix}$ are $d_S\times n$ matrices with $\ket{r_j}=R_j\ket{\psi}$, $\ket{x_j}=X_j\ket{\psi}$ for $1\le j\le n$.
$\mathbb{S}_{jk}=\mathbb{S}_{kj}=\mathbb{S}_{jk}^{\dagger}$ indicates that $S$ is real symmetric.
From $\mathbb{S}\geq \mathbb{R}\mathbb{R}^{\dagger}$, we have $S\geq R^{\dagger}R=(R^{\dagger}R)_{\mathrm{Re}}+i(R^{\dagger}R)_{\mathrm{Im}}$.
By taking transpose on both sides, we also have $S^{T}=S\geq (R^{\dagger}R)_{\mathrm{Re}}-i(R^{\dagger}R)_{\mathrm{Im}}$, thus $W^{\frac{1}{2}}(S-(R^{\dagger}R)_{\mathrm{Re}})W^{\frac{1}{2}}\geq\pm iW^{\frac{1}{2}}(R^{\dagger}R)_{\mathrm{Im}}W^{\frac{1}{2}}$ and we get
\begin{equation}
    \Tr(WS)\geq \Tr(W(R^{\dagger}R)_{\mathrm{Re}})+\Tr\left|W^{\frac{1}{2}}(R^{\dagger}R)_{\mathrm{Im}}W^{\frac{1}{2}}\right|,
\end{equation}
where the inequality can be saturated if $S=(R^{\dagger}R)_{\mathrm{Re}}+W^{-\frac{1}{2}}\left|W^{\frac{1}{2}}(R^{\dagger}R)_{\mathrm{Im}}W^{\frac{1}{2}}\right|W^{-\frac{1}{2}}$.
Thus we can further simplify $\mathcal{E}_0$ for pure states as
\begin{equation}
    \begin{aligned}
        \mathcal{E}_0&=\min_{S,R}\left\{\Tr\left[W\left(S-R^{\dagger}X-X^{\dagger}R+X^{\dagger}X\right)\right]|S\ \text{real symmetric}, S\geq R^{\dagger}R\right\}\\
        &=\min_{R}\left\{\Tr\left[W\left((R^{\dagger}R)_{\mathrm{Re}}-R^{\dagger}X-X^{\dagger}R+X^{\dagger}X\right)\right]+\Tr\left|W^{\frac{1}{2}}(R^{\dagger}R)_{\mathrm{Im}}W^{\frac{1}{2}}\right|\right\}.
    \end{aligned}
\end{equation}
The computation of $\mathcal{E}_0$ can be largely simplified by the following lemma.
\begin{lemma}
    $\mathcal{E}_0=\mathcal{E}_0'$, where
    \begin{equation}
        \mathcal{E}_0'=\min_{R}\left\{\Tr\left[W\left((R^{\dagger}R)_{\mathrm{Re}}-R^{\dagger}X-X^{\dagger}R+X^{\dagger}X\right)\right]|(R^{\dagger}R)_{\mathrm{Im}}=0\right\}.
    \end{equation}
\end{lemma}
\begin{proof}
    Denote $\mathcal{X}$ as the linear span of $\{\ket{\psi},\ket{x_1},...,\ket{x_n}\}$ and $\mathcal{X}^{\bot}$ as its orthogonal complement.
    Each $\ket{r_j}$ can then be decomposed as $\ket{r_j}=\ket{a_j}+\ket{b_j}$ with $\ket{a_j}\in\mathcal{X}$ and $\ket{b_j}\in\mathcal{X}^{\bot}$.
    Define $A=\begin{pmatrix}\ket{a_1} & \ket{a_2} & \cdots & \ket{a_n}\end{pmatrix}$ and $B=\begin{pmatrix}\ket{b_1} & \ket{b_2} & \cdots & \ket{b_n}\end{pmatrix}$, we then have $R=A+B$ and $R^{\dagger}R=A^{\dagger}A+B^{\dagger}B$.
    Next we consider the minimization in $\mathcal{E}_0'$ over $B$ with any fixed $A$.
    The Lagrangian of the minimization can be written as
    \begin{equation}
        \begin{aligned}
            \mathcal{L}(B,\Lambda)&=\Tr\left[W\left((R^{\dagger}R)_{\mathrm{Re}}-R^{\dagger}X-X^{\dagger}R+X^{\dagger}X\right)\right]+\Tr\left[(R^{\dagger}R)_{\mathrm{Im}}\Lambda\right]\\
            &=\Tr\left[W\left((A^{\dagger}A)_{\mathrm{Re}}+(B^{\dagger}B)_{\mathrm{Re}}-A^{\dagger}X-X^{\dagger}A+X^{\dagger}X\right)\right]+\Tr\left[(A^{\dagger}A)_{\mathrm{Im}}\Lambda\right]+\Tr\left[(B^{\dagger}B)_{\mathrm{Im}}\Lambda\right].
        \end{aligned}
    \end{equation}
    Here since $(R^{\dagger}R)_{\mathrm{Im}}$ is antisymmetric, we have $\Tr\left[(R^{\dagger}R)_{\mathrm{Im}}\Lambda\right]=\Tr\left[(R^{\dagger}R)_{\mathrm{Im}}^{T}\Lambda^{T}\right]=-\Tr\left[(R^{\dagger}R)_{\mathrm{Im}}\Lambda^{T}\right]$, which further gives $\Tr\left[(R^{\dagger}R)_{\mathrm{Im}}\Lambda\right]=\Tr\left[(R^{\dagger}R)_{\mathrm{Im}}(\Lambda-\Lambda^{T})/2\right]$, where $(\Lambda-\Lambda^{T})/2$ is the antisymmetric part of $\Lambda$.
    Thus without loss of generality we assume $\Lambda$ is antisymmetric.
    
    The matrix derivative of $\mathcal{L}(B)$ with respect to $B$ can be calculated as
    \begin{equation}
        \begin{aligned}
            \partial_{B}\mathcal{L}(B,\Lambda)&=\lim_{\varepsilon\rightarrow 0}\frac{\mathcal{L}(B+\varepsilon\cdot\delta B,\Lambda)-\mathcal{L}(B,\Lambda)}{\varepsilon}\\
            &=\Tr\left[W(B^{\dagger}\delta B+\delta B^{\dagger} B)_{\mathrm{Re}}+(B^{\dagger}\delta B+\delta B^{\dagger} B)_{\mathrm{Im}}\Lambda\right]\\
            &=\Tr\left[W(B^{\dagger}\delta B+\delta B^{\dagger} B)_{\mathrm{Re}}-(iB^{\dagger}\delta B+i\delta B^{\dagger} B)_{\mathrm{Re}}\Lambda\right]\\
            &=2\mathrm{Re}\Tr\left[\delta B^{\dagger}B(W-i\Lambda)\right].
        \end{aligned}
    \end{equation}
    Then we let $\partial_{B}\mathcal{L}(B,\Lambda)=0$ for arbitrary $\delta B$ and get $B(W-i\Lambda)=0$, which further gives
    \begin{equation}
        W^{\frac{1}{2}}B^{\dagger}BW^{\frac{1}{2}}=iW^{\frac{1}{2}}B^{\dagger}BW^{\frac{1}{2}}W^{-\frac{1}{2}}\Lambda W^{-\frac{1}{2}}.
    \end{equation}
    By taking the real part and imaginary part on both sides, we have
    \begin{equation}\label{apdx:solveB}
        \begin{aligned}
            W^{\frac{1}{2}}(B^{\dagger}B)_{\mathrm{Re}}W^{\frac{1}{2}}&=-W^{\frac{1}{2}}(B^{\dagger}B)_{\mathrm{Im}}W^{\frac{1}{2}}W^{-\frac{1}{2}}\Lambda W^{-\frac{1}{2}},\\
            W^{\frac{1}{2}}(B^{\dagger}B)_{\mathrm{Im}}W^{\frac{1}{2}}&=W^{\frac{1}{2}}(B^{\dagger}B)_{\mathrm{Re}}W^{\frac{1}{2}}W^{-\frac{1}{2}}\Lambda W^{-\frac{1}{2}}.
        \end{aligned}
    \end{equation}
    Note that $W^{\frac{1}{2}}(B^{\dagger}B)_{\mathrm{Re}}W^{\frac{1}{2}}$ is real symmetric, from the first equation we can get $\left[W^{\frac{1}{2}}(B^{\dagger}B)_{\mathrm{Im}}W^{\frac{1}{2}},W^{-\frac{1}{2}}\Lambda W^{-\frac{1}{2}}\right]=0$. $W^{\frac{1}{2}}(B^{\dagger}B)_{\mathrm{Re}}W^{\frac{1}{2}}$, $W^{\frac{1}{2}}(B^{\dagger}B)_{\mathrm{Im}}W^{\frac{1}{2}}$ and $W^{-\frac{1}{2}}\Lambda W^{-\frac{1}{2}}$ can thus be simultaneously diagonalized.
    Let $W^{\frac{1}{2}}(B^{\dagger}B)_{\mathrm{Re}}W^{\frac{1}{2}}=UD_{\mathrm{Re}}U^{\dagger}$, $W^{\frac{1}{2}}(B^{\dagger}B)_{\mathrm{Im}}W^{\frac{1}{2}}=UD_{\mathrm{Im}}U^{\dagger}$ and $W^{-\frac{1}{2}}\Lambda W^{-\frac{1}{2}}=UD_{\Lambda}U^{\dagger}$, where $D_{\mathrm{Re}}$, $D_{\mathrm{Im}}$ and $D_{\Lambda}$ are diagonal matrices and $U$ is a unitary matrix.
    From Eq.(\ref{apdx:solveB}), we then have $D_{\mathrm{Re}}=-D_{\mathrm{Im}}D_{\Lambda}$ and $D_{\mathrm{Im}}=D_{\mathrm{Re}}D_{\Lambda}$.
    Taking the absolute values on both sides, we further have $|D_{\mathrm{Re}}|=|D_{\mathrm{Im}}||D_{\Lambda}|$ and $|D_{\mathrm{Im}}|=|D_{\mathrm{Re}}||D_{\Lambda}|$.
    This gives
    \begin{equation}
        (|D_{\mathrm{Re}}|-|D_{\mathrm{Im}}|)(I+|D_{\Lambda}|)=0,
    \end{equation}
    which indicates that $|D_{\mathrm{Re}}|=|D_{\mathrm{Im}}|$. 
    Thus we have
    \begin{equation}
        \Tr\left[W(B^{\dagger}B)_{\mathrm{Re}}\right]=\Tr(D_{\mathrm{Re}})=\Tr|D_{\mathrm{Re}}|=\Tr|D_{\mathrm{Im}}|=\Tr\left|W^{\frac{1}{2}}(B^{\dagger}B)_{\mathrm{Im}}W^{\frac{1}{2}}\right|.
    \end{equation}
    Since $(R^{\dagger}R)_{\mathrm{Im}}=0$, we also have $(B^{\dagger}B)_{\mathrm{Im}}=-(A^{\dagger}A)_{\mathrm{Im}}$, thus $\Tr\left[W(B^{\dagger}B)_{\mathrm{Re}}\right]=\Tr\left|W^{\frac{1}{2}}(A^{\dagger}A)_{\mathrm{Im}}W^{\frac{1}{2}}\right|$ and
    \begin{equation}
        \begin{aligned}
            \mathcal{E}_0'&=\min_{R}\left\{\Tr\left[W\left((R^{\dagger}R)_{\mathrm{Re}}-R^{\dagger}X-X^{\dagger}R+X^{\dagger}X\right)\right]|(R^{\dagger}R)_{\mathrm{Im}}=0\right\}\\
            &=\min_{A:\ket{a_j}\in\mathcal{X},B:\ket{b_j}\in\mathcal{X}^{\bot}}\left\{\Tr\left[W\left((A^{\dagger}A)_{\mathrm{Re}}+(B^{\dagger}B)_{\mathrm{Re}}-A^{\dagger}X-X^{\dagger}A+X^{\dagger}X\right)\right]|(A^{\dagger}A)_{\mathrm{Im}}+(B^{\dagger}B)_{\mathrm{Im}}=0\right\}\\
            &=\min_{A:\ket{a_j}\in\mathcal{X}}\left\{\Tr\left[W\left((A^{\dagger}A)_{\mathrm{Re}}-A^{\dagger}X-X^{\dagger}A+X^{\dagger}X\right)\right]+\Tr\left|W^{\frac{1}{2}}(A^{\dagger}A)_{\mathrm{Im}}W^{\frac{1}{2}}\right|\right\}\geq\mathcal{E}_0.
        \end{aligned}
    \end{equation}
    Here the last inequality holds as there exist an additional constraint in $\mathcal{E}_0'$.
    On the other hand, for $\mathcal{E}_0$, we have
    \begin{equation}
        \begin{aligned}
            \mathcal{E}_0=&\min_{S,R}\left\{\Tr\left[W\left(S-R^{\dagger}X-X^{\dagger}R+X^{\dagger}X\right)\right]|S\ \text{real symmetric}, S\geq R^{\dagger}R\right\}\\
            =&\min_{S,A:\ket{a_j}\in\mathcal{X},B:\ket{b_j}\in\mathcal{X}^{\bot}}\left\{\Tr\left[W\left(S-A^{\dagger}X-X^{\dagger}A+X^{\dagger}X\right)\right]|S\ \text{real symmetric}, S\geq A^{\dagger}A+B^{\dagger}B\right\}\\
            \geq&\min_{S,A:\ket{a_j}\in\mathcal{X}}\left\{\Tr\left[W\left(S-A^{\dagger}X-X^{\dagger}A+X^{\dagger}X\right)\right]|S\ \text{real symmetric}, S\geq A^{\dagger}A\right\}\\
            =&\min_{A:\ket{a_j}\in\mathcal{X}}\left\{\Tr\left[W\left((A^{\dagger}A)_{\mathrm{Re}}-A^{\dagger}X-X^{\dagger}A+X^{\dagger}X\right)\right]+\Tr\left|W^{\frac{1}{2}}(A^{\dagger}A)_{\mathrm{Im}}W^{\frac{1}{2}}\right|\right\}=\mathcal{E}_0',
        \end{aligned}
    \end{equation}
    where the inequality holds since $S\geq A^{\dagger}A+B^{\dagger}B$ indicates that $S\geq A^{\dagger}A$.
    Combining with $\mathcal{E}_0'\geq\mathcal{E}_0$, we have $\mathcal{E}_0'=\mathcal{E}_0$.
\end{proof}

The tightness of $\mathcal{E}_0$ can then be proved by the following lemma.

\begin{lemma}\label{apdx:lemma.tightness}
    $\mathcal{E}_0$ is tight, i.e., if $(R^{\dagger}R)_{\mathrm{Im}}=0$, there always exist a POVM $\{M_m\}$ and a set of $\{f_j(m)\}$, such that $\sum_mf_j(m)M_m\ket{\psi}=\ket{r_j}$ for $R=\begin{pmatrix}\ket{r_1} & \ket{r_2} & \cdots & \ket{r_n}\end{pmatrix}$, and the weighted root-mean-square error 
    \begin{equation}     \mathcal{E}=\Tr\left[W\left((R^{\dagger}R)_{\mathrm{Re}}-R^{\dagger}X-X^{\dagger}R+X^{\dagger}X\right)\right].
    \end{equation} 
\end{lemma}
\begin{proof}
    For arbitrary $R=\begin{pmatrix}\ket{r_1} & \ket{r_2} & \cdots & \ket{r_n}\end{pmatrix}$, by doing Schmidt's orthonormalization on $\{\ket{\psi},\ket{r_1},\ket{r_2},...,\ket{r_n}\}$, we can obtain an orthonormal basis $\{\ket{u_k}\}_{k=1}^{n_d}$ with $n_d=\min\{n+1, d\}$, and
    \begin{eqnarray}\label{eq:Schmidt}
        \aligned
        |u_1\rangle&=|\psi\rangle,\\
        |u_2\rangle&=\mathcal{N}(|r_1\rangle-\langle \psi|r_1\rangle|\psi\rangle),\\
        &\vdots
        \endaligned
    \end{eqnarray}
    where $\mathcal{N}$ is a normalization factor which is a real number. Since $R_j$ is Hermitian, we have $\langle \psi|r_j\rangle=\langle \psi|R_j|\psi\rangle\in \mathcal{R}$, and from $(R^{\dagger}R)_{\mathrm{Im}}=0$, we also have $\mathrm{Im}\inp{r_j}{r_k}=0$, $\forall j,k$. Thus the coefficients in Eq.(\ref{eq:Schmidt}) are thus all real numbers. Each $\ket{r_j}$ can then be expanded as $\ket{r_j}=\sum_{k=1}^{n_d}\lambda_{jk}\ket{u_k}$with $\lambda_{jk}\in \mathcal{R}$.
    
    We then choose an $n_d\times n_d$ real orthogonal matrix $P$, with $P_{jk}$ being its $jk$th element, to obtain a rotated basis $\ket{u_j'}=\sum_{k=1}^{n_d}P_{jk}\ket{u_k}$. By choosing a $P$ such that $P_{j1}\neq 0$, we have $\inp{\psi}{u_j'}\neq 0$, $\forall j$.
    In the new basis, $\{\ket{u_j'}\}$, $\ket{r_j}$ can be represented as
    \begin{equation}
        \ket{r_j}=\sum_{k=1}^{n_d}\lambda_{jk}\ket{u_k}=\sum_{k=1}^{n_d}\lambda_{jk}\sum_{l=1}^{n_d}P_{lk}\ket{u_l'}=\sum_{l=1}^{n_d}\left(\sum_{k=1}^{n_d}\lambda_{jk}P_{lk}\right)\ket{u_l'}.
    \end{equation}
    We can now choose a measurement 
    \begin{equation}\label{apdx:optm}
        M_m=\ket{u_m'}\bra{u_m'},\ \text{for}\ 1\le m\le n_d,\quad M_0=I-\sum_{m=1}^{n_d}\ket{u_m'}\bra{u_m'},
    \end{equation}
    and
    \begin{equation}\label{apdx:optfjm}
        f_j(m)=\frac{\sum_{k=1}^{n_d}\lambda_{jk}P_{mk}}{\inp{u_m'}{\psi}}\ \text{for}\ 1\le m\le n_d,\quad f_j(m)=0\ \text{for}\ m=0,
    \end{equation}
    for which we have $\sum_{m=0}^{n_d} f_j(m)M_m\ket{\psi}=\sum_{m=1}^{n_d} f_j(m)\inp{u_m'}{\psi}\ket{u_m'}+f_j(0)M_0\ket{\psi}=|r_j\rangle$.
    The bound given by
    \begin{equation}
        \mathcal{E}_0=\min_{R}\left\{\Tr\left[W\left((R^{\dagger}R)_{\mathrm{Re}}-R^{\dagger}X-X^{\dagger}R+X^{\dagger}X\right)\right]|(R^{\dagger}R)_{\mathrm{Im}}=0\right\}
    \end{equation}
    can thus be achieved.
    
\end{proof}

\subsection{Tightness for mixed states}

In this subsection, we provide an explicit example that for mixed states, $\mathcal{E}\geq\mathcal{E}_0$ is in general not tight.

Take the state $\rho$ and observables $X_1$, $X_2$ as
\begin{equation}
    \rho=\frac{1}{3}\begin{pmatrix}
        1-p & 0 & 0 & 0\\
        0 & 1 & 0 & 0\\
        0 & 0 & 1+p & 0\\
        0 & 0 & 0 & 0
    \end{pmatrix},
    X_1=\begin{pmatrix}
        1 & 0 & 0 & 0\\
        0 & 1 & 0 & \frac{1}{2}\\
        0 & 0 & -2 & 0\\
        0 & \frac{1}{2} & 0 & 0
    \end{pmatrix},
    X_2=\begin{pmatrix}
        0 & 1 & 0 & 0\\
        1 & 0 & 1 & 0\\
        0 & 1 & 0 & -\frac{3}{2}\\
        0 & 0 & -\frac{3}{2} & 0
    \end{pmatrix}.
\end{equation}
Besides the calculation of the SDP $\mathcal{E}_0$, we do a brute-force search over all the possible measurements for the error $\epsilon_1^2+\epsilon_2^2$ numerically.
Then a gap between the SDP bound and a brute-force search can be illustrated as Fig.\ref{fig.SDPgap}, indicating that for mixed states, $\mathcal{E}\geq\mathcal{E}_0$ is in general not tight.

\begin{figure}[htb]
	\centering
	\includegraphics[width=0.5\textwidth]{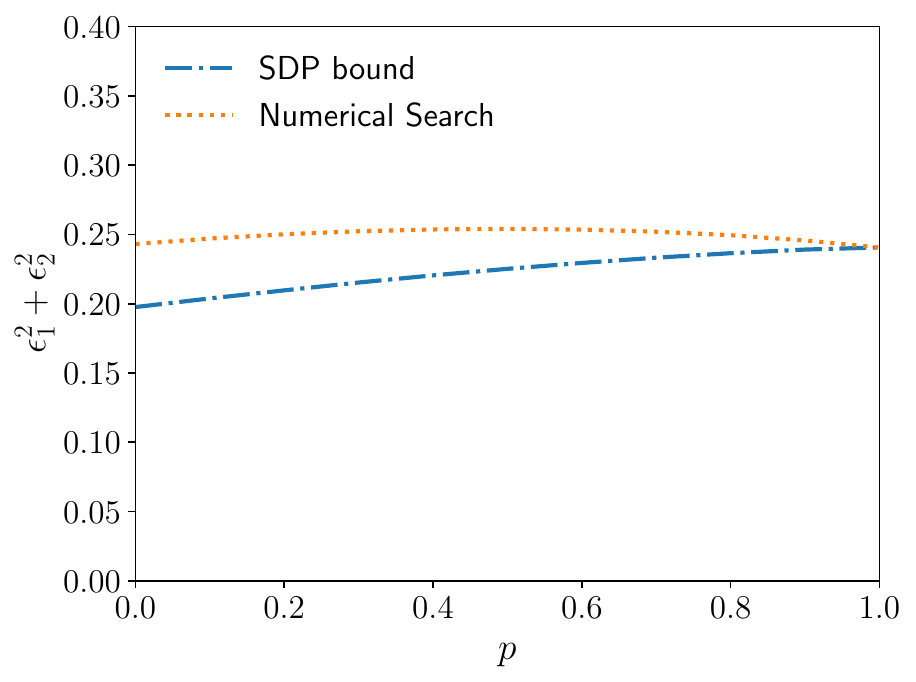}
	\caption{Comparison of the SDP bound and a brute-force search.}
    \label{fig.SDPgap}
\end{figure}

\section{Analytical error-tradeoff relation and optimal measurement for two observables on pure states}\label{apdx:tightpuretwo}

In this section, we analytically solve the optimization in $\mathcal{E}_0$ in the case of two observables with $W=diag(w_1,w_2)$ and $\rho=\ket{\psi}\bra{\psi}$ as a pure state, and provide the optimal measurement that saturates the bound. 

In the case of two observables, the condition, $(R^{\dagger}R)_{\mathrm{Im}}=0$, is equivalent to $\Tr(\rho\left[R_1,R_2\right])=0$ with $\rho=\ket{\psi}\bra{\psi}$.
The minimization in $\mathcal{E}_0$ can thus be equivalently written as
\begin{equation}
    \begin{aligned}
        \mathcal{E}_0&=\min_{R_1,R_2}\left\{\sum_{j=1}^2 w_j\Tr\left[\rho (R_jR_j-R_jX_j-X_jR_j+X_jX_j)\right]|\Tr(\rho\left[R_1,R_2\right])=0,R_1,R_2\ \text{Hermitian}\right\}.
    \end{aligned}
\end{equation}
Denote $\boldsymbol{R}=\sqrt{w_1}\tilde{R}_1+i\sqrt{w_2}\tilde{R}_2$ and $\boldsymbol{X}=\sqrt{w_1}\tilde{X}_1+i\sqrt{w_2}\tilde{X}_2$, $\mathcal{E}_0$ can then be written as
\begin{equation}
    \mathcal{E}_0=\min_{\boldsymbol{R}}\left\{\Tr(\rho \boldsymbol{R}^{\dagger}\boldsymbol{R})-\frac{1}{2}\Tr\left[\rho (\boldsymbol{R}\boldsymbol{X}^{\dagger}+\boldsymbol{X}^{\dagger}\boldsymbol{R}+\boldsymbol{R}^{\dagger}\boldsymbol{X}+\boldsymbol{X}\boldsymbol{R}^{\dagger}-\boldsymbol{X}\boldsymbol{X}^{\dagger}-\boldsymbol{X}^{\dagger}\boldsymbol{X})\right]|\Tr(\rho \boldsymbol{R}^{\dagger}\boldsymbol{R})=\Tr(\rho \boldsymbol{R}\boldsymbol{R}^{\dagger})\right\}.
\end{equation}
The minimization can be solved via the Lagrangian multiplier with 
\begin{equation}\label{apdx:Lagrangian}
    \mathcal{L}(\boldsymbol{R},\mu)=\Tr(\rho \boldsymbol{R}^{\dagger}\boldsymbol{R})-\frac{1}{2}\Tr\left[\rho (\boldsymbol{R}\boldsymbol{X}^{\dagger}+\boldsymbol{X}^{\dagger}\boldsymbol{R}+\boldsymbol{R}^{\dagger}\boldsymbol{X}+\boldsymbol{X}\boldsymbol{R}^{\dagger}-\boldsymbol{X}\boldsymbol{X}^{\dagger}-\boldsymbol{X}^{\dagger}\boldsymbol{X})\right]+\mu\left[\Tr(\rho \boldsymbol{R}\boldsymbol{R}^{\dagger})-\Tr(\rho \boldsymbol{R}^{\dagger}\boldsymbol{R})\right].
\end{equation}
The matrix derivative of $\mathcal{L}(\boldsymbol{R},\mu)$ with respect to $\boldsymbol{R}$ can be calculated as
\begin{equation}
    \begin{aligned}
        \partial_{\boldsymbol{R}}\mathcal{L}(\boldsymbol{R},\mu)&=\lim_{\varepsilon\rightarrow 0}\frac{\mathcal{L}(\boldsymbol{R}+\varepsilon\cdot\delta\boldsymbol{R},\mu)-\mathcal{L}(\boldsymbol{R},\mu)}{\varepsilon}\\
        &=\Tr\left\{\delta\boldsymbol{R}\left[\mu\boldsymbol{R}^{\dagger}\rho+(1-\mu)\rho\boldsymbol{R}^{\dagger}-\frac{1}{2}\left(\boldsymbol{X}^{\dagger}\rho+\rho\boldsymbol{X}^{\dagger}\right)\right]+\delta\boldsymbol{R}^{\dagger}\left[\mu\rho\boldsymbol{R}+(1-\mu)\boldsymbol{R}\rho-\frac{1}{2}\left(\rho\boldsymbol{X}+\boldsymbol{X}\rho\right)\right]\right\}.
    \end{aligned}
\end{equation}
Then we let $\partial_{\boldsymbol{R}}\mathcal{L}(\boldsymbol{R},\mu)=0$ for arbitrary $\delta\boldsymbol{R}$ and obtain
\begin{equation}\label{eq:suppR}
    \mu\rho\boldsymbol{R}+(1-\mu)\boldsymbol{R}\rho-\frac{1}{2}\left(\rho\boldsymbol{X}+\boldsymbol{X}\rho\right)=0.
\end{equation}
By taking trace on both sides of the equation, we have $\bra{\psi}\boldsymbol{R}\ket{\psi}=\bra{\psi}\boldsymbol{X}\ket{\psi}$.
By multiplying $\ket{\psi}$ from either left or right side of Eq.(\ref{eq:suppR}), we get
\begin{equation}\label{apdx:RRdagger}
    \begin{aligned}
        (1-\mu)\boldsymbol{R}\ket{\psi}&=\frac{1}{2}(1-2\mu)\bra{\psi}\boldsymbol{X}\ket{\psi}\ket{\psi}+\frac{1}{2}\boldsymbol{X}\ket{\psi},\\
        \mu\boldsymbol{R}^{\dagger}\ket{\psi}&=\frac{1}{2}(2\mu-1)\bra{\psi}\boldsymbol{X}^{\dagger}\ket{\psi}\ket{\psi}+\frac{1}{2}\boldsymbol{X}^{\dagger}\ket{\psi}.
    \end{aligned}
\end{equation}
Using Eq.(\ref{apdx:RRdagger}) and $\bra{\psi}\boldsymbol{R}^{\dagger}\boldsymbol{R}\ket{\psi}=\bra{\psi}\boldsymbol{R}\boldsymbol{R}^{\dagger}\ket{\psi}$, we can then solve $\mu$. 
We first consider the cases where $\mu\notin \{0,1\}$. In this case we have
\begin{equation}
    \begin{aligned}
        \bra{\psi}\boldsymbol{R}^{\dagger}\boldsymbol{R}\ket{\psi}&=\frac{1}{(1-\mu)^2}\left(\frac{1}{2}(1-2\mu)\bra{\psi}\boldsymbol{X}^{\dagger}\ket{\psi}\bra{\psi}+\frac{1}{2}\bra{\psi}\boldsymbol{X}^{\dagger}\right)\left(\frac{1}{2}(1-2\mu)\bra{\psi}\boldsymbol{X}\ket{\psi}\ket{\psi}+\frac{1}{2}\boldsymbol{X}\ket{\psi}\right)\\
        &=\frac{(1-2\mu)(3-2\mu)}{4(1-\mu)^2}|\bra{\psi}\boldsymbol{X}\ket{\psi}|^2+\frac{1}{4(1-\mu)^2}\bra{\psi}\boldsymbol{X}^{\dagger}\boldsymbol{X}\ket{\psi},\\
        \bra{\psi}\boldsymbol{R}\boldsymbol{R}^{\dagger}\ket{\psi}&=\frac{1}{\mu^2}\left(\frac{1}{2}(2\mu-1)\bra{\psi}\boldsymbol{X}\ket{\psi}\bra{\psi}+\frac{1}{2}\bra{\psi}\boldsymbol{X}\right)\left(\frac{1}{2}(2\mu-1)\bra{\psi}\boldsymbol{X}^{\dagger}\ket{\psi}\ket{\psi}+\frac{1}{2}\boldsymbol{X}^{\dagger}\ket{\psi}\right)\\
        &=\frac{(2\mu-1)(2\mu+1)}{4\mu^2}|\bra{\psi}\boldsymbol{X}\ket{\psi}|^2+\frac{1}{4\mu^2}\bra{\psi}\boldsymbol{X}\boldsymbol{X}^{\dagger}\ket{\psi}.
    \end{aligned}
\end{equation}
$\bra{\psi}\boldsymbol{R}^{\dagger}\boldsymbol{R}\ket{\psi}-\bra{\psi}\boldsymbol{R}\boldsymbol{R}^{\dagger}\ket{\psi}=0$ then leads to
\begin{equation}\label{apdx:eqmu}
    \frac{1}{4}\left(\frac{1}{(1-\mu)^2}-\frac{1}{\mu^2}\right)\alpha+\frac{1}{4}\left(\frac{1}{\mu^2}+\frac{1}{(1-\mu)^2}\right)\beta=0,
\end{equation}
where 
\begin{equation}\label{apdx:alphabeta}
    \begin{aligned}
        \alpha&=\frac{1}{2}\bra{\psi}(\boldsymbol{X}^{\dagger}\boldsymbol{X}+\boldsymbol{X}\boldsymbol{X}^{\dagger})\ket{\psi}-|\bra{\psi}\boldsymbol{X}\ket{\psi}|^2=w_1(\Delta X_1)^2+w_2(\Delta X_2)^2,\\
        \beta&=\frac{1}{2}\bra{\psi}(\boldsymbol{X}^{\dagger}\boldsymbol{X}-\boldsymbol{X}\boldsymbol{X}^{\dagger})\ket{\psi}=i\sqrt{w_1w_2}\bra{\psi}[X_1,X_2]\ket{\psi},
    \end{aligned}
\end{equation}
here $(\Delta X_j)^2=\bra{\psi}X_j^2\ket{\psi}-\bra{\psi}X_j\ket{\psi}^2$ denotes the variance of $X_j$ for $j=1,2$.
From Eq.(\ref{apdx:eqmu}) we get two solutions,
\begin{equation}
    \mu_{\pm}=\frac{-(\alpha-\beta)\pm\sqrt{\alpha^2-\beta^2}}{2\beta}.
\end{equation}
With $\mu=\mu_{\pm}$, and $R_1$, $R_2$ satisfies Eq.(\ref{apdx:RRdagger}) we have
\begin{equation}
    \begin{aligned}
        \bra{\psi}\boldsymbol{R}\boldsymbol{X}^{\dagger}\ket{\psi}&=\frac{1}{\mu}\left(\frac{1}{2}(2\mu-1)\bra{\psi}\boldsymbol{X}\ket{\psi}\bra{\psi}+\frac{1}{2}\bra{\psi}\boldsymbol{X}\right)\boldsymbol{X}^{\dagger}\ket{\psi}=\frac{2\mu-1}{2\mu}|\bra{\psi}\boldsymbol{X}\ket{\psi}|^2+\frac{1}{2\mu}\bra{\psi}\boldsymbol{X}\boldsymbol{X}^{\dagger}\ket{\psi},\\
        \bra{\psi}\boldsymbol{X}^{\dagger}\boldsymbol{R}\ket{\psi}&=\frac{1}{1-\mu}\bra{\psi}\boldsymbol{X}^{\dagger}\left(\frac{1}{2}(1-2\mu)\bra{\psi}\boldsymbol{X}\ket{\psi}\ket{\psi}+\frac{1}{2}\boldsymbol{X}\ket{\psi}\right)=\frac{1-2\mu}{2(1-\mu)}|\bra{\psi}\boldsymbol{X}\ket{\psi}|^2+\frac{1}{2(1-\mu)}\bra{\psi}\boldsymbol{X}^{\dagger}\boldsymbol{X}\ket{\psi},
    \end{aligned}
\end{equation}
and
\begin{equation}
    \begin{aligned}
        \mathcal{E}_0=&\bra{\psi}\boldsymbol{R}^{\dagger}\boldsymbol{R}\ket{\psi}-\frac{1}{2}\bra{\psi}(\boldsymbol{R}\boldsymbol{X}^{\dagger}+\boldsymbol{X}^{\dagger}\boldsymbol{R}+\boldsymbol{R}^{\dagger}\boldsymbol{X}+\boldsymbol{X}\boldsymbol{R}^{\dagger}-\boldsymbol{X}\boldsymbol{X}^{\dagger}-\boldsymbol{X}^{\dagger}\boldsymbol{X})\ket{\psi}\\
        =&\left(\frac{(1-2\mu)(3-2\mu)}{4(1-\mu)^2}|\bra{\psi}\boldsymbol{X}\ket{\psi}|^2+\frac{1}{4(1-\mu)^2}\bra{\psi}\boldsymbol{X}^{\dagger}\boldsymbol{X}\ket{\psi}\right)+\frac{1}{2}\bra{\psi}(\boldsymbol{X}\boldsymbol{X}^{\dagger}+\boldsymbol{X}^{\dagger}\boldsymbol{X})\ket{\psi}\\
        &-\left(\frac{2\mu-1}{2\mu}|\bra{\psi}\boldsymbol{X}\ket{\psi}|^2+\frac{1}{2\mu}\bra{\psi}\boldsymbol{X}\boldsymbol{X}^{\dagger}\ket{\psi}\right)-\left(\frac{1-2\mu}{2(1-\mu)}|\bra{\psi}\boldsymbol{X}\ket{\psi}|^2+\frac{1}{2(1-\mu)}\bra{\psi}\boldsymbol{X}^{\dagger}\boldsymbol{X}\ket{\psi}\right)\\
        =&\left(\frac{(1-2\mu)(3-2\mu)}{4(1-\mu)^2}-\frac{2\mu-1}{2\mu}-\frac{1-2\mu}{2(1-\mu)}\right)|\bra{\psi}\boldsymbol{X}\ket{\psi}|^2+\frac{1}{2}\bra{\psi}(\boldsymbol{X}\boldsymbol{X}^{\dagger}+\boldsymbol{X}^{\dagger}\boldsymbol{X})\ket{\psi}\\
        &+\left(\frac{1}{4(1-\mu)^2}-\frac{1}{2(1-\mu)}\right)\bra{\psi}\boldsymbol{X}^{\dagger}\boldsymbol{X}\ket{\psi}-\frac{1}{2\mu}\bra{\psi}\boldsymbol{X}\boldsymbol{X}^{\dagger}\ket{\psi}\\
        =&\left(1-\frac{2-3\mu}{4\mu(1-\mu)^2}\right)\alpha+\frac{4\mu^2-5\mu+2}{4\mu(1-\mu)^2}\beta\\
        =&\frac{1}{2}\left(\alpha\mp\sqrt{\alpha^2-\beta^2}\right).
    \end{aligned}
\end{equation}
The minimal $\mathcal{E}_0=\frac{1}{2}\left(\alpha-\sqrt{\alpha^2-\beta^2}\right)$ is achieved with $\mu=\mu_{+}$.

In the cases where $\mu\in\{0,1\}$, Eq.(\ref{apdx:RRdagger}) have solutions only if $\alpha=\beta$. For example, if $\mu=0$, Eq.(\ref{apdx:RRdagger}) becomes
\begin{equation}\label{apdx:mu01}
    \begin{aligned}
        \boldsymbol{R}\ket{\psi}&=\frac{1}{2}\bra{\psi}\boldsymbol{X}\ket{\psi}\ket{\psi}+\frac{1}{2}\boldsymbol{X}\ket{\psi},\\
        0&=-\frac{1}{2}\bra{\psi}\boldsymbol{X}^{\dagger}\ket{\psi}\ket{\psi}+\frac{1}{2}\boldsymbol{X}^{\dagger}\ket{\psi}.
    \end{aligned}
\end{equation}
Multiply $\langle \psi|\boldsymbol{X}$ from the left of the second equation, we have $\langle \psi|\boldsymbol{X}\boldsymbol{X}^\dagger|\psi\rangle=|\langle\psi|\boldsymbol{X}|\psi\rangle|^2$, which gives $\alpha=\beta$.  
In this case 
\begin{equation}
    \aligned
    \langle\psi|\boldsymbol{R}^\dagger \boldsymbol{R}|\psi\rangle&=\frac{1}{4}(3|\langle \psi| \boldsymbol{X}|\psi\rangle|^2+ \langle \psi|\boldsymbol{X}^\dagger \boldsymbol{X}|\psi\rangle),\\
    \langle\psi|\boldsymbol{R}\boldsymbol{X}^\dagger |\psi\rangle&=|\langle\psi|\boldsymbol{X}|\psi\rangle|^2,\\
    \langle\psi|\boldsymbol{R}^\dagger \boldsymbol{X}|\psi\rangle&=\frac{1}{2}(\langle \psi|\boldsymbol{X}^\dagger \boldsymbol{X}|\psi\rangle+|\langle\psi|\boldsymbol{X}|\psi\rangle|^2).
    \endaligned
\end{equation}
Then
\begin{equation}
    \aligned
    \mathcal{E}_0=&\bra{\psi}\boldsymbol{R}^{\dagger}\boldsymbol{R}\ket{\psi}-\frac{1}{2}\bra{\psi}(\boldsymbol{R}\boldsymbol{X}^{\dagger}+\boldsymbol{X}^{\dagger}\boldsymbol{R}+\boldsymbol{R}^{\dagger}\boldsymbol{X}+\boldsymbol{X}\boldsymbol{R}^{\dagger}-\boldsymbol{X}\boldsymbol{X}^{\dagger}-\boldsymbol{X}^{\dagger}\boldsymbol{X})\ket{\psi}\\
    =&\frac{1}{4}(\langle\psi|\boldsymbol{X}^\dagger \boldsymbol{X}|\psi\rangle-|\langle\psi|\boldsymbol{X}|\psi\rangle|^2)\\
    =&\frac{1}{4}(\langle\psi|\boldsymbol{X}^\dagger \boldsymbol{X}|\psi\rangle-\langle\psi|\boldsymbol{X}\boldsymbol{X}^\dagger|\psi\rangle)\\
    =&\frac{1}{2}\alpha.
    \endaligned
\end{equation}
Similarly for $\mu=1$, we also have $\alpha=\beta$ and $\mathcal{E}_0=\frac{1}{2}\alpha$. Thus $\mu\in\{0,1\}$ can only happen when $\alpha=\beta$, in which case the minimal $\mathcal{E}_0$ can also be written as 
$\mathcal{E}_0=\frac{1}{2}(\alpha-\sqrt{\alpha^2-\beta^2})$. To summarize, we have the tight bound
\begin{equation}
    w_1\epsilon_1^2+w_2\epsilon_2^2\geq\frac{1}{2}\left(\alpha-\sqrt{\alpha^2-\beta^2}\right),
\end{equation}
where $\alpha=w_1(\Delta X_1)^2+w_2(\Delta X_2)^2,\ \beta=i\sqrt{w_1w_2}\bra{\psi}[X_1,X_2]\ket{\psi}$. 

We note that since Ozawa's relation is tight for two observables on pure states, the tight bound can also be obtained from the Ozawa's relation, which is presented in the next section. The advantage of the derivation here is it can directly use Lemma \ref{apdx:lemma.tightness} in Appendix \ref{apdx:tightpure} to construct the optimal measurement. Without loss of generality, we can assume $\bra{\psi}X_1\ket{\psi}=\bra{\psi}X_2\ket{\psi}=0$, which is equivalent to replacing $X_j=X_j-\bra{\psi}X_j\ket{\psi}I$ for $j=1,2$ and does not affect the root-mean-square error and the optimal measurement.
Then with Eq.(\ref{apdx:RRdagger}), we have
\begin{equation}
    \begin{aligned}
        (1-\mu)(\ket{r_1}+i\ket{r_2})&=\frac{1}{2}(\ket{x_1}+i\ket{x_2}),\\
        \mu(\ket{r_1}-i\ket{r_2})&=\frac{1}{2}(\ket{x_1}-i\ket{x_2}),
    \end{aligned}
\end{equation}
which gives $\ket{r_1}$, $\ket{r_2}$ as
\begin{equation}
    \begin{aligned}
        \ket{r_1}&=\frac{1}{4(1-\mu)}(\ket{x_1}+i\ket{x_2})+\frac{1}{4\mu}(\ket{x_1}-i\ket{x_2})=\frac{1}{4\mu(1-\mu)}\ket{x_1}-i\frac{1-2\mu}{4\mu(1-\mu)}\ket{x_2},\\
        \ket{r_2}&=\frac{1}{i4(1-\mu)}(\ket{x_1}+i\ket{x_2})-\frac{1}{i4\mu}(\ket{x_1}-i\ket{x_2})=i\frac{1-2\mu}{4\mu(1-\mu)}\ket{x_1}+\frac{1}{4\mu(1-\mu)}\ket{x_2},
    \end{aligned}
\end{equation}
with $\ket{x_j}=X_j\ket{\psi}$, $j=1,2$, $\mu=\mu_{+}=\frac{-(\alpha-\beta)+\sqrt{\alpha^2-\beta^2}}{2\beta}$. 
By finding the orthonormal basis $\ket{u_m'}$ of the span $\{\ket{\psi},\ket{r_1},\ket{r_2}\}$ such that $\inp{\psi}{u_m'}\neq 0$, $\forall m$, the optimal measurement $\{M_m\}$ and optimal $\{f_j(m)\}$ can be obtained as Eq.(\ref{apdx:optm}) and Eq.(\ref{apdx:optfjm}).

\section{Tighter analytical bound for two observables on mixed states}\label{apdx:tightermixed}
In this section, we derive a tighter analytical bound for two observables on mixed states than the Ozawa's bound.
\subsection{Minimal weighted error from Ozawa's relation}\label{apdx:ozawa}
For the comparison, we first identify the minimal weighted error, $w_1\epsilon_1^2+w_2\epsilon_2^2$, derived from the Ozawa's relation,
\begin{equation}\label{eq:suppozawa}
    \epsilon_1^2\cdot (\Delta X_2)^2+(\Delta X_1)^2\cdot\epsilon_2^2+2\sqrt{(\Delta X_1)^2\cdot(\Delta X_2)^2-c_{12}^2}\epsilon_1\cdot \epsilon_2 \geq c_{12}^2,
\end{equation}
where $c_{12}=\frac{1}{2}\|\sqrt{\rho}[X_1,X_2]\sqrt{\rho}\|_1$. 

The minimal weighted error can be obtained via the Lagrange multiplier with
\begin{equation}
    \mathcal{L}(\epsilon_1,\epsilon_2,\lambda)=w_1\epsilon_1^2+w_2\epsilon_2^2+\lambda\left(\epsilon_1^2\cdot (\Delta X_2)^2+(\Delta X_1)^2\cdot\epsilon_2^2+2\sqrt{(\Delta X_1)^2\cdot(\Delta X_2)^2-c_{12}^2}\epsilon_1\cdot \epsilon_2 - c_{12}^2\right),
\end{equation}
which gives
\begin{equation}\label{eq:suppLag}
    \begin{aligned}
        \partial_{\epsilon_1}\mathcal{L}(\epsilon_1,\epsilon_2,\lambda)&=2w_1\epsilon_1+\lambda\left(2\epsilon_1\cdot (\Delta X_2)^2+2\sqrt{(\Delta X_1)^2\cdot(\Delta X_2)^2-c_{12}^2} \epsilon_2\right)=0,\\
        \partial_{\epsilon_2}\mathcal{L}(\epsilon_1,\epsilon_2,\lambda)&=2w_2\epsilon_2+\lambda\left(2\epsilon_2\cdot (\Delta X_1)^2+2\sqrt{(\Delta X_1)^2\cdot(\Delta X_2)^2-c_{12}^2} \epsilon_1\right)=0,\\
    \end{aligned}
\end{equation}
and $\lambda(\epsilon_1^2\cdot (\Delta X_2)^2+(\Delta X_1)^2\cdot\epsilon_2^2+2\sqrt{(\Delta X_1)^2\cdot(\Delta X_2)^2-c_{12}^2}\epsilon_1\cdot \epsilon_2 - c_{12}^2)=0$. 
If $\lambda=0$, we have $\epsilon_1=\epsilon_2=0$, which can only satisfy the inequality in Eq.(\ref{eq:suppozawa}) when $c_{12}=0$. If $\lambda\neq 0$, we then have $\epsilon_1^2\cdot (\Delta X_2)^2+(\Delta X_1)^2\cdot\epsilon_2^2+2\sqrt{(\Delta X_1)^2\cdot(\Delta X_2)^2-c_{12}^2}\epsilon_1\cdot \epsilon_2 - c_{12}^2=0$, which, combines with Eq.(\ref{eq:suppLag}), gives the optimal $\epsilon_1$ and $\epsilon_2$ as
\begin{equation}\label{eq:suppoptimal}
    \begin{aligned}
        \epsilon_1^{\star}=&\sqrt{\frac{1}{2}(\Delta X_1)^2-\frac{1}{2\sqrt{\alpha^2-\beta^2}}\left(w_1(\Delta X_1)^4+w_2(\Delta X_1)^2\cdot(\Delta X_2)^2-2w_2c_{12}^2\right)},\\
        \epsilon_2^{\star}=&\sqrt{\frac{1}{2}(\Delta X_2)^2-\frac{1}{2\sqrt{\alpha^2-\beta^2}}\left(w_2(\Delta X_2)^4+w_1(\Delta X_1)^2\cdot(\Delta X_2)^2-2w_1c_{12}^2\right)},
    \end{aligned}
\end{equation}
where $\alpha=w_1(\Delta X_1)^2+w_2(\Delta X_2)^2$, $\beta=2\sqrt{w_1w_2}c_{12}$. Note that when $c_{12}=0$, $\epsilon_1^{\star}=\epsilon_2^{\star}=0$ are also optimal. Eq.(\ref{eq:suppoptimal}) thus give the optimal value for all cases.  
By substituting $\epsilon_1^{\star}$ and $\epsilon_2^{\star}$ in $w_1\epsilon_1^2+w_2\epsilon_2^2$, we then obtain the bound on the minimal weighted error from the Ozawa's relation as
\begin{equation}\label{eq:suppozawabound}
    w_1\epsilon_1^2+w_2\epsilon_2^2\geq\mathcal{E}_{Ozawa}\equiv w_1\epsilon_1^{\star 2}+w_2\epsilon_2^{\star 2}=\frac{1}{2}\left(\alpha-\sqrt{\alpha^2-\beta^2}\right),
\end{equation}
where $\alpha=w_1(\Delta X_1)^2+w_2(\Delta X_2)^2$, $\beta=\sqrt{w_1w_2}\|\sqrt{\rho}[X_1,X_2]\sqrt{\rho}\|_1$. Note that $\epsilon_1^{\star}$ and $\epsilon_2^{\star}$ saturates the Ozawa's relation with $\epsilon_1^{\star 2}\cdot (\Delta X_2)^2+(\Delta  X_1)^2\cdot\epsilon_2^{\star 2}+2\sqrt{(\Delta X_1)^2\cdot(\Delta X_2)^2-c_{12}^2}\epsilon_1^{\star}\cdot \epsilon_2^{\star}=c_{12}^2$, the bound in Eq.(\ref{eq:suppozawabound}) is thus as tight as the Ozawa's relation. It is saturable for pure state, however, for mixed states, however, the bound is not tight in general since the Ozawa's relation is in general not tight for mixed states.

\subsection{Tighter analytical bound for the weighted error}\label{apdx:tightertwo}

We now derive a tighter analytical bound. Given any mixed state $\rho$, the minimal $\mathcal{E}$ is lower bounded by  $\mathcal{E}_0$ in Eq.(\ref{apdx:SDP0}).
Given any $\{\ket{u_q}\}$ that satistifies $\sum_q\ket{u_q}\bra{u_q}=I$, we have $\rho=\sum_q\sqrt{\rho}\ket{u_q}\bra{u_q}\sqrt{\rho}=\sum_q\lambda_q\ket{\phi_q}\bra{\phi_q}$, here
\begin{equation}
    \lambda_q=\bra{u_q}\rho\ket{u_q},\quad \ket{\phi_q}=\frac{\sqrt{\rho}\ket{u_q}}{\sqrt{\bra{u_q}\rho\ket{u_q}}}.
\end{equation}
Then we have $\mathcal{E}_0\geq\sum_q\lambda_q\mathcal{E}_{|\phi_q\rangle}$, with
\begin{equation}
    \begin{aligned}
        \mathcal{E}_{|\phi_q\rangle}=\min_{\mathbb{S},\{R_j\}_{j=1}^n} &\Tr\left[(W\otimes\ket{\phi_q}\bra{\phi_q})(\mathbb{S}-\mathbb{R}\mathbb{X}^{\dagger}-\mathbb{X}\mathbb{R}^{\dagger}+\mathbb{X}\mathbb{X}^{\dagger})\right]\\
        \text{subject to}\quad &\mathbb{S}_{jk}=\mathbb{S}_{kj}=\mathbb{S}_{jk}^{\dagger},\ \forall j,k\\
        &R_j=R_j^{\dagger},\ \forall j\\
        &\begin{pmatrix}
            I & \mathbb{R}^{\dagger}\\
            \mathbb{R} & \mathbb{S}
        \end{pmatrix}\geq 0.
    \end{aligned}
\end{equation}
For two observables this can be analytically solved as in previous section, which gives  
\begin{equation}
    \mathcal{E}_{|\phi_q\rangle}=\frac{1}{2}\left(\alpha_q-\sqrt{\alpha_q^2-\beta_q^2}\right),
\end{equation}
where $\alpha_q=w_1(\Delta_{\ket{\phi_q}}X_1)^2+w_2(\Delta_{\ket{\phi_q}}X_2)^2,\ \beta_q=i\sqrt{w_1w_2}\bra{\phi_q}[X_1,X_2]\ket{\phi_q}$, $(\Delta_{\ket{\phi_q}}X_j)^2=\bra{\phi_q}X_j^2\ket{\phi_q}-\bra{\phi_q}X_j\ket{\phi_q}^2$.
Since $\mathcal{E}_0\geq\sum_q\lambda_q\mathcal{E}_{|\phi_q\rangle}$, we then get an analytical bound $\mathcal{E}\geq \mathcal{E}_0\geq  \mathcal{E}_A$ with 
\begin{equation}
    \mathcal{E}_A=\sum_q\lambda_q\mathcal{E}_{|\phi_q\rangle}=\sum_q\frac{1}{2}\lambda_q\left(\alpha_q-\sqrt{\alpha_q^2-\beta_q^2}\right).    
\end{equation}


Any choice of $\{|u_q\rangle\}$ with $\sum_q |u_q\rangle\langle u_q|=I$ can be used to obtain an analytical bound. Specifically by choosing $\ket{u_q}$ as the eigenvectors of $\sqrt{\rho}[X_1,X_2]\sqrt{\rho}$, we can get an analytical bound that is in general tighter than Ozawa's bound for mixed states.


For the comparison of $\mathcal{E}_A$ with $\mathcal{E}_{Ozawa}$, we first note that
\begin{equation}
    \sum_q\lambda_q|\beta_q|=\sum_q\bra{u_q}\rho\ket{u_q} \left|\sqrt{w_1w_2}\frac{\bra{u_q}\sqrt{\rho}[X_1,X_2]\sqrt{\rho}\ket{u_q}}{\bra{u_q}\rho\ket{u_q}}\right|=\sqrt{w_1w_2}\|\sqrt{\rho}[X_1,X_2]\sqrt{\rho}\|_1=\beta,
\end{equation}
\begin{equation}
    \begin{aligned}
        \sum_q\lambda_q\alpha_q&=\sum_{j=1,2}w_j\sum_q\bra{u_q}\rho\ket{u_q}(\Delta_{\ket{\phi_q}}X_j)^2\\
        &=\sum_{j=1,2}w_j\sum_q\bra{u_q}\rho\ket{u_q}\left(\frac{\bra{u_q}\sqrt{\rho}X_j^2\sqrt{\rho}\ket{u_q}}{\bra{u_q}\rho\ket{u_q}}-\bra{\phi_q}X_j\ket{\phi_q}^2\right)\\
        &\leq\sum_{j=1,2}w_j\sum_q\bra{u_q}\rho\ket{u_q}\left(\frac{\bra{u_q}\sqrt{\rho}X_j^2\sqrt{\rho}\ket{u_q}}{\bra{u_q}\rho\ket{u_q}}\right)=\sum_{j=1,2}w_j\Tr(\rho X_j^2)=\sum_{j=1,2}w_j(\Delta X_j)^2,
    \end{aligned}
\end{equation}
where  without loss of generality we assumed that $\Tr(\rho X_1)=\Tr(\rho X_2)=0$. We thus have $\sum_q\lambda_q\alpha_q\leq\alpha$ and $\sum_q\lambda_q|\beta_q|=\beta$.
We also have (note $\alpha_j\geq |\beta_j|$)
\begin{equation}
    \left(\alpha_1-\sqrt{\alpha_1^2-\beta_1^2}\right)+\left(\alpha_2-\sqrt{\alpha_2^2-\beta_2^2}\right)\geq \left(\alpha_1+\alpha_2-\sqrt{(\alpha_1+\alpha_2)^2-(|\beta_1|+|\beta_2|)^2}\right),
\end{equation}
which can be directly verified as 
\begin{eqnarray}
    \aligned
    \Leftrightarrow\quad    &\sqrt{\alpha_1^2-\beta_1^2}+\sqrt{\alpha_2^2-\beta_2^2}\leq \sqrt{(\alpha_1+\alpha_2)^2-(|\beta_1|+|\beta_2|)^2},\\
    \Leftrightarrow\quad &\alpha_1^2-\beta_1^2+\alpha_2^2-\beta_2^2+2\sqrt{(\alpha_1^2-\beta_1^2)(\alpha_2^2-\beta_2^2)}\leq \alpha_1^2-\beta_1^2+\alpha_2^2-\beta_2^2+2(\alpha_1\alpha_2-|\beta_1\beta_2|),\\
    \Leftrightarrow\quad & (\alpha_1^2-\beta_1^2)(\alpha_2^2-\beta_2^2)\leq (\alpha_1\alpha_2-|\beta_1\beta_2|)^2,\\
    \Leftrightarrow\quad & \alpha_1^2\beta_2^2+\beta_1^2\alpha_2^2\geq 2\alpha_1\alpha_2|\beta_1\beta_2|\\
    \Leftrightarrow\quad &
    (\alpha_1|\beta_2|-|\beta_1|\alpha_2)^2\geq 0.
    \endaligned
\end{eqnarray}
We thus have
\begin{equation}
    \begin{aligned}
        \mathcal{E}_A&=\frac{1}{2}\sum_q\lambda_q\left(\alpha_q-\sqrt{\alpha_q^2-\beta_q^2}\right)=\frac{1}{2}\sum_q\left((\lambda_q\alpha_q)-\sqrt{(\lambda_q\alpha_q)^2-(\lambda_q|\beta_q|)^2}\right)\\
        &\geq \frac{1}{2}\left((\sum_q\lambda_q\alpha_q)-\sqrt{(\sum_q\lambda_q\alpha_q)^2-(\sum_q\lambda_q|\beta_q|)^2}\right)\\
        &\geq \frac{1}{2}\left(\alpha-\sqrt{\alpha^2-\beta^2}\right)=\mathcal{E}_{Ozawa},
    \end{aligned}
\end{equation}
where the last inequality we used the fact that $f(\alpha)=\alpha-\sqrt{\alpha^2-\beta^2}$, $\alpha\in (\beta,\infty)$, is a monotonically non-increasing function since $f'(\alpha)=1-\frac{\alpha}{\sqrt{\alpha^2-\beta^2}}\leq 0$. The analytical bound, 
$\mathcal{E}_A=\sum_q\frac{1}{2}\lambda_q\left(\alpha_q-\sqrt{\alpha_q^2-\beta_q^2}\right)$,    
is thus always tighter than the bound obtained from the Ozawa's relation,
$\mathcal{E}_{Ozawa}= \frac{1}{2}\left(\alpha-\sqrt{\alpha^2-\beta^2}\right).$

\subsection{Example}
Here we provide an example to illustrate the difference among the SDP bound $\mathcal{E}_0$ (Eq.(\ref{apdx:SDP0})), the analytical bound $\mathcal{E}_A$, and the bound obtained from the Ozawa's relation, $\mathcal{E}_{Ozawa}$. 


Take $\rho$ and three observables, $X_1$, $X_2$, $X_3$ as
\begin{equation}
	\begin{aligned}
		\rho=\begin{pmatrix}
			\frac{p}{2} & 0 & 0\\
			0 & 1-p & 0\\
			0 & 0 & \frac{p}{2}
		\end{pmatrix},\ 
		X_1=\begin{pmatrix}
			0 & 1 & 0\\
			1 & 0 & 1\\
			0 & 1 & 0
		\end{pmatrix},\\
		X_2=\begin{pmatrix}
			0 & -i & 0\\
			i & 0 & -i\\
			0 & i & 0
		\end{pmatrix},\ 
		X_3=\begin{pmatrix}
			1 & 0 & 0\\
			0 & 0 & 0\\
			0 & 0 & -1
		\end{pmatrix}.
	\end{aligned}
\end{equation}
Here $X_1$, $X_2$, $X_3$ are the spin-1 operators with $[X_j,X_k]\neq 0$ for all $j\neq k$.
In this case we have $(\Delta X_1)^2=(\Delta X_2)^2=2-p$, $(\Delta X_3)^2=p$, and
\begin{equation}
	\begin{aligned}
		c_{12}&=\frac{1}{2}\|\sqrt{\rho}[X_1,X_2]\sqrt{\rho}\|_1=p,\\
		c_{23}&=\frac{1}{2}\|\sqrt{\rho}[X_2,X_3]\sqrt{\rho}\|_1=\sqrt{p(1-p)},\\
		c_{31}&=\frac{1}{2}\|\sqrt{\rho}[X_3,X_1]\sqrt{\rho}\|_1=\sqrt{p(1-p)}.
	\end{aligned}
\end{equation}

For the simultaneous measurement of $(X_1,X_2)$, $(X_2,X_3)$, $(X_3,X_1)$, the Ozawa's relation gives
\begin{equation}
	\begin{aligned}
		\epsilon_1^2+\epsilon_2^2&\geq \mathcal{E}_{Ozawa}(X_1,X_2)=2-p-2\sqrt{1-p},\\
		\epsilon_2^2+\epsilon_3^2&\geq \mathcal{E}_{Ozawa}(X_2,X_3)=1-\sqrt{1-p(1-p)},\\
		\epsilon_3^2+\epsilon_1^2&\geq \mathcal{E}_{Ozawa}(X_3,X_1)=1-\sqrt{1-p(1-p)}.
	\end{aligned}
\end{equation}
While $\mathcal{E}_A$ can be calculated by substituting $\{\ket{u_q}\}$ with the eigenvectors of $ \sqrt{\rho}[X_j,X_k]\sqrt{\rho}$.
Specifically for $(X_2,X_3)$, we have $\ket{u_0}=(1,\sqrt{2},1)^T/2$, $\ket{u_1}=(1,-\sqrt{2},1)^T/2$, $\ket{u_2}=(1,0,-1)^T/\sqrt{2}$,
which gives
\begin{equation}
	\begin{aligned}
		&\epsilon_2^2+\epsilon_3^2\geq \mathcal{E}_{A}(X_2,X_3)=\sum_{q=0}^2\frac{1}{2}\lambda_q\left(\alpha_q-\sqrt{\alpha_q^2-\beta_q^2}\right)\\
		&=\frac{2-p}{4}\left(\frac{4-3p}{2-p}-\sqrt{\left(\frac{4-3p}{2-p}\right)^2-\left(\frac{4\sqrt{p(1-p)}}{2-p}\right)^2}\right)\\
		&=\frac{1}{4}(4-3p-|4-5p|).
	\end{aligned}
\end{equation}
Similarly for $(X_1,X_2)$ and $(X_3,X_1)$, we have
\begin{equation}
	\begin{aligned}
		\epsilon_1^2+\epsilon_2^2&\geq \mathcal{E}_{A}(X_1,X_2)=p,\\
		\epsilon_3^2+\epsilon_1^2&\geq \mathcal{E}_{A}(X_3,X_1)=\frac{1}{4}(4-3p-|4-5p|).
	\end{aligned}
\end{equation}

In Fig. \ref{fig:example-spin1}, we illustrate the comparisons among the three bounds: (i) $\mathcal{E}_0$ that can be computed with SDP; (ii) $\mathcal{E}_A$ that can be analytically obtained  and (iii) $\mathcal{E}_{Ozawa}$ that obtained from the Ozawa's relation. 
It can be seen that for the simultaneous measurement of two observables, we already have $\mathcal{E}_0>\mathcal{E}_A>\mathcal{E}_{Ozawa}$.
While for the simultaneous measurement of three observables, the SDP bound is tighter than the bound obtained by simply adding the results for each pairs of the observables.

\begin{figure*}[htb]
	\centering
	\includegraphics[width=\textwidth]{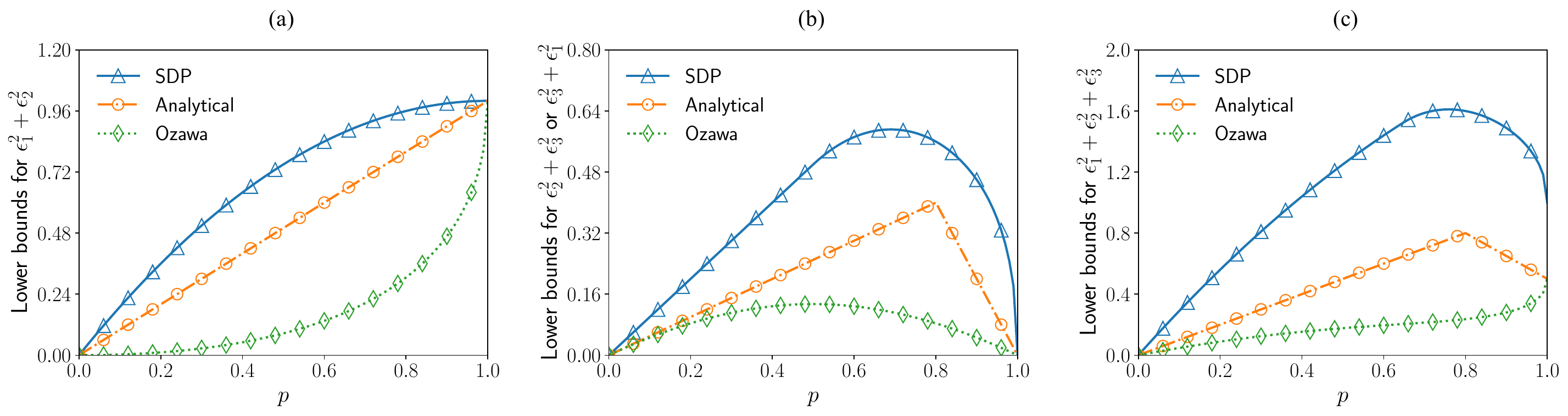}
	\caption{Lower bounds for the errors of the simultaneous measurement of (a) $(X_1,X_2)$, (b) $(X_2,X_3)$ or $(X_3,X_1)$, (c) $(X_1,X_2,X_3)$. The bounds for $(X_2,X_3)$ and $(X_3,X_1)$ are the same, so they are plotted in the same figure (b). The analytical and Ozawa's bounds for $(X_1,X_2,X_3)$ are obtained by summing the bounds for all pairs $(X_j,X_k)$, while the SDP bound can be directly obtained via the SDP for three observables.}
	\label{fig:example-spin1}
\end{figure*}

\section{Comparison of the analytical bound for more than two observables with the Ozawa's bound on pure states}\label{apdx:comparAnalyOzawa}
We have shown that in the case of two observables, $\mathcal{E}_A= \mathcal{E}_{Ozawa}$ for pure states and $\mathcal{E}_A\geq \mathcal{E}_{Ozawa}$ for mixed states. Thus for mixed states we already get tighter analytical bound than the Ozawa's bound. 
In this section we show that for pure states the analytical bound obtained for an arbitrary number of observables, 
\begin{equation}
    \Tr(S_{\text{Re}}^{-1}Q_{\text{Re}})\geq \left(\sqrt{\|S_{\text{Re}}^{-\frac{1}{2}}\tilde{S}_{\text{Im}}S_{\text{Re}}^{-\frac{1}{2}}\|_{F}+1}-1\right)^2,
\end{equation}
is always tighter than the simple summation of the Ozawa's bound when the number of observables is bigger than 4 ($n\geq 5$). 

Since $\Tr(S_{\text{Re}}^{-1}Q_{\text{Re}})$ does not change under linear transformations, without loss of generality, we can assume $S_{\text{Re}}=I$, the bound then becomes  
\begin{equation}\label{eq:suppSre}
    \sum_j \epsilon_j^2\geq \left(\sqrt{\|\tilde{S}_{\text{Im}}\|_{F}+1}-1\right)^2.
\end{equation}
For pure state $\rho=|\psi\rangle\langle\psi|$, with arbitrary $\{\ket{u_q}\}$ we have $(S_{u_q})_{jk}=\bra{\psi}X_jX_k\ket{\psi}|\inp{\psi}{u_q}|^2$,
and $(S_{u_q,\text{Im}})_{jk}=\frac{1}{2i}\bra{\psi}[X_j,X_k]\ket{\psi}|\inp{\psi}{u_q}|^2$.
The matrices $S_{u_q,\text{Im}}$ for different $q$ then equal to a common antisymmetric matrix multiplied by non-negative numbers $|\inp{\psi}{u_q}|^2$.
In this case the optimal choice is not to take any transpose.
With arbitrary $\{\ket{u_q}\}$ with $\sum_q \ket{u_q}\bra{u_q}=I$, we have $\sum_q |\inp{\psi}{u_q}|^2=1$, and 
\begin{equation}
    (\tilde{S}_{\text{Im}})_{jk}=\sum_q (S_{u_q,\text{Im}})_{jk}=\sum_q \frac{1}{2i}\bra{\psi}[X_j,X_k]\ket{\psi}|\inp{\psi}{u_q}|^2=\frac{1}{2i}\langle\psi|[X_j,X_k]|\psi\rangle
\end{equation}

From the Ozawa's relation, we show in Appendix \ref{apdx:tightertwo} that
\begin{equation}
    w_1\epsilon_j^2+w_2\epsilon_k^2\geq\mathcal{E}_{Ozawa}\equiv \frac{1}{2}\left(\alpha-\sqrt{\alpha^2-\beta^2}\right),
\end{equation}
where $\alpha=w_1(\Delta X_j)^2+w_2(\Delta X_k)^2$, $\beta=\sqrt{w_1w_2}\|\sqrt{\rho}[X_j,X_k]\sqrt{\rho}\|_1$. For the comparison we choose $w_1=w_2=1$ and for normalized observables with $S_{\text{Re}}=I$ we have $(\Delta X_j)^2=1$, $\forall j$. For pure state, $\rho=|\psi\rangle\langle\psi|$, the Ozawa's bound is then given by
\begin{equation}
    \epsilon_j^2+\epsilon_k^2\geq\mathcal{E}_{Ozawa}(X_j,X_k)\equiv \frac{1}{2}\left(2-\sqrt{4-|\langle\psi|[X_j,X_k]|\psi\rangle|^2}\right)=1-\sqrt{1-\frac{1}{2}|\langle\psi|[X_j,X_k]|\psi\rangle|^2},
\end{equation}
By summing all pairs, we obtain
\begin{equation}\label{eq:suppsum}
    \sum_l\epsilon_l^2 \geq \frac{1}{n-1}\sum_{j,k} \left(1-\sqrt{1-\frac{1}{2}|\langle\psi|[X_j,X_k]|\psi\rangle|^2}\right).
\end{equation}
Since $1-\frac{1}{2}|\langle\psi|[X_j,X_k]|\psi\rangle|^2\leq 1$, we have $\sqrt{1-\frac{1}{2}|\langle\psi|[X_j,X_k]|\psi\rangle|^2}\geq 1-\frac{1}{2}|\langle\psi|[X_j,X_k]|\psi\rangle|^2$ and 
\begin{equation}
    \sum_{j,k}\left(1-\sqrt{1-\frac{1}{2}|\langle\psi|[X_j,X_k]|\psi\rangle|^2}\right)\leq \sum_{j,k}\frac{1}{2}|\langle\psi|[X_j,X_k]|\psi\rangle|^2=\frac{1}{2}\|\tilde{S}_{\text{Im}}\|_{F}^2.  
\end{equation}
The bound in Eq.(\ref{eq:suppSre}) is then tighter than the bound in Eq.(\ref{eq:suppsum}) when 
\begin{equation}
    \left(\sqrt{\|\tilde{S}_{\text{Im}}\|_{F}+1}-1\right)^2 \geq \frac{1}{2(n-1)}\|\tilde{S}_{\text{Im}}\|_{F}^2\geq \frac{1}{n-1}\sum_{j,k} (1-\sqrt{1-\frac{1}{2}|\langle\psi|[X_j,X_k]|\psi\rangle|^2}).  
\end{equation}
The first inequality is equivalent to 
\begin{eqnarray}
    \aligned
    &\sqrt{\|\tilde{S}_{\text{Im}}\|_{F}+1}-1\geq \frac{1}{\sqrt{2(n-1)}}\|\tilde{S}_{\text{Im}}\|_F\\
    \Leftrightarrow \quad &\frac{1}{\sqrt{2(n-1)}}(\|\tilde{S}_{\text{Im}}\|_F+1)-\frac{1}{\sqrt{2(n-1)}}-\sqrt{\|\tilde{S}_{\text{Im}}\|_{F}+1}+1 \leq 0\\
    \Leftrightarrow \quad &(\|\tilde{S}_{\text{Im}}\|_F+1)-\sqrt{2(n-1)}\sqrt{\|\tilde{S}_{\text{Im}}\|_{F}+1}+\sqrt{2(n-1)}-1\leq 0\\
    \Leftrightarrow \quad &\left(\sqrt{\|\tilde{S}_{\text{Im}}\|_F+1}-\sqrt{\frac{n-1}{2}}\right)^2-\frac{n-1}{2}+\sqrt{2(n-1)}-1\leq 0\\
    \Leftrightarrow \quad & \left(\sqrt{\|\tilde{S}_{\text{Im}}\|_F+1}-\sqrt{\frac{n-1}{2}}\right)^2\leq \left(\sqrt{\frac{n-1}{2}}-1\right)^2 \\
    \Leftrightarrow \quad & \sqrt{\frac{n-1}{2}}-\left(\sqrt{\frac{n-1}{2}}-1\right) \leq\sqrt{\|\tilde{S}_{\text{Im}}\|_F+1}\leq \sqrt{\frac{n-1}{2}}+\left(\sqrt{\frac{n-1}{2}}-1\right)\\
    \Leftrightarrow \quad & 1\leq \|\tilde{S}_{\text{Im}}\|_F+1\leq (\sqrt{2(n-1)}-1)^2\\
    \Leftrightarrow \quad & 0\leq \|\tilde{S}_{\text{Im}}\|_F\leq 2(n-1)-2\sqrt{2(n-1)}, 
    \endaligned
\end{eqnarray}
this always holds when $n\geq 5$ since from $S_{\text{Re}}=I$ and $S_{\text{Re}}+i\tilde{S}_{\text{Im}}\geq 0$ we have  $\|\tilde{S}_{\text{Im}}\|_F\leq \|S_{\text{Re}}\|_F=\sqrt{n}$, and $\sqrt{n}\leq 2(n-1)-2\sqrt{2(n-1)}$ when $n\geq 5$. Thus for $n\geq 5$ observables on pure states, the bound in Eq.(\ref{eq:suppSre}) is always tighter than the simple summation of Ozawa's relation for all pairs of the observables. We note that for pure states, $\mathcal{E}_A=\mathcal{E}_{Ozawa}$, the bound in Eq.(\ref{eq:suppSre}) is thus also tighter than the simple summation of $\mathcal{E}_A$ for all pairs. 

\section{Comparison of the SDP bound for more than two observables with the Ozawa's bound on mixed states}\label{apdx:comparSDPOzawa}
\noindent
For general $n$ observables, $\{X_1,\cdots, X_n\}$, we can use the results on two observables to show that the SDP bound for $\sum_{i=1}^n \epsilon_i^2$ is always tighter than the bound obtained by simply summing up the Ozawa's relation for all $\epsilon_j^2+\epsilon_k^2$. 

First note that for any pair of observables, we have $\mathcal{E}_0(X_j, X_k)\geq\sum_q\lambda_q\mathcal{E}_{|\phi_q\rangle}(X_j,X_k)$, and $\sum_q\lambda_q\mathcal{E}_{|\phi_q\rangle}(X_j,X_k)$ is tighter than the bound obtained from the Ozawa's relation. Thus the SDP bound for two observables is already tighter than the bound obtained from the Ozawa's relation. As $\sum_{i=1}^n \epsilon_i^2=\frac{1}{n-1}\sum_{j\neq k}(\epsilon_j^2+\epsilon_k^2)$, we can use the bound on two observables to obtain a bound on $\sum_{i=1}^n \epsilon_i^2$ as $\sum_{i=1}^n \epsilon_i^2\geq \frac{1}{n-1}\sum_{j\neq k}\mathcal{E}_0(X_j,X_k)$. This is tighter than the bound by summing up the Ozawa's relation for all pairs since each $\mathcal{E}_0(X_j,X_k)$ is tighter. To show the SDP bound for $n$ observables is tighter than the bound obtained by summing up the Ozawa's relation, we just need to show  
\begin{equation}\label{eq:SDPOzawa}
	\mathcal{E}_0(X_1,\cdots, X_n)\geq \frac{1}{n-1}\sum_{j\neq k}\mathcal{E}_0(X_j,X_k).
\end{equation}
By taking $W=I_n$ with $I_n$ as $n\times n$ Identity matrix, we have
\begin{eqnarray}
	\aligned
	&\mathcal{E}_0(X_1,\cdots, X_n)\\
	=&\min_{\mathbb{S},\{R_j\}_{j=1}^n}\Tr[(I\otimes\rho)(\mathbb{S}-\mathbb{R}\mathbb{X}^{\dagger}-\mathbb{X}\mathbb{R}^{\dagger}+\mathbb{X}\mathbb{X}^{\dagger})\\
	&\quad|\mathbb{S}\geq \mathbb{R}\mathbb{R}^{\dagger},\mathbb{S}_{jk}=\mathbb{S}_{kj}=\mathbb{S}_{jk}^{\dagger}, R_j=R_j^{\dagger}]\\
	=&\Tr[(I\otimes\rho)(\mathbb{S}^*-\mathbb{R}^*\mathbb{X}^{\dagger}-\mathbb{X}\mathbb{R}^{*\dagger }+\mathbb{X}\mathbb{X}^{\dagger})]\\
	=&\frac{1}{n-1}\sum_{j,k}\Tr[(I^{(j,k)}\otimes\rho)(\mathbb{S}^*-\mathbb{R}^*\mathbb{X}^{\dagger}-\mathbb{X}\mathbb{R}^{* \dagger}+\mathbb{X}\mathbb{X}^{\dagger})],
	\endaligned
\end{eqnarray}
here $\mathbb{S}^*$ and $\mathbb{R}^*$ denote the optimal operators that achieve the minimum for $n$ observables, $I^{(j,k)}$ denote the diagonal matrix with its $j$th and $k$th diagonal element equal to 1 while others equal to 0. When $W=I^{(j,k)}$, only $2\times 2$ blocks with indexes $(j,k)$ contribute and we have
\begin{eqnarray}
	\aligned
	\Tr[&(I^{(j,k)}\otimes\rho)(\mathbb{S}^*-\mathbb{R}^*\mathbb{X}^{\dagger}-\mathbb{X}\mathbb{R}^{* \dagger }+\mathbb{X}\mathbb{X}^{\dagger})]\\
	=\Tr[&(I_2\otimes\rho)(\mathbb{S}^{*(j,k)}-\mathbb{R}^{*(j,k)}\mathbb{X}^{\dagger (j,k)}\\
	&-\mathbb{X}^{(j,k)}\mathbb{R}^{ *(j,k) \dagger}+\mathbb{X}^{(j,k)}\mathbb{X}^{\dagger(j,k)})]
	\endaligned
\end{eqnarray}
here $I_2$ is $2\times 2$ Identity matrix, $\mathbb{S}^{*(j,k)}=\begin{pmatrix}
	\mathbb{S}_{jj}^* & \mathbb{S}_{jk}^*\\
	\mathbb{S}_{kj}^* & \mathbb{S}_{kk}^*
\end{pmatrix}$, $\mathbb{R}^{*(j,k)}=\begin{pmatrix} R_j^* & R_k^* \end{pmatrix}^\dagger$, $\mathbb{X}^{(j,k)}=\begin{pmatrix} X_j & X_k \end{pmatrix}^\dagger$. Note that $\mathbb{S}_{jk}^*=\mathbb{S}_{kj}^*=\mathbb{S}_{jk}^{*\dagger}$ for each $(j,k)$, $R_j^*=R_j^{ *\dagger}$ for each $j$, and from $\mathbb{S}^*\geq \mathbb{R}^*\mathbb{R}^{* \dagger}$, we have
\begin{equation}
	\begin{pmatrix}
		\mathbb{S}_{jj}^* & \mathbb{S}_{jk}^*\\
		\mathbb{S}_{kj}^* & \mathbb{S}_{kk}^*
	\end{pmatrix}\geq
	\begin{pmatrix}
		R_{j}^*R_{j}^* & R_{j}^*R_{k}^*\\
		R_{k}^*R_{j}^* & R_{k}^*R_{k}^*
	\end{pmatrix},\ \text{for each}\ (j,k).
\end{equation} 
Thus 
$\Tr[(I_2\otimes\rho)(\mathbb{S}^{*(j,k)}-\mathbb{R}^{*(j,k)}\mathbb{X}^{\dagger (j,k)}
-\mathbb{X}^{(j,k)}\mathbb{R}^{ *(j,k) \dagger}+\mathbb{X}^{(j,k)}\mathbb{X}^{\dagger(j,k)})]
\geq \mathcal{E}_0(X_j,X_k)$ since $\mathcal{E}_0(X_j,X_k)= \min_{\mathbb{S},\mathbb{R}} \Tr[(I_2\otimes\rho)(\mathbb{S}-\mathbb{R}\mathbb{X}^{\dagger}-\mathbb{X}\mathbb{R}^{\dagger}+\mathbb{X}\mathbb{X}^{\dagger})$,
where $\mathbb{X}=\begin{pmatrix} X_j & X_k \end{pmatrix}^\dagger$ and the minimum is taken over all $\mathbb{S}=\begin{pmatrix}
	\mathbb{S}_{jj} & \mathbb{S}_{jk}\\
	\mathbb{S}_{kj} & \mathbb{S}_{kk}
\end{pmatrix}$ and $\mathbb{R}=\begin{pmatrix} R_j & R_k \end{pmatrix}^\dagger$ with $\mathbb{S}\geq \mathbb{R}\mathbb{R}^{\dagger},\mathbb{S}_{jk}=\mathbb{S}_{kj}=\mathbb{S}_{jk}^{\dagger}, R_j=R_j^{\dagger}$. Thus 
\begin{eqnarray}
	\aligned
	&\mathcal{E}_0(X_1,\cdots, X_n)\\
	=&\frac{1}{n-1}\sum_{j\neq k}\Tr[(I^{(j,k)}\otimes\rho)(\mathbb{S}^*-\mathbb{R}^*\mathbb{X}^{\dagger}-\mathbb{X}\mathbb{R}^{* \dagger}+\mathbb{X}\mathbb{X}^{\dagger})]\\
	\geq &\frac{1}{n-1}\sum_{j\neq k} \mathcal{E}_0(X_j,X_k).  
	\endaligned
\end{eqnarray}
Since $\mathcal{E}_0(X_j,X_k)$ is tighter than the Ozawa's relation, the SDP bound is thus always tighter than the bound obtained from the simple summation of the Ozawa's relation for all pairs.

\section{Examples}\label{apdx:example}

In this section, we demonstrate the applications of (i) the bounds of $\Tr(S_{\text{Re}}^{-1}Q_{\text{Re}})$; (ii) $\mathcal{E}_A$ that can be analytically obtained and (iii) $\mathcal{E}_0$ that can be computed with SDP in specific examples.

\subsection{Example: Error tradeoff relation for three observables on a qubit}

\noindent
\textit{Bounds obtained from $\Tr(S_{\text{Re}}^{-1}Q_{\text{Re}})$:}

Consider the approximate measurement of the three Pauli operators $\frac{1}{2}\sigma_x$, $\frac{1}{2}\sigma_y$, $\frac{1}{2}\sigma_z$ for a general spin-1/2 state $\rho=\frac{1}{2}(I+r_x\sigma_x+r_y\sigma_y+r_z\sigma_z)$, where $|r_x|^2+|r_y|^2+|r_z|^2\leq 1$.
In this case $X_1=\frac{1}{2}\sigma_x$, $X_2=\frac{1}{2}\sigma_y$, $X_3=\frac{1}{2}\sigma_z$, we have  $(S_{\text{Re}})_{jk}=\frac{1}{2}\Tr(\rho\{X_j,X_k\})=\frac{1}{2}\Tr(\rho\cdot\frac{1}{2}\delta_{jk} I)=\frac{1}{4}\delta_{jk}$ for $1\leq j,k\leq 3$.
Thus $S_{\text{Re}}=\frac{1}{4} I$, the error tradeoff relation is then given by
\begin{equation}\label{eq:example1_Sim}
    \sum_{j=1}^{3}\epsilon_j^2=\Tr(Q_{\text{Re}})\geq \frac{1}{4}\left(\sqrt{4\|\tilde{S}_{\text{Im}}\|_{F}+1}-1\right)^2,
\end{equation}
here $\tilde{S}_{\text{Im}}$ is the imaginary part of $\tilde{S}=\sum_q\tilde{S}_{u_q}$ with each $\tilde{S}_{u_q}$ equals to either $S_{u_q}$ or $\bar{S}_{u_q}$, where $(S_{u_q})_{jk}=\bra{u_q}(\sqrt{\rho}\otimes \sqrt{\sigma})(X_jX_k\otimes I)(\sqrt{\rho}\otimes \sqrt{\sigma})\ket{u_q}$ with $\sigma=|\xi_0\rangle\langle \xi_0|$. 
The tightest bound is obtained by maximizing $\|\tilde{S}_{\text{Im}}\|_{F}$ over all choices of $\{\ket{u_q}\}$ such that $\sum_q\ket{u_q}\bra{u_q}= I$. Here instead of maximizing over all $\{|u_q\rangle\}$, we consider the maximization over $\{\ket{u_1}\otimes\ket{\xi_l}, \ket{u_2}\otimes\ket{\xi_l}|l=0,1,\cdots, d_A-1\}$, with $\{\ket{u_1},\ket{u_2}\}$ and $\{\ket{\xi_l}|l=0,1,\cdots, d_A-1\}$ being orthonormal basis of the system and ancilla, respectively.  
With such choices we have
\begin{eqnarray}
    \aligned
    (S_{u_{ql}})_{jk}&=\bra{\xi_l}\bra{u_{q}}(\sqrt{\rho}\otimes \sqrt{\sigma})(X_jX_k\otimes I)(\sqrt{\rho}\otimes \sqrt{\sigma})\ket{u_{q}}\ket{\xi_l}\\
    &=\bra{u_q}\sqrt{\rho}X_jX_k\sqrt{\rho}\ket{u_q}\delta_{0l}.
    \endaligned
\end{eqnarray}
The optimization is now over $\{|u_1\rangle, |u_2\rangle\}$ such that $|u_1\rangle\langle u_1|+|u_2\rangle\langle u_2|=I_S$ with $I_S$ as the Identity on the system, together with the choices of taking the transpose on $(S_{u_q})_{jk}=\bra{u_q}\sqrt{\rho}X_jX_k\sqrt{\rho}\ket{u_q}$, $q=1,2$. The optimization leads to a lower bound on the minimal error of approximating $\{X_j\}$ with any measurement on $\rho$. The bound could be further tightened by optimizing over all choices of $\{|u_q\rangle\}$ in the space of system and ancilla.

Before making the optimization, we emphasize that any choice leads to a valid bound. A simple choice is just to choose $\{\ket{u_1},\ket{u_2}\}$ as $\{\ket{0},\ket{1}\}$, and we are free to make the choices of taking the transpose on $\{S_{u_0},S_{u_1}\}$. A direct calculation on $\{S_{u_0},S_{u_1}\}$ gives 
\begin{equation}
    \max \|\tilde{S}_{\text{Im}}\|_{F}^2=\max\left\{
    \|S_{u_0,\text{Im}}+S_{u_1,\text{Im}}\|_F^2, \|S_{u_0,\text{Im}}+S_{u_1,\text{Im}}^{T}\|_F^2
    \right\}=\max\left\{\frac{1}{8}(r_x^2+r_y^2+r_z^2),\frac{1}{8}(1-r_x^2-r_y^2)\right\},
\end{equation}
which leads to a tradeoff relation
\begin{equation}\label{eq:simplechoice}
    \sum_{j=1}^{3}\epsilon_j^2\geq \frac{1}{4}\left(\sqrt{4\max \|\tilde{S}_{\text{Im}}\|_{F}+1}-1\right)^2=\frac{1}{4}\left(\sqrt{4\times\frac{1}{2\sqrt{2}}\max\left\{\sqrt{r_x^2+r_y^2+r_z^2},\sqrt{1-r_x^2-r_y^2}\right\}+1}-1\right)^2.
\end{equation}
    

We now consider the maximization of $\|\tilde{S}_{\text{Im}}\|_{F}$ over the choices of 
$\{\ket{u_1},\ket{u_2}\}$. 
In this case we have
\begin{equation}\label{eq:optqubit}
    \max_{\{\ket{u_q}\}} \|\tilde{S}_{\text{Im}}\|_F^2=\max\left\{\sum_{jk}|w_{jk}|^2,\lambda_{\max}(S^{xx})\right\},
\end{equation}
where $w_{jk}=\Tr(\mathcal{X}_{jk})$ with $\mathcal{X}_{jk}=\frac{1}{2i}\sqrt{\rho}[X_j,X_k]\sqrt{\rho}$, and $S^{xx}=\sum_{jk}x_{jk}x_{jk}^{T}$ with $x_{jk}$ as a 3-dimensional vector with its $r$th entry given by $x_{jk}^{(r)}=\Tr(\mathcal{X}_{jk}\sigma_r)$, $1\leq r\leq 3$,
$\lambda_{\max}(S^{xx})$ denotes the largest eigenvalue of $S^{xx}$.
We will show that for any $\rho=\frac{1}{2}(I+r_x\sigma_x+r_y\sigma_y+r_z\sigma_z)$, we have $\sum_{jk}|w_{jk}|^2=\frac{1}{8}\left(r_x^2+r_y^2+r_z^2\right)\leq \frac{1}{8}$ and $\lambda_{\max}(S^{xx})=\frac{1}{8}$. This then leads to a bound as
\begin{equation}
    \sum_{j=1}^{3}\epsilon_j^2\geq \frac{1}{4}\left(\sqrt{4\max_{\{\ket{u_q}\}}\|\tilde{S}_{\text{Im}}\|_{F}+1}-1\right)^2=\frac{1}{4}\left(\sqrt{4\times\frac{1}{2\sqrt{2}}+1}-1\right)^2=\frac{1}{4}\left(\sqrt{\sqrt{2}+1}-1\right)^2,
\end{equation}
which is generally tighter than Eq.(\ref{eq:simplechoice}). 

We now sketch the procedure of the optimization to get the above bound. First note that the entries of $\tilde{S}_{\text{Im}}$ are given by
\begin{equation}
    (\tilde{S}_{\text{Im}})_{jk}=\sum_{q=1}^2 (\tilde{S}_{u_q,\text{Im}})_{jk}=\sum_{q=1}^2 (-1)^{s_q}\frac{1}{2i}\bra{u_q}\sqrt{\rho}[X_j, X_k]\sqrt{\rho}\ket{u_q},
\end{equation}
where $s_q=0$ if $\tilde{S}_{u_q}=S_{u_q}$, $s_q=1$ if $\tilde{S}_{u_q}=\bar{S}_{u_q}$.
The Frobenius norm of $\tilde{S}_{\text{Im}}$ can then be written as
\begin{equation}
    \begin{aligned}
        \|\tilde{S}_{\text{Im}}\|_F^2&=\sum_{jk}|(\tilde{S}_{\text{Im}})_{jk}|^2=\sum_{jk}\left|\sum_q (-1)^{s_q}\frac{1}{2i}\bra{u_q}\sqrt{\rho}[X_j, X_k]\sqrt{\rho}\ket{u_q}\right|^2\\
        &=\sum_{jk}\left|\Tr\left(\left(\sum_q (-1)^{s_q}M_q\right)\frac{1}{2i}\sqrt{\rho}[X_j, X_k]\sqrt{\rho}\right)\right|^2\\
        &=\sum_{jk}\left|\Tr\left(\left(\sum_q (-1)^{s_q}M_q\right)\mathcal{X}_{jk}\right)\right|^2,
    \end{aligned}
\end{equation}
here $M_q=\ket{u_q}\bra{u_q}$ and $\mathcal{X}_{jk}=\frac{1}{2i}\sqrt{\rho}[X_j, X_k]\sqrt{\rho}$.
The optimization can thus be reformulated as
\begin{equation}
    \begin{aligned}
        \max_{\{M_q\}}\ &\sum_{jk}\left|\Tr\left(\mathcal{X}_{jk}\sum_{q=1}^2 (-1)^{s_q}M_q\right)\right|^2\\
        \text{subject to}\ &\sum_q M_q= I,\ M_q\geq 0,\ M_q^2=M_q,
    \end{aligned}
\end{equation}
where $s_q\in\{0,1\},\forall q\in \{1,2\}$.
For the qubit system, we introduce the Hermitian basis $\{\lambda_r\}_{r=0}^3=\{\frac{1}{\sqrt{2}} I,\frac{1}{\sqrt{2}}\sigma_1,\frac{1}{\sqrt{2}}\sigma_2,\frac{1}{\sqrt{2}}\sigma_3\}$ such that $\Tr(\lambda_r\lambda_s)=\delta_{rs}$ for $0\leq r,s\leq 3$. The Hermitian operators $\mathcal{X}_{jk}$ and $M_q$ can then be vectorized as
\begin{equation}
    \mathcal{X}_{jk}=\sum_{r=0}^3 \Tr(\mathcal{X}_{jk}\lambda_r)\lambda_r=\frac{1}{2}\left[\Tr(\mathcal{X}_{jk}) I+\sum_{r=1}^3\Tr(\mathcal{X}_{jk}\sigma_r)\sigma_r\right]=\frac{1}{2}\left(w_{jk} I+\sum_{r=1}^3x_{jk}^{(r)}\sigma_r\right),
\end{equation}
\begin{equation}
    M_q=\sum_{r=0}^3 \Tr(M_q\lambda_r)\lambda_r=\frac{1}{2}\left[\Tr(M_q) I+\sum_{r=1}^3\Tr(M_q\sigma_r)\sigma_r\right]=\frac{1}{2}\left(w_{q} I+\sum_{r=1}^3x_{q}^{(r)}\sigma_r\right),
\end{equation}
where we have defined $w_{jk}=\Tr(\mathcal{X}_{jk})$, $x_{jk}^{(r)}=\Tr(\mathcal{X}_{jk}\sigma_r)$, $w_q=\Tr(M_q)=1$, $x_q^{(r)}=\Tr(M_q\sigma_r)$. Let $x_q=(x_q^{(1)},x_q^{(2)},x_q^{(3)})$ as a three-dimensional vectors, from $M_q\geq 0$, we can obtain $\sum_{r=1}^3\left(x_q^{(r)}\right)^2\le 1$, i.e., $\|x_q\|\le 1$, and from $M_q^2=M_q$, we have $\|x_q\|=1$. 
The optimization can thus be reformulated as
\begin{equation}
    \begin{aligned}
        \max_{\{x_q\}}\ &\frac{1}{4}\sum_{jk}\left|w_{jk}\left[\sum_q (-1)^{s_q}\right]+\sum_{r=1}^3x_{jk}^{(r)}\left[\sum_q(-1)^{s_q}x_q^{(r)}\right]\right|^2\\
        \text{subject to}\ &\sum_qx_q=\mathbf{0},\ \|x_q\|=1,
    \end{aligned}
\end{equation}
where $s_q\in\{0,1\},\forall q\in \{1,2\}$. 

When $s_1=s_2=0$ or $s_1=s_2=1$,
\begin{equation}
    \max_{\{\ket{u_q}\}} \|\tilde{S}_{\text{Im}}\|_F^2=\max_{\{x_q\}}\ \frac{1}{4}\sum_{jk}\left|2w_{jk}+\sum_{r=1}^3x_{jk}^{(r)}\left(x_1^{(r)}+x_2^{(r)}\right)\right|^2=\sum_{jk}|w_{jk}|^2.
\end{equation}

When $s_1=0$, $s_2=1$ or $s_1=1$, $s_2=0$,
\begin{equation}\
    \begin{aligned}
        \max_{\{\ket{u_q}\}} \|\tilde{S}_{\text{Im}}\|_F^2 &= \max_{\{x_q\}}\ \frac{1}{4}\sum_{jk}\left|\sum_{r=1}^3x_{jk}^{(r)}\left(x_1^{(r)}-x_2^{(r)}\right)\right|^2\\
        &=\max_{\{x_1\}}\ \frac{1}{4}\sum_{jk}\left|2\sum_{r=1}^3x_{jk}^{(r)}x_1^{(r)}\right|^2=\max_{\{x_1\}}\ \sum_{jk}\sum_{rs}x_{jk}^{(r)}x_1^{(r)}x_{jk}^{(s)}x_1^{(s)}\\
        &=\max_{\{x_1\}}\ \sum_{rs}x_1^{(r)}\left(\sum_{jk}x_{jk}^{(r)}x_{jk}^{(s)}\right)x_1^{(s)}\\
        &=\max_{\{x_1\}}\ x_1^{T} S^{xx}x_1\\
        &=\lambda_{\max}(S^{xx}).
    \end{aligned}
\end{equation}
where we have defined a $3\times 3$ symmetric matrix $S^{xx}$ with its $rs$-th entry as $(S^{xx})_{rs}=\sum_{jk}x_{jk}^{(r)}x_{jk}^{(s)}$, i.e., $S^{xx}=\sum_{jk}x_{jk}x_{jk}^{T}$. $\lambda_{\max}(S^{xx})$ then denotes the largest eigenvalue of the matrix $S^{xx}$.
The optimal $x_1$ is given by the eigenvector of $S^{xx}$ which corresponds to the maximal eigenvalue, which further gives the optimal choice of $\{\ket{u_q}\}$.

Combining the two cases together, we have
\begin{equation}
    \max_{\{\ket{u_q}\}} \|\tilde{S}_{\text{Im}}\|_F^2=\max\left\{\sum_{jk}|w_{jk}|^2,\lambda_{\max}(S^{xx})\right\}.
\end{equation}

By substituting $\rho$ and $X_1=\frac{1}{2}\sigma_x$, $X_2=\frac{1}{2}\sigma_y$, $X_3=\frac{1}{2}\sigma_z$ in the definition of $w_{jk}$, we have
\begin{equation}
    \sum_{jk}|w_{jk}|^2=\sum_{jk} \left|\frac{1}{2i}\Tr(\rho[X_j,X_k])\right|^2=\frac{1}{8}\left(r_x^2+r_y^2+r_z^2\right).
\end{equation}

To calculate $\lambda_{\max}(S^{xx})$, note that for  $\rho=\frac{1}{2}(I+r_x\sigma_x+r_y\sigma_y+r_z\sigma_z)$, we can equivalently write $r_x=\lambda\sin\theta\cos\varphi$, $r_y=\lambda\sin\theta\sin\varphi$, and $r_z=\lambda\cos\theta$ with $|\lambda|\leq 1$.
$\rho$ then has an eigendecomposition as $\rho=\frac{1+\lambda}{2}\ket{\psi_1}\bra{\psi_1}+\frac{1-\lambda}{2}\ket{\psi_2}\bra{\psi_2}$ with $\ket{\psi_1}=\begin{pmatrix}
    e^{-i\varphi}\cos\frac{\theta}{2}\\
    \sin\frac{\theta}{2}
\end{pmatrix}$, $\ket{\psi_2}=\begin{pmatrix}
    e^{-i\varphi}\sin\frac{\theta}{2}\\
    -\cos\frac{\theta}{2}
\end{pmatrix}$ and $\inp{\psi_1}{\psi_2}=0$.
Thus
\begin{equation}\label{eq:rhosqrt}
    \begin{aligned}
        \sqrt{\rho}&=\sqrt{\frac{1+\lambda}{2}}\ket{\psi_1}\bra{\psi_1}+\sqrt{\frac{1-\lambda}{2}}\ket{\psi_2}\bra{\psi_2}
        =\frac{1}{2}\begin{pmatrix}
            p+q\cos\theta & qe^{-i\varphi}\sin\theta\\
            qe^{i\varphi}\sin\theta & p-q\cos\theta
        \end{pmatrix},
    \end{aligned}
\end{equation}
where we have defined $p=\sqrt{\frac{1+\lambda}{2}}+\sqrt{\frac{1-\lambda}{2}}$, $q=\sqrt{\frac{1+\lambda}{2}}-\sqrt{\frac{1-\lambda}{2}}$.
Then we have $x_{11}=x_{22}=x_{33}=\mathbf{0}$, and
\begin{equation}
    \begin{aligned}
        x_{12}&=-x_{21}=\frac{1}{8}\begin{pmatrix}
            q^2\cos\varphi\sin 2\theta\\
            q^2\sin\varphi\sin 2\theta\\
            p^2+q^2\cos 2\theta
        \end{pmatrix},\quad
        x_{13}=-x_{31}=-\frac{1}{8}\begin{pmatrix}
            q^2\sin 2\varphi\sin^2\theta\\
            p^2-q^2(\cos^2\theta+\cos 2\varphi\sin^2\theta)\\
            q^2\sin\varphi\sin 2\theta
        \end{pmatrix}\\
        x_{23}&=-x_{32}=\frac{1}{8}\begin{pmatrix}
            p^2-q^2(\cos^2\theta-\cos 2\varphi\sin^2\theta)\\
            q^2\sin 2\varphi\sin^2\theta\\
            q^2\cos\varphi\sin 2\theta
        \end{pmatrix}.
    \end{aligned}
\end{equation}
A few lines of calculation then gives
\begin{equation}
    S^{xx}=\sum_{j,k=1}^3 x_{jk}x_{jk}^{T}=\frac{1}{8}
    \begin{pmatrix}
        1-r_y^2-r_z^2 & r_xr_y & r_zr_x\\
        r_xr_y & 1-r_z^2-r_x^2 & r_yr_z\\
        r_zr_x & r_yr_z & 1-r_x^2-r_y^2
    \end{pmatrix},
\end{equation}
where we have represented $\lambda$, $\theta$, $\varphi$ with the original notations $r_x$, $r_y$, $r_z$.
The largest eigenvalue of $S^{xx}$ is $\frac{1}{8}$, with corresponding eigenvectors $x_1=\frac{1}{\sqrt{r_x^2+r_y^2+r_z^2}}(r_x\ r_y\ r_z)^{T}$ for $r_x^2+r_y^2+r_z^2\neq 0$. In this case the optimal choice of $\ket{u_1}$ is the pure state with $x_1$ as its bloch vector, and the optimal $\ket{u_2}$ is the orthogonal complement of $\ket{u_1}$.

Thus for any $\rho=\frac{1}{2}(I+r_x\sigma_x+r_y\sigma_y+r_z\sigma_z)$, we have $\sum_{jk}|w_{jk}|^2=\frac{1}{8}\left(r_x^2+r_y^2+r_z^2\right)\leq \frac{1}{8}$ and $\lambda_{\max}(S^{xx})=\frac{1}{8}$, Eq.(\ref{eq:example1_Sim}) and Eq.(\ref{eq:optqubit}) then lead to 
\begin{equation}
    \sum_{j=1}^{3}\epsilon_j^2\geq \frac{1}{4}\left(\sqrt{4\max_{\{\ket{u_q}\}}\|\tilde{S}_{\text{Im}}\|_{F}+1}-1\right)^2=\frac{1}{4}\left(\sqrt{4\times\frac{1}{2\sqrt{2}}+1}-1\right)^2=\frac{1}{4}\left(\sqrt{\sqrt{2}+1}-1\right)^2.
\end{equation}
This bounds the minimal error of approximately implementing $\{\frac{1}{2}\sigma_x, \frac{1}{2}\sigma_y,\frac{1}{2}\sigma_z\}$ with a single measurement on $\rho$.

\noindent
\textit{Bounds obtained from $\mathcal{E}_A$ and the SDP $\mathcal{E}_0$:}

Next we calculate the analytical bound $\mathcal{E}_A$ for each pairs of the observables.
For the ease of comparison, here we choose specific values of $\theta=\frac{\pi}{2}$, $\phi=0$, such that $\rho=\frac{1}{2}(I+\lambda\sigma_x)$.
For the simultaneous measurement of $(X_1,X_2)$, it is straightforward to obtain that
\begin{equation}
    \sqrt{\rho}[X_1,X_2]\sqrt{\rho}=\frac{i}{4}\sqrt{1-\lambda^2}\begin{pmatrix}1 & 0\\0 & -1\end{pmatrix}.
\end{equation}
By choosing $\{\ket{u_q}\}$ as the eigenvectors of $\sqrt{\rho}[X_1,X_2]\sqrt{\rho}$, i.e., $\{\ket{u_q}\}=\{\ket{0},\ket{1}\}$, the bound $\mathcal{E}_A$ can be directly calculated as
\begin{equation}
    \epsilon_1^2+\epsilon_2^2\geq \mathcal{E}_A(X_1,X_2)= \frac{1}{4}(1-\lambda^2).
\end{equation}
Similarly for the simultaneous measurement of $(X_2,X_3)$ and $(X_3,X_1)$, the bound $\mathcal{E}_A$ can be calculated as
\begin{equation}
    \begin{aligned}
        \epsilon_2^2+\epsilon_3^2&\geq \mathcal{E}_A(X_2,X_3)= \frac{1}{4},\\
        \epsilon_3^2+\epsilon_1^2&\geq \mathcal{E}_A(X_3,X_1)= \frac{1}{4}(1-\lambda^2).
    \end{aligned}
\end{equation}
For the simultaneous measurement of $(X_1,X_2,X_3)$, a direct summation of the bounds for each pairs of the observables gives
\begin{equation}
    \epsilon_1^2+\epsilon_2^2+\epsilon_3^2\geq\frac{1}{2}(\mathcal{E}_A(X_1,X_2)+\mathcal{E}_A(X_2,X_3)+\mathcal{E}_A(X_3,X_1))=\frac{1}{8}+\frac{1}{4}(1-\lambda^2).
\end{equation}

The SDP $\mathcal{E}_0$ also gives a bound for $(X_1,X_2,X_3)$ directly, which can be computed efficiently.
In Fig. \ref{fig:example1-qubit}, we make comparisons among the three bounds, (i) bound that obtained from $\Tr(S_{\text{Re}}^{-1}Q_{\text{Re}})$; (ii) the simple summation of the analytical bound $\mathcal{E}_A$ for all pairs and (iii) $\mathcal{E}_0$ that can be computed with SDP. It can be seen that the SDP bound is the tightest among the three bounds.
We note that although the simple summation of $\mathcal{E}_A$ for each pairs of the observables is tighter than $\Tr(S_{\text{Re}}^{-1}Q_{\text{Re}})$ in this example, there are cases the bound obtained from $\Tr(S_{\text{Re}}^{-1}Q_{\text{Re}})$ can be tighter, as we have shown in the previous section.
\begin{figure}
    \centering
    \includegraphics[width=0.4\textwidth]{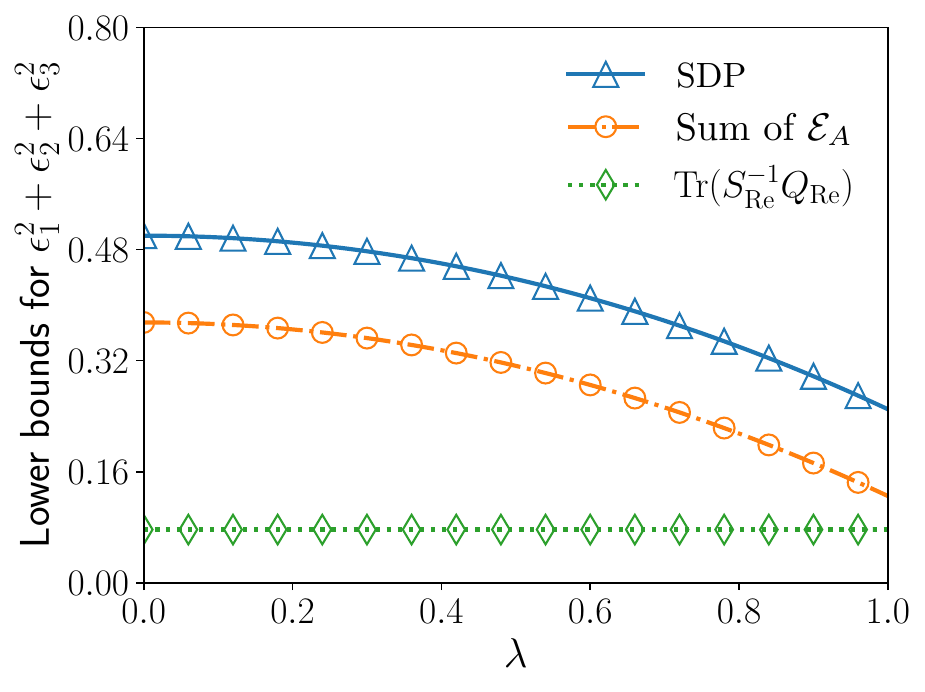}
    \caption{Lower bounds for the errors of the simultaneous measurement of $(X_1,X_2,X_3)$ given by (i) bound that obtained from $\Tr(S_{\text{Re}}^{-1}Q_{\text{Re}})$; (ii) the simple summation of the analytical bound $\mathcal{E}_A$ for all pairs and (iii) $\mathcal{E}_0$ that can be computed with SDP.}
    \label{fig:example1-qubit}
\end{figure}

\subsection{Example: Tradeoff relation for the estimation of multiple parameters in a qubit}
    
\noindent
\textit{Bounds obtained from $\Tr(S_{\text{Re}}^{-1}Q_{\text{Re}})$:}

The error-tradeoff relation for multiple observables can be used to obtain tradeoff relation in multiparameter quantum estimation. Here we consider the simultaneous estimation of the three parameters $(\lambda, \theta, \varphi)$ in $\rho=\frac{1}{2} (I+r_x\sigma_x+r_y\sigma_y+r_z\sigma_z)$, where $r_x=\lambda\sin\theta\cos\varphi$, $r_y=\lambda\sin\theta\sin\varphi$, and $r_z=\lambda\cos\theta$, here $|\lambda|\leq 1$.

For any measurement on $\rho$, the classical Fisher information matrix, $F_C$, is bounded as 
\begin{equation}\label{apdx:precision_bound}
    \Tr[F_Q^{-1}F_C]\leq 3-\left(\sqrt{\|F_Q^{-\frac{1}{2}}\tilde{S}_{\text{Im}}F_Q^{-\frac{1}{2}}\|_F+1}-1\right)^2,
\end{equation}
The SLDs for $(\lambda, \theta, \varphi)$ are given by
\begin{equation}
    L_{\lambda}=\frac{1}{1-\lambda^2}(\vec{n}\cdot\vec{\sigma}-\lambda I),\quad
    L_{\theta}=\lambda(\partial_{\theta}\vec{n}\cdot\vec{\sigma}),\quad
    L_{\varphi}=\lambda(\partial_{\varphi}\vec{n}\cdot\vec{\sigma}),
\end{equation}
here $\vec{n}=(\sin\theta\cos\varphi, \sin\theta\sin\varphi, \cos\theta)$, $\vec{\sigma}=(\sigma_x, \sigma_y, \sigma_z)$,
and the quantum Fisher information matrix is
\begin{equation}
    F_Q=\begin{pmatrix}
        \frac{1}{1-\lambda^2} & 0 & 0\\
        0 & \lambda^2 & 0\\
        0 & 0 & \lambda^2\sin^2\theta
    \end{pmatrix}.
\end{equation}
We can use Eq.(\ref{eq:optqubit}) to maximize $\tilde{S}_{\text{Im}}$ over $\{\ket{u_q}\}$. Note that under the reparameterization which makes $F_Q= I$, we have $F_Q^{-\frac{1}{2}}\tilde{S}_{\text{Im}}F_Q^{-\frac{1}{2}}=\tilde{S}_{\text{Im}}$. The SLD operators under this reparametrization are given by $\tilde{L}_{k}=\sum_{j}(F_Q^{-\frac{1}{2}})_{jk}L_j$, i.e.,
\begin{equation}
    \tilde{L}_{\lambda}=\frac{1}{\sqrt{1-\lambda^2}}(\vec{n}\cdot\vec{\sigma}-\lambda I),\quad
    \tilde{L}_{\theta}=\partial_{\theta}\vec{n}\cdot\vec{\sigma},\quad
    \tilde{L}_{\varphi}=\frac{1}{|\sin\theta|}(\partial_{\varphi}\vec{n}\cdot\vec{\sigma}).
\end{equation}
From Eq.(\ref{eq:rhosqrt}) we can obtain $\sqrt{\rho}=\frac{1}{2}(p I+q\vec{n}\cdot{\sigma})$, thus we have
\begin{equation}
    \begin{aligned}
        \mathcal{X}_{12}=\frac{1}{2i}\sqrt{\rho}[\tilde{L}_{\lambda},\tilde{L}_{\theta}]\sqrt{\rho}&=\frac{1}{2\sin\theta}(\partial_{\varphi}\vec{n}\cdot\vec{\sigma}),\\
        \mathcal{X}_{23}=\frac{1}{2i}\sqrt{\rho}[\tilde{L}_{\theta},\tilde{L}_{\varphi}]\sqrt{\rho}&=\frac{|\sin\theta|}{2\sin\theta}\left(\lambda I+\vec{n}\cdot\vec{\sigma}\right),\\
        \mathcal{X}_{31}=\frac{1}{2i}\sqrt{\rho}[\tilde{L}_{\varphi},\tilde{L}_{\lambda}]\sqrt{\rho}&=\frac{|\sin\theta|}{2\sin\theta}\left(\partial_{\theta}\vec{n}\cdot\vec{\sigma}\right),
    \end{aligned}
\end{equation}
which gives 
\begin{equation}
    x_{12}=-x_{21}=\frac{1}{\sin\theta}\partial_{\varphi}\vec{n},\quad
    x_{23}=-x_{32}=\frac{|\sin\theta|}{\sin\theta}\vec{n},\quad
    x_{13}=-x_{31}=-\frac{|\sin\theta|}{\sin\theta}\partial_{\theta}\vec{n}.
\end{equation}
Note that $x_{11}=x_{22}=x_{33}=\mathbf{0}$, it is then easy to obtain $S^{xx}=\sum_{j,k=1}^3 x_{jk}x_{jk}^{T}=2 I$. In this case $\{\ket{u_q}\}$ can be chosen arbitrarily and we always have $\lambda_{\max}(S^{xx})=2$.
It is also straightforward to obtain $\sum_{jk}|w_{jk}|^2=\sum_{jk}|\Tr(\mathcal{X}_{jk})|^2=2\lambda^2\leq 2$, thus
\begin{equation}
    \max_{\{\ket{u_q}\}} \|F_Q^{-\frac{1}{2}}\tilde{S}_{\text{Im}}F_Q^{-\frac{1}{2}}\|_F=\sqrt{2}.
\end{equation}
This leads to a tradeoff relation for the estimation of $\{\lambda, \theta,\varphi\}$ as
\begin{equation}
    \Tr[F_Q^{-1}F_C]\leq 3-\left(\sqrt{\max_{\{\ket{u_q}\}} \|F_Q^{-\frac{1}{2}}\tilde{S}_{\text{Im}}F_Q^{-\frac{1}{2}}\|_F+1}-1\right)^2=3-\left(\sqrt{\sqrt{2}+1}-1\right)^2.
\end{equation}

\noindent
\textit{Bounds obtained from $\mathcal{E}_A$ and the SDP $\mathcal{E}_0$:}

By doing summation over the analytical bound $\mathcal{E}_A$ for each pairs of $(\tilde{L}_j,\tilde{L}_k)$, we directly have
\begin{equation}
    \Tr[F_Q^{-1}F_C]\leq n-\frac{1}{2(n-1)}\|C\|_F^2,
\end{equation}
with the $jk$-th entry of $C$ given as $(C)_{jk}=\sqrt{\mathcal{E}_A(\tilde{L}_j,\tilde{L}_k)}$.
Again for the ease of comparison, we choose specific values of $\theta=\frac{\pi}{2}$, $\phi=0$.
Substituting $\tilde{L}_{\lambda}$, $\tilde{L}_{\theta}$, $\tilde{L}_{\varphi}$ in $\mathcal{E}_A(\tilde{L}_j,\tilde{L}_k)$, we directly have $\mathcal{E}_A(\tilde{L}_{\lambda},\tilde{L}_{\theta})=\mathcal{E}_A(\tilde{L}_{\theta},\tilde{L}_{\varphi})=\mathcal{E}_A(\tilde{L}_{\varphi},\tilde{L}_{\lambda})=1$, and
\begin{equation}
    \Tr[F_Q^{-1}F_C]\leq 3-\frac{1}{4}\|C\|_F^2=\frac{3}{2}.
\end{equation}
By substituting $\tilde{L}_{\lambda}$, $\tilde{L}_{\theta}$, $\tilde{L}_{\varphi}$, the SDP $\mathcal{E}_0$ also gives a bound directly as $\Tr[F_Q^{-1}F_C]\leq n-\mathcal{E}_0=1$, which is the tightest among all bounds and coincides with the Gill-Massar bound for qubit states.

\subsection{Example: Tradeoff relation under collective measurement}

While the tradeoff relation is obtained under local measurements, it can also be used to obtain tradeoff relations for collective measurements on $p$ copies of the states by simply replacing $\rho$ with $\rho_{x}^{\otimes p}$.  
As an example we consider the simultaneous estimation of the three parameters $(\lambda,\theta,\varphi)$ in a state $\rho=\frac{1}{2} (I+r_x\sigma_x+r_y\sigma_y+r_z\sigma_z)$, where $r_x=\lambda\sin\theta\cos\varphi$, $r_y=\lambda\sin\theta\sin\varphi$, and $r_z=\lambda\cos\theta$, $|\lambda|\leq 1$. But different from previous section, here we also allow collective measurement on two-copies of the state, $\rho\otimes\rho$.

To obtain the tradeoff relation under the collective measurement on 2-copies of the state, we can simply treat $\rho\otimes\rho$ as a larger single state. Its SLDs are given by $L_{j2}=L_j\otimes I+ I\otimes L_j$ with $j=\lambda,\theta,\varphi$ and the
quantum Fisher information matrix is $F_{Q2}=2F_{Q}$ with $F_Q$ as the quantum Fisher information matrix of a single $\rho$.
    
\noindent
\textit{Bounds obtained from $\Tr(S_{\text{Re}}^{-1}Q_{\text{Re}})$:}

If we choose $\{\ket{u_q}\}$ as the computational basis with $\ket{u_0}=\ket{00}$, $\ket{u_1}=\ket{01}$, $\ket{u_2}=\ket{10}$, $\ket{u_3}=\ket{11}$, the corresponding $S_{u_q,\text{Im}2}$ are given by
\begin{equation}
    \begin{aligned}
        S_{u_0,\text{Im}2}&=
        \begin{pmatrix}
            0 & 0 & -\frac{\lambda(1+\lambda\cos\theta)\sin^2\theta}{2\sqrt{1-\lambda^2}}\\
            0 & 0 & -\frac{1}{2}\lambda^2\sin\theta(1+\lambda\cos\theta)(\cos\theta+\lambda)\\
            \frac{\lambda(1+\lambda\cos\theta)\sin^2\theta}{2\sqrt{1-\lambda^2}} & \frac{1}{2}\lambda^2\sin\theta(1+\lambda\cos\theta)(\cos\theta+\lambda) & 0
        \end{pmatrix}\\
        S_{u_1,\text{Im}2}&=S_{u_2,\text{Im}2}=
        \begin{pmatrix}
            0 & 0 & \frac{\lambda^2\cos\theta\sin^2\theta}{2\sqrt{1-\lambda^2}}\\
            0 & 0 & -\frac{1}{2}\lambda^3\sin^3\theta\\
            -\frac{\lambda^2\cos\theta\sin^2\theta}{2\sqrt{1-\lambda^2}} & \frac{1}{2}\lambda^3\sin^3\theta & 0
        \end{pmatrix}\\
        S_{u_3,\text{Im}2}&=
        \begin{pmatrix}
            0 & 0 & \frac{\lambda(1-\lambda\cos\theta)\sin^2\theta}{2\sqrt{1-\lambda^2}}\\
            0 & 0 & \frac{1}{2}\lambda^2\sin\theta(1-\lambda\cos\theta)(\cos\theta-\lambda)\\
            -\frac{\lambda(1-\lambda\cos\theta)\sin^2\theta}{2\sqrt{1-\lambda^2}} & -\frac{1}{2}\lambda^2\sin\theta(1-\lambda\cos\theta)(\cos\theta-\lambda) & 0
        \end{pmatrix}
    \end{aligned}
\end{equation}
The optimal $\|F_{Q2}^{-\frac{1}{2}}\tilde{S}_{\text{Im}2}F_{Q2}^{-\frac{1}{2}}\|_F$ is thus given by the maximum of the following values,
\begin{equation}
    \begin{aligned}
        \|F_{Q2}^{-\frac{1}{2}}(S_{u_0,\text{Im}2}+S_{u_1,\text{Im}2}+S_{u_2,\text{Im}2}+S_{u_3,\text{Im}2})F_{Q2}^{-\frac{1}{2}}\|_F&=\sqrt{2}|\lambda|,\\
        \|F_{Q2}^{-\frac{1}{2}}(S_{u_0,\text{Im}2}+S_{u_1,\text{Im}2}^{T}+S_{u_2,\text{Im}2}^{T}+S_{u_3,\text{Im}2})F_{Q2}^{-\frac{1}{2}}\|_F&=\sqrt{2}|\lambda\cos\theta|,\\
        \|F_{Q2}^{-\frac{1}{2}}(S_{u_0,\text{Im}2}+S_{u_1,\text{Im}2}+S_{u_2,\text{Im}2}+S_{u_3,\text{Im}2}^{T})F_{Q2}^{-\frac{1}{2}}\|_F&=\sqrt{\frac{1}{2}(1+\lambda^2)(1+\lambda^2\cos^2\theta)+\lambda^3\cos\theta\sin^2\theta},\\
        \|F_{Q2}^{-\frac{1}{2}}(S_{u_0,\text{Im}2}+S_{u_1,\text{Im}2}^{T}+S_{u_2,\text{Im}2}^{T}+S_{u_3,\text{Im}2}^{T})F_{Q2}^{-\frac{1}{2}}\|_F&=\sqrt{\frac{1}{2}(1+\lambda^2)(1+\lambda^2\cos^2\theta)-\lambda^3\cos\theta\sin^2\theta}.
    \end{aligned}
\end{equation}
By simply replacing $F_Q$ and $\tilde{S}_{\text{Im}}$ with $F_{Q2}$ and $\tilde{S}_{\text{Im}2}$ in Eq.(\ref{apdx:precision_bound}), we can get
\begin{equation}
    \begin{aligned}
        \frac{1}{2}\Tr[F_Q^{-1}F_{C2}]&\leq 3-(\sqrt{\|F_{Q2}^{-\frac{1}{2}}\tilde{S}_{\text{Im}2}F_{Q2}^{-\frac{1}{2}}\|_F+1}-1)^2\\
        &= 3-\left(\sqrt{\max\left\{\sqrt{2}|\lambda|, \sqrt{\frac{1}{2}(1+\lambda^2)(1+\lambda^2\cos^2\theta)+|\lambda^3\cos\theta\sin^2\theta|}\right\}+1}-1\right)^2,
    \end{aligned}
\end{equation}
here $F_{C2}$ denotes the classical Fisher information matrix under the collective measurement on $\rho\otimes \rho$, $\frac{1}{2}F_{C2}$ can be regarded as the average Fisher information on each $\rho$. This bounds the achievable precision limit under the collective measurements on 2 copies of the quantum state.
Specifically for $\theta=\frac{\pi}{2}$, we have
\begin{equation}
    \frac{1}{2}\Tr[F_Q^{-1}F_{C2}]\leq 3-\left(\sqrt{\max\left\{\sqrt{2}|\lambda|, \sqrt{\frac{1}{2}(1+\lambda^2)}\right\}+1}-1\right)^2.
\end{equation}

\noindent
\textit{Bounds obtained from $\mathcal{E}_A$ and the SDP $\mathcal{E}_0$:}

Under the reparameterization that $F_{Q2}=2F_Q=I$, by replacing $(X_j,X_k)$ with $(\tilde{L}_{j2},\tilde{L}_{k2})$, the bound $\mathcal{E}_A$ gives
\begin{equation}
    \epsilon_j^2+\epsilon_k^2=1-(F_{C2})_{jj}+1-(F_{C2})_{kk}\geq \mathcal{E}_A(\tilde{L}_{j2},\tilde{L}_{k2}).
\end{equation}
By doing summations over each pairs of $(j,k)$, we directly have
\begin{equation}
    \frac{1}{2}\Tr[F_Q^{-1}F_{C2}]\leq n-\frac{1}{2(n-1)}\|C_2\|_F^2,
\end{equation}
with the $jk$-th entry of $C_2$ given as $(C_2)_{jk}=\sqrt{\mathcal{E}_A(\tilde{L}_{j2},\tilde{L}_{k2})}$.
Again for the ease of comparison, we choose specific values of $\theta=\frac{\pi}{2}$, $\phi=0$.
Substituting $\tilde{L}_{\lambda 2}$, $\tilde{L}_{\theta 2}$, $\tilde{L}_{\varphi 2}$ in $\mathcal{E}_A(\tilde{L}_{j2},\tilde{L}_{k2})$, we can directly obtain a bound.
By substituting $\tilde{L}_{\lambda 2}$, $\tilde{L}_{\theta 2}$, $\tilde{L}_{\varphi 2}$, the SDP $\mathcal{E}_0$ also gives a bound directly as $\frac{1}{2}\Tr[F_Q^{-1}F_{C2}]\leq n-\mathcal{E}_0$.
Fig. \ref{fig:example2-twoqubit} then illustrate the comparison among the three bounds with $\lambda$ varies from 0 to 1.

\begin{figure}
    \centering
    \includegraphics[width=0.4\textwidth]{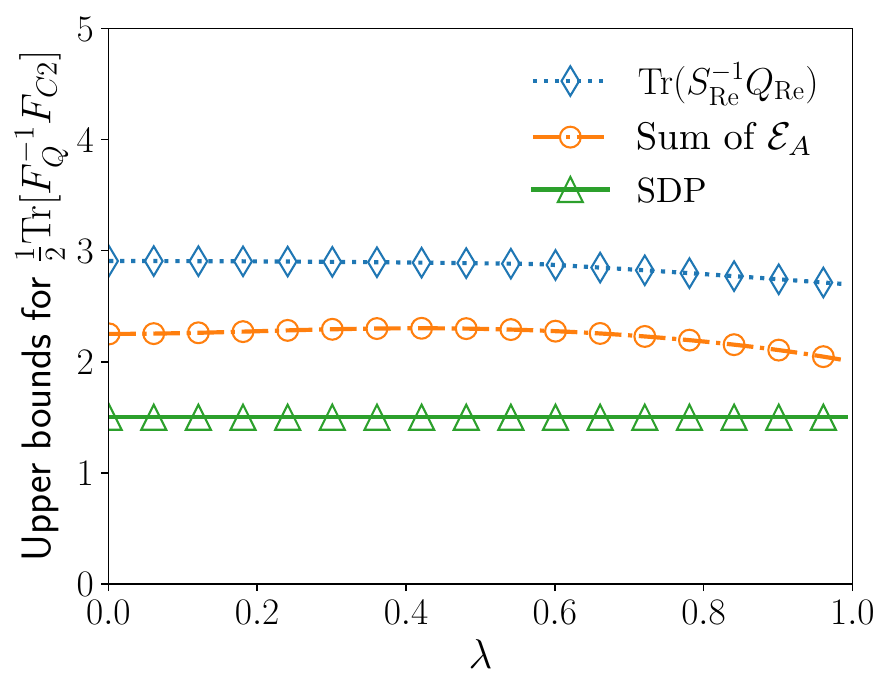}
    \caption{Upper bounds for $\frac{1}{2}\Tr[F_Q^{-1}F_{C2}]$ given by (i) bound that obtained from $\Tr(S_{\text{Re}}^{-1}Q_{\text{Re}})$; (ii) the simple summation of the analytical bound $\mathcal{E}_A$ for all pairs and (iii) $\mathcal{E}_0$ that can be computed with SDP.}
    \label{fig:example2-twoqubit}
\end{figure}

\subsection{Example: Tradeoff relation for the estimation of multiple parameters in three qubits}

We provide an example with three qubits to show that the analytical bound obtained from $\Tr(S_{\text{Re}}^{-1}Q_{\text{Re}})$ can be tighter than the simple summation over the bounds $\mathcal{E}_A$ for each pairs of observables. 

We consider a parameterized pure state $\ket{\psi}=\sin\theta_0\ket{\psi_1}+\cos\theta_0 e^{i\phi_0} \ket{\psi_2}$, where
\begin{equation}
    \begin{aligned}
        \ket{\psi_1}=\sin\theta_1\sin\theta_2\ket{001}+\sin\theta_1\cos\theta_2e^{i\phi_1}\ket{010}+\cos\theta_1e^{i\phi_2}\ket{100},\\
        \ket{\psi_2}=\sin\theta_3\sin\theta_4\ket{110}+\sin\theta_3\cos\theta_4e^{i\phi_3}\ket{101}+\cos\theta_3e^{i\phi_4}\ket{011}.
    \end{aligned}
\end{equation}

We first consider the estimation of $\{\theta_0,\theta_1,\theta_2,\phi_1,\phi_2\}$ in $\ket{\psi}=\sin\theta_0\ket{\psi_1}+\cos\theta_0 e^{i\phi_0} \ket{\psi_2}$. In this case $F_Q=\begin{pmatrix}
    F_{Q1} & \mathbf{0}\\
    \mathbf{0} & F_{Q2}
\end{pmatrix}$ with
\begin{equation}
    \begin{aligned}
        F_{Q1}&=\operatorname{diag}\{4,4\sin^2\theta_0,4\sin^2\theta_0\sin^2\theta_1\},\\
        F_{Q2}&=\begin{pmatrix}
            4\sin^2\theta_0\sin^2\theta_1\cos^2\theta_2(1-\sin^2\theta_0\sin^2\theta_1\cos^2\theta_2) & -\sin^4 \theta_0\sin^2 2\theta_1\cos^2\theta_2\\
            -\sin^4 \theta_0\sin^2 2\theta_1\cos^2\theta_2 & 4\sin^2\theta_0\cos^2\theta_1(1-\sin^2\theta_0\cos^2\theta_1)
        \end{pmatrix},
    \end{aligned}
\end{equation}
and $S_{\text{Im}}=\begin{pmatrix}
    \mathbf{0} & S_{\text{Im}1}^{\dagger}\\
    S_{\text{Im}1} & \mathbf{0}
\end{pmatrix}$ where
\begin{equation}
    S_{\text{Im}1}=2\begin{pmatrix}
        \sin 2\theta_0\sin^2\theta_1\cos^2\theta_2 & \sin^2\theta_0\sin 2\theta_1\cos^2\theta_2 & -\sin^2\theta_0\sin^2 \theta_1\sin 2\theta_2\\
        \sin 2\theta_0\cos^2\theta_1 & -\sin^2\theta_0\sin 2\theta_1 & 0 
    \end{pmatrix}.
\end{equation}
Thus
\begin{equation}
    F_Q^{-\frac{1}{2}}S_{\text{Im}}F_Q^{-\frac{1}{2}}=
    \begin{pmatrix}
        F_{Q1}^{-\frac{1}{2}} & \mathbf{0}\\
        \mathbf{0} & F_{Q2}^{-\frac{1}{2}}
    \end{pmatrix}
    \begin{pmatrix}
        \mathbf{0} & S_{\text{Im}1}^{\dagger}\\
        S_{\text{Im}1} & \mathbf{0}
    \end{pmatrix}
    \begin{pmatrix}
        F_{Q1}^{-\frac{1}{2}} & \mathbf{0}\\
        \mathbf{0} & F_{Q2}^{-\frac{1}{2}}
    \end{pmatrix}=
    \begin{pmatrix}
        \mathbf{0} & F_{Q1}^{-\frac{1}{2}}S_{\text{Im}1}^{\dagger}F_{Q2}^{-\frac{1}{2}}\\
        F_{Q2}^{-\frac{1}{2}} S_{\text{Im}1} F_{Q1}^{-\frac{1}{2}} & \mathbf{0}
    \end{pmatrix},
\end{equation}
and
\begin{equation}
    \|F_Q^{-\frac{1}{2}}S_{\text{Im}}F_Q^{-\frac{1}{2}}\|_F^2=2\|F_{Q2}^{-\frac{1}{2}} S_{\text{Im}1} F_{Q1}^{-\frac{1}{2}}\|_F^2=2\Tr\left(F_{Q2}^{-1}S_{\text{Im}1} F_{Q1}^{-1}S_{\text{Im}1}^{\dagger}\right)=4.
\end{equation}
We then have
\begin{equation}
    \Tr[F_Q^{-1}F_C]\leq n-\left(\sqrt{\|F_Q^{-\frac{1}{2}}\tilde{S}_{\text{Im}}F_Q^{-\frac{1}{2}}\|_F+1}-1\right)^2=5-(\sqrt{3}-1)^2\approx 4.464.
\end{equation}
Compared with the simple summation of the analytical bound for pairs of the observables, which cannot be tighter than $n-\frac{1}{2(n-1)}\|F_Q^{-\frac{1}{2}}S_{\text{Im}}F_Q^{-\frac{1}{2}}\|_F^2=5-\frac{1}{8}\times 4=\frac{9}{2}$, the bound obtained from $\Tr(S_{\text{Re}}^{-1}Q_{\text{Re}})$ is strictly tighter.
The tight bound can be obtained by computing the SDP $\mathcal{E}_0$.
Specifically, for $\theta_j=\pi/4$, $1\leq j\leq 4$ and $\phi_k=\pi/4$, $0\leq k\leq 4$, we compare the three bounds with $\theta_0$ varies from $0$ to $\pi/2$.
As illustrated in Fig. \ref{fig:example4-threequbit}(a), the analytical bound obtained from $\Tr(S_{\text{Re}}^{-1}Q_{\text{Re}})$ is tighter than the simple summation of $\mathcal{E}_A$ for all pairs, while the tight bound is given by $\mathcal{E}_0$ that can be computed with SDP, which is approximately $\Tr[F_Q^{-1}F_C]\leq 3$.

\begin{figure}
    \centering
    \includegraphics[width=0.8\textwidth]{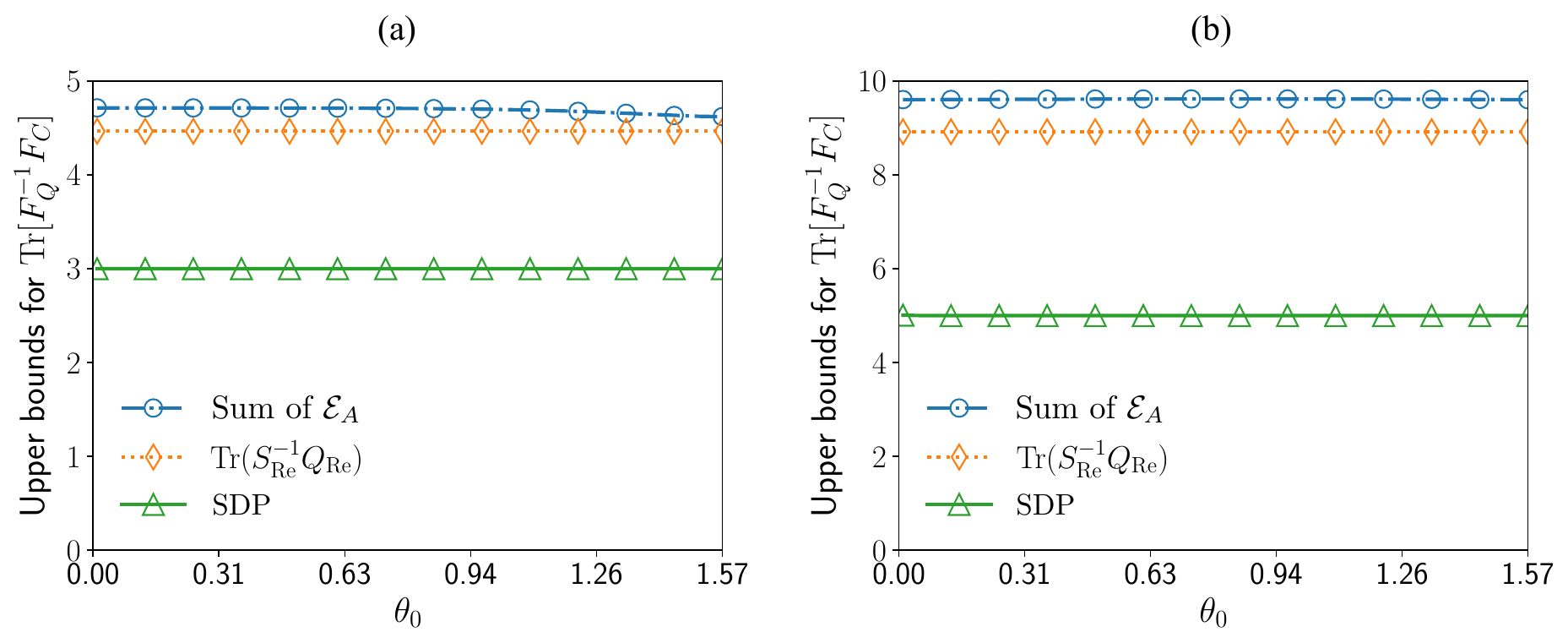}
    \caption{Upper bounds for $\Tr[F_Q^{-1}F_{C}]$ of (a) estimation of $n=5$ parameters $\{\theta_0,\theta_1,\theta_2,\phi_1,\phi_2\}$, (b) estimation of $n=10$ parameters $\{\theta_0,\theta_1,\theta_2,\theta_3,\theta_4,\phi_0,\phi_1,\phi_2,\phi_3,\phi_4\}$.
        Here the bounds are given by (i) bound that obtained from $\Tr(S_{\text{Re}}^{-1}Q_{\text{Re}})$; (ii) the simple summation of the analytical bound $\mathcal{E}_A$ for all pairs and (iii) $\mathcal{E}_0$ that can be computed with SDP.}
    \label{fig:example4-threequbit}
\end{figure}

For the estimation of the full set, $\{\theta_0,\theta_1,\theta_2,\theta_3,\theta_4,\phi_0,\phi_1,\phi_2,\phi_3,\phi_4\}$, we can obtain the bound numerically. 
Specifically, for $\theta_j=\pi/4$ and $\phi_k=\pi/4$, $0\leq j,k\leq 4$, we have $F_Q=\begin{pmatrix}
    F_{Q1} & \mathbf{0}\\
    \mathbf{0} & F_{Q2}
\end{pmatrix}$ where
\begin{equation}
    F_{Q1}=\operatorname{diag}\{4,2,1,2,1\},
    F_{Q2}=\begin{pmatrix}
        1 & -\frac{1}{4} & -\frac{1}{2} & \frac{1}{4} & \frac{1}{2}\\
        -\frac{1}{4} & \frac{7}{16} & -\frac{1}{8} & -\frac{1}{16} & -\frac{1}{8}\\
        -\frac{1}{2} & -\frac{1}{8} & \frac{3}{4} & -\frac{1}{8} & -\frac{1}{4}\\
        \frac{1}{4} & -\frac{1}{16} & -\frac{1}{8} & \frac{7}{16} & -\frac{1}{8}\\
        \frac{1}{2} & -\frac{1}{8} & -\frac{1}{4} & -\frac{1}{8} & \frac{3}{4}
    \end{pmatrix},
\end{equation}
and $S_{\text{Im}}=\begin{pmatrix}
    \mathbf{0} & S_{\text{Im}1}^{\dagger}\\
    S_{\text{Im}1} & \mathbf{0}
\end{pmatrix}$ with
\begin{equation}
    S_{\text{Im}1}=\begin{pmatrix}
        -2 & 0 & 0 & 0 & 0\\
        \frac{1}{2} & \frac{1}{2} & -\frac{1}{2} & 0 & 0\\
        1 & -1 & 0 & 0 & 0\\
        -\frac{1}{2} & 0 & 0 & \frac{1}{2} & -\frac{1}{2}\\
        -1 & 0 & 0 & -1 & 0
    \end{pmatrix}.
\end{equation}
We then have
\begin{equation}
    \|F_Q^{-\frac{1}{2}}S_{\text{Im}}F_Q^{-\frac{1}{2}}\|_F^2=2\|F_{Q2}^{-\frac{1}{2}} S_{\text{Im}1} F_{Q1}^{-\frac{1}{2}}\|_F^2=2\Tr\left(F_{Q2}^{-1}S_{\text{Im}1} F_{Q1}^{-1}S_{\text{Im}1}^{\dagger}\right)=10,
\end{equation}
and
\begin{equation}
    \Tr[F_Q^{-1}F_C]\leq n-\left(\sqrt{\|F_Q^{-\frac{1}{2}}\tilde{S}_{\text{Im}}F_Q^{-\frac{1}{2}}\|_F+1}-1\right)^2=10-\left(\sqrt{\sqrt{10}+1}-1\right)^2\approx 8.918.
\end{equation}
Compared with the simple summation of the analytical bound for pairs of the observables, which cannot be tighter than $n-\frac{1}{2(n-1)}\|F_Q^{-\frac{1}{2}}S_{\text{Im}}F_Q^{-\frac{1}{2}}\|_F^2=10-\frac{1}{18}\times 10=\frac{85}{9}\approx 9.444$, the bound obtained from $\Tr(S_{\text{Re}}^{-1}Q_{\text{Re}})$ is strictly tighter.
The tight bound can be obtained by computing the SDP $\mathcal{E}_0$.
For $\theta_j=\pi/4$, $1\leq j\leq 4$ and $\phi_k=\pi/4$, $0\leq k\leq 4$, we compare the three bounds with $\theta_0$ varies from $0$ to $\pi/2$.
As illustrated in Fig. \ref{fig:example4-threequbit}(b), the analytical bound obtained from $\Tr(S_{\text{Re}}^{-1}Q_{\text{Re}})$ is tighter than the simple summation of $\mathcal{E}_A$ for all pairs, while the tight bound is given by $\mathcal{E}_0$ that can be computed with SDP, which is approximately $\Tr[F_Q^{-1}F_C]\leq 5$.

\end{widetext}

\end{document}